\documentclass[12pt]{article}

\usepackage{amsmath}
\usepackage{amssymb}
\usepackage{amsthm}
\usepackage{latexsym}
\usepackage{authblk}
\usepackage{tikz,graphicx,subfigure}
\usetikzlibrary{arrows}
\usepgflibrary{decorations.markings}
\usetikzlibrary{decorations.markings}
\usepackage{overpic}
\usepackage{comment}

\usepackage[textsize=footnotesize]{todonotes}

\usepackage{epsfig,cancel}

\usepackage{color}
\usepackage{graphicx,framed}
\usepackage[colorlinks=true, pdfstartview=FitV, linkcolor=black, citecolor=black, urlcolor=blue]{hyperref}
\definecolor{shadecolor}{rgb}{0.95, 0.95, 0.86}

\usepackage{soul}
\usepackage{framed}

\usepackage{amssymb,cancel}
\usepackage{wrapfig}
\usepackage{mathrsfs}
\usepackage{amsbsy}

\usepackage{wrapfig}
\usepackage{color}
\usepackage{graphicx,framed}

\DeclareMathOperator{\Ai}{Ai}
\DeclareMathOperator{\area}{area}
\DeclareMathOperator{\diag}{diag}
\DeclareMathOperator{\dist}{dist}

\DeclareMathOperator{\Tr}{Tr}
\DeclareMathOperator{\Real}{Re}
\DeclareMathOperator{\Imag}{Im}
\renewcommand{\Re}{\Real}
\renewcommand{\Im}{\Imag}

\newtheorem{theorem}{Theorem}[section]
\newtheorem{lemma}[theorem]{Lemma}
\newtheorem{proposition}[theorem]{Proposition}
\newtheorem{corollary}[theorem]{Corollary}
\newtheorem{definition}[theorem]{Definition}

\newtheorem{Remark}[theorem]{Remark}
\newenvironment{remark}{\begin{Remark}\rm}{\end{Remark}}

\newtheorem{RHproblem}[theorem]{RH problem}
\newenvironment{rhproblem}{\begin{RHproblem}\rm}{\end{RHproblem}}

\def\bt{\begin{theorem}}
\def\et{\end{theorem}}
\def\bc{\begin{corollary}}
\def\ec{\end{corollary}}
\def\bx{\begin{example}\small}
\def\ex{\end{example}}
\def\bxr{\begin{exercise}\small}
\def\exr{\end{exercise}}
\def\bl{\begin{lemma}}
\def\el{\end{lemma}}
\def\bd{\begin{definition}}
\def\ed{\end{definition}}
\def\bp{\begin{proposition}}
\def\ep{\end{proposition}}

\def\br{\begin{remark}}
\def\er{\end{remark}}

\def\be{\begin{equation}}
\def\ee{\end{equation}}

\def \part{\partial}

\numberwithin{equation}{section}

\author[1]{A.B.J. Kuijlaars}\author[2]{A. Tovbis} 
\affil[1]{\normalsize Department of Mathematics, KU Leuven, Belgium, email:~arno.kuijlaars@wis.kuleuven.be}
\affil[2]{Department of Mathematics, University of Central Florida, 
	Orlando, FL 32816-1364, email:~alexander.tovbis@ucf.edu}
\title{The supercritical regime in the normal matrix model with cubic potential}
\date{\today}

\begin{document}
\maketitle

\begin{abstract} 
The normal matrix model with a cubic potential is ill-defined and it develops a critical
behavior in finite time. We follow the approach of Bleher and Kuijlaars to reformulate the model
in terms of orthogonal polynomials with respect to a Hermitian form. 
This reformulation was shown to capture the essential features
of the normal matrix model in the subcritical regime, namely that the zeros of the polynomials
tend to a number of segments (the motherbody) inside a domain (the droplet) that attracts
the eigenvalues in the normal matrix model. 

In the present paper we analyze the supercritical regime and we find that the large $n$ behavior 
is described by the evolution of a spectral curve satisfying the Boutroux condition.
The Boutroux condition determines a system of contours $\Sigma_1$, consisting of the motherbody and whiskers
sticking out of the domain. 
We find a second critical behavior at which the original motherbody shrinks to a point at the origin
and only the whiskers remain.

In the regime before the second criticality we also give strong asymptotics of the orthogonal polynomials
by means of a steepest descent analysis of a $3 \times 3$ matrix valued Riemann-Hilbert problem. 
It follows that the zeros of the orthogonal polynomials tend to $\Sigma_1$, with the exception of
at most three spurious zeros.
\end{abstract}

\section{Introduction} \label{intro}

The normal matrix model is a probability measure on the space of $n \times n$
normal matrices $M$ of the form
\[ \frac{1}{Z_n} e^{-\frac{n}{t} \Tr (M M^* - V(M) - \overline{V}(M^*))} \, dM, \qquad t > 0, \]
with a given potential function $V$ and $\overline{V}(z)=\overline{V(\overline{z})}$.
In the limit as $n \to \infty$ the eigenvalues of $M$ fill out a two-dimensional domain
$\Omega = \Omega(t)$, called the droplet, see Figure \ref{fig:subcriticalgrowth}, whose boundary evolves 
according to Laplacian growth (also known as Hele-Shaw flow) 
as the time parameter $t$ increases. See \cite{TBAZW, WZ} and the surveys \cite{MPT,Zab}.

The eigenvalues of $M$ are a determinantal point process that is analyzed by polynomials 
$P_{k,n}$, $\deg P_{k,n} = k$,  that are orthogonal with respect to the inner product 
\begin{equation} \label{innerproduct} 
	\langle f, g \rangle = \iint_{\mathbb C} f(z) \overline{g(z)} e^{-\frac{n}{t}(|z|^2 - V(z) - \overline{V(z)})} \, dA(z) 
	\end{equation}
where $dA$ denotes Lebesgue measure on the complex plane, see \cite{Elb}. The inner product varies with $n$.
The zeros of the diagonal polynomials $P_{n,n}$ do not fill out the droplet,
but instead are believed to  cluster on certain one-dimensional arcs  inside $\Omega$. These arcs
are referred to as the motherbody or skeleton. The orthogonal polynomials were analyzed under
various situations cases in \cite{BBLM, Elb, IT}.

In the interesting case where $V$ is a polynomial of degree $\geq 3$, the above has to be modified
since the integrals do not converge. This is done by Elbau and Felder \cite{Elb, EF} by using a cut-off, or by
Ameur, Hedenmalm and Makarov \cite{AHM1,AHM2} by a modification of $V$ outside the droplet.
These approaches work for $t$ small enough, and in fact up to a critical time $t_*$ when
the boundary of the droplet develops one or more cusps.
At the same time the motherbody meets the boundary at the cusp. 

After criticality the Laplacian growth breaks down, but it may be continued in a weak averaged
form as in \cite{LTW2,LTW3}.   A main feature of the supercritical regime is the appearance of one-dimensional
arcs  (we call them whiskers) that point out of the droplet. In the context of \cite{LTW2} these
whiskers are interpreted as pressure shock waves. Appearance of the whiskers in the supercritical
regime was also observed in \cite{BT} (for complex orthogonal polynomials with some exponential weight).

We are going to analyze this phenomenon in the simplest case where $V(z) = \frac{1}{3} z^3$
is a polynomial of degree $3$. In this case we want to view the Laplacian growth as the
evolution of an algebraic equation (the spectral curve) also after the critical time.
Before criticality the spectral curve has genus zero. After criticality the curve
has higher genus and it is characterized by the Boutroux condition, which means that all
the periods of a certain meromorphic differential are purely imaginary.

We show in the model with a cubic potential that we can select such a curve
with the Boutroux condition for  $t \in (t_*, t_{**})$ where $t_{**}$ is a {finite} second
critical time. 
Then the motherbody consists of
\[ \Sigma_1 = \Sigma_1^o \cup \Sigma_1^w \]
where $\Sigma_1^o$ is the part that remains from the original motherbody  
and $\Sigma_1^w$ are the whiskers, that (partly) stick out of the droplet $\Omega$.
At the second criticality $t_{**}$ we find that both $\Omega$ and $\Sigma_1^o$ 
disappear, and only the whiskers $\Sigma_1^w$ remain. 

We also follow in this paper the approach of 
 \cite{BK} where the orthogonal polynomials are replaced by polynomials that
are orthogonal with respect to a certain  bilinear form. This bilinear form
is well-defined in the cubic model for all time. Before critical time
it was shown in \cite{BK} that the zeros of the polynomials $P_{n,n}$ 
tend to the motherbody $\Sigma_1$ with a limiting probability distribution $\mu_1$.
See \cite{KL} for an extension to higher degree potentials.
After criticality we find that the Boutroux condition guarantees
the existence of a certain probability measure $\mu_1$ on $\Sigma_1$ and
we prove that all but at most three zeros of the polynomials $P_{n,n}$ 
cluster on $\Sigma_1$ with $\mu_1$ as limiting zero counting measure.
This result further supports the use of the Boutroux condition after criticality.
From a different perspective, the Boutroux condition was also suggested 
in \cite{LTW3}, Section 2.2 as a possible way to define the
evolution of the Laplacian growth beyond criticatility.

We next give a more precise statement of our results.

\section{Statement of results}
\subsection{Spectral curves and Boutroux condition} \label{sec:Boutroux}
As it was already stated, we study the model of Laplacian growth in the plane with a cubic potential
$V(z) = \frac{1}{3} z^3$. In \cite{BK} the  more general potential $V(z) = \frac{t_3}{3} z^3$
with $t_3 > 0$ was studied, but it can be reduced to $t_3=1$ by a simple scaling.

\begin{figure}[t] 
\centering
\begin{overpic}[width=5cm,height=5cm]{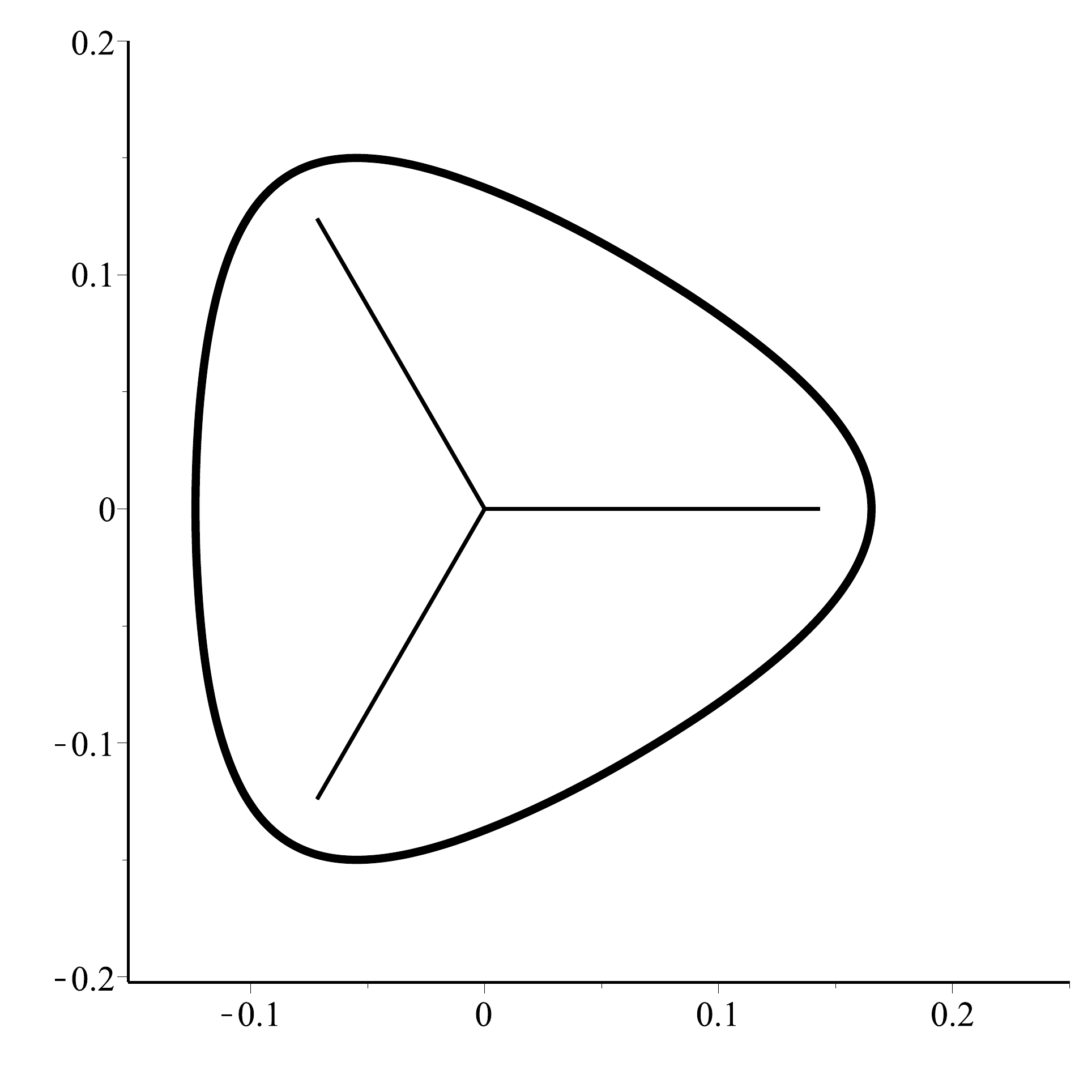}
	\put(70,48){$z_1$}
	\put(31,77){$\omega z_1$}
	\put(30,25){$\omega^2 z_1$}
	\put(65,75){$\partial \Omega$}
\end{overpic}
\qquad
\begin{overpic}[width=5cm,height=5cm]{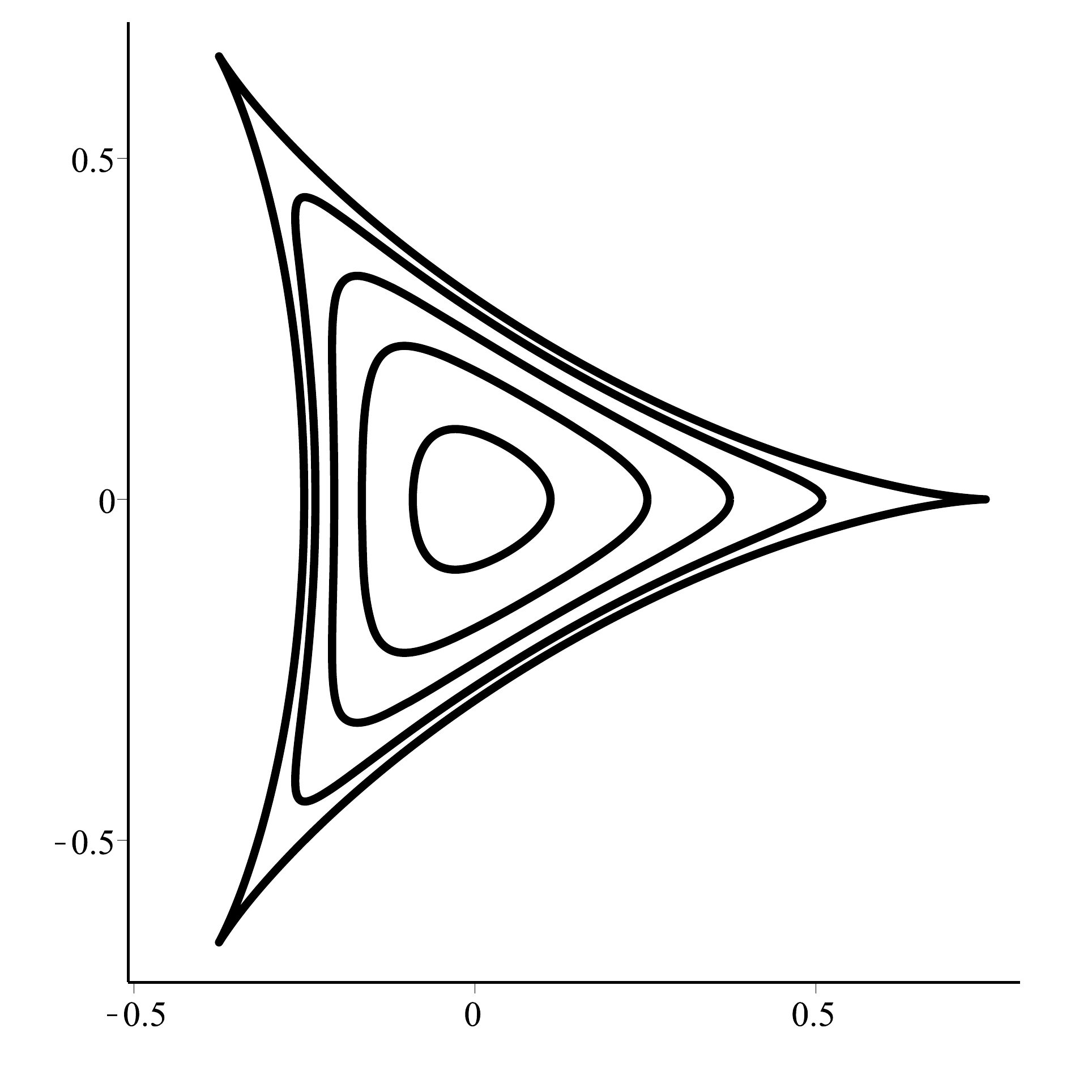}
\end{overpic}
\caption{Laplacian growth in the subcritical case. 
Left panel shows the boundary of the droplet $\Omega$ and skeleton $\Sigma_1$ for the value $t=0.02$. 
Right panel show the growth of $\partial \Omega$ for the values $t=0.01, 0.04, 0.07, 0.10, 0.125$.
\label{fig:subcriticalgrowth}}
\end{figure}

Before the first criticality $t_*$ the Laplacian growth is governed by the algebraic equation
\begin{equation} 
\label{def:speccurve}
	P(\xi, z)  = \xi^3 - z^2 \xi^2 - (1+t) z \xi + z^3 + A = 0
	\end{equation}
with {the $t$-dependent} constant 
\begin{equation} \label{def:A1}
 A = A_1(t) =   \frac{1}{32} \left(1 + 20 t - 8 t^2 - (1-8t)^{3/2} \right), \qquad t < t_* = \frac{1}{8},
\end{equation}
see \cite[Lemma 4.3]{BK}.
The choice of constant \eqref{def:A1} is dictated by the requirement that the Riemann surface {$\mathcal R$}
associated with \eqref{def:speccurve} has genus zero.

For $A = A_1(t)$, there are three branch points $z_1$, $\omega z_1$, $\omega^2 z_1$
and three nodes $z_2$, $\omega z_2$, $\omega^2 z_2$ {of $\mathcal R$}, where $z_2 > z_1 > 0$ and $\omega = e^{2\pi i/3}$.
Precise values are $z_1 = \frac{3}{4} \left(1- \sqrt{1-8 t} \right)^{2/3}$ and $z_2 = \frac{1}{4} \left(3 + \sqrt{1-8t}\right)$.

There is a solution $\xi_1(z)$ of \eqref{def:speccurve} that behaves like
\begin{equation} \label{xi1:asymptotics} 
	\xi_1(z) = z^2 + t z^{-1} + O(z^{-4}), \qquad \text{ as } z \to \infty. 
	\end{equation}
The solution \eqref{xi1:asymptotics} has an analytic continuation to the domain $\mathbb C \setminus \Sigma_1$ where
\begin{equation} \label{Sigma1}
	\Sigma_1 =  [0, z_1] \cup [0, \omega z_1] \cup [0, \omega^2 z_1].
	\end{equation}
Then the equation $\xi_1(z) = \overline{z}$ with $z \in \mathbb C \setminus \Sigma_1$ 
defines the boundary of a domain $\Omega(t)$ that evolves according to the model of Laplacian growth,
see \cite[Theorem 2.6]{BK}.
Putting $\xi = \overline{z}$ in \eqref{def:speccurve} we find that $\partial \Omega(t)$ is characterized
by the equation
\begin{equation} \label{dOmega} 
	\partial \Omega(t) : \quad 2 \Re(z^3) - |z|^4 - (1+t) |z|^2 + A_1(t) = 0. 
	\end{equation}
The boundary $\partial \Omega(t)$ encloses the motherbody $\Sigma_1$,
and the nodes $z_2$, $\omega z_2$, $\omega^2 z_2$ are exterior to $\partial \Omega(t)$.
Note that $z= \omega^j z_2$ for $j=0,1,2$ also satisfies the equation \eqref{dOmega}, 
but these are isolated points, and are not considered to be part of $\partial \Omega(t)$.

In the cut-off model of Elbau and Felder \cite{Elb,EF}, the eigenvalues in the normal matrix model cluster
on $\Omega(t)$ with uniform density. The zeros of the orthogonal polynomials accumulate
on the set $\Sigma_1$ with a probability measure $\mu_1$ on $\Sigma_1$ as the limit of the normalized
zero counting measures. The measure has the property that
\[ \int_{\Sigma_1} \frac{d\mu_1(s)}{z-s} = \frac{1}{\pi t} \iint_{\Omega} \frac{dA(s)}{z-s}, 
\qquad z \in \mathbb C \setminus \Omega. \]
At the critical $t= t_*$ the three branch points $\omega^j z_1$, $j=0,1,2$, come to the boundary $\partial \Omega(t)$,
which develops three cusps. Also the node $z_2$ then coincides with $z_1$,
see Figure \ref{fig:subcriticalgrowth}.

For the supercritical $t > t_*$, we are still working with a spectral curve of the form \eqref{def:speccurve} 
but with a different determination of $A = A(t)$. 
The corresponding Riemann surface $\mathcal R$ has now genus three, where the branch points 
of $\mathcal R$ can be obtained through the 
discriminant $ D(P)(z) = Q(z^3) $ of \eqref{def:speccurve} with respect to the variable $\xi$ 
(the cubic polynomial $Q(w)$ is given in \eqref{def:Q}).
The evolution is now described by the following
condition.

\begin{definition}\label{def-Boutroux}
Let $t > 0$ and $A \in \mathbb R$. We say that the meromorphic differential $\xi dz$ 
defined on the compact Riemann surface $\mathcal R$ associated with the equation \eqref{def:speccurve} 
has the \textbf{Boutroux condition}  if all the periods of $\xi dz$ are purely imaginary. That is,
\begin{equation} \label{cond:Boutroux} 
	\oint_{\gamma} \xi dz \in i \mathbb R 
	\end{equation}
for every closed contour $\gamma$ on $\mathcal R$ that avoids the poles of $\xi dz$.
\end{definition}
The poles of $\xi dz$ are at the two points at infinity, with real residues $\pm t$,
see e.g.\ \eqref{eq:xiatinfinity}.
The condition \eqref{cond:Boutroux} is therefore satisfied for contours $\gamma$ that only 
go around the poles and are homotopic to zero on $\mathcal R$. 
In particular the Boutroux condition is satisfied if the Riemann surface
 has genus zero.



In the supercritical case $t > t_*$, the Riemann surface associated with \eqref{def:speccurve}
has genus three (unless $A=0$ in which case the genus
is two) and then \eqref{cond:Boutroux} presents a condition on the $\xi dz$ periods of the 
non-trivial cycles $\gamma$ on the  surface. 

It turns out that the cubic equation \eqref{def:speccurve} has nine branch points, namely
$\omega^j z_k$, $j=0,1,2$, $k=1,2,3$ where $\omega = e^{2\pi i/3}$ and  $z_1, z_2, z_3$
are the branch points lying in the sector
\begin{equation} \label{eq:S0} 
		S_0 = \{ z \in \mathbb C \mid - \tfrac{\pi}{3} < \arg z < \tfrac{\pi}{3} \}.
\end{equation}
In the case of interest we can take $z_1$ to be real and $z_3 = \overline{z}_2$ with $\Im z_2 > 0$.
The three sheeted Riemann surface $\mathcal R$ associated with \eqref{def:speccurve} has the 
sheet structure as in Figure \ref{fig:three-sheets}.

\begin{figure}[!t]
\centering
\subfigure[$\mathcal R_1 = \mathbb C \setminus \Sigma_1$]{
\begin{tikzpicture}[scale=2]
\draw[very thick] (0,0)--(0.608,0) node[left,above]{$z_1$};
\draw[very thick] (0.608,0)--(0.821,0.135) node[right]{$z_2$};
\draw[very thick] (0.608,0)--(0.821,-0.135) node[right]{$z_3$};
\draw[very thick] (0,0)--(-0.304,0.527) node[right]{$\omega z_1$};
\draw[very thick] (-0.304,0.527)--(-0.527,0.643) node[left]{$\omega z_2$};
\draw[very thick] (-0.304,0.527)--(-0.294,0.778) node[right]{$\omega z_3$};
\draw[very thick] (0,0)--(-0.304,-0.527) node[right]{$\omega^2 z_1$};
\draw[very thick] (-0.304,-0.527)--(-0.527,-0.643) node[left]{$\omega^2 z_3$};
\draw[very thick] (-0.304,-0.527)--(-0.294,-0.778) node[right]{$\omega^2 z_2$};
\draw  (-1.0,-1.0) rectangle (1.5,1.0);
\end{tikzpicture}
}  
\end{figure}
\begin{figure}[!t] \setcounter{subfigure}{1}
\centering
\subfigure[$\mathcal R_2 = \mathbb C \setminus (\Sigma_1 \cup \Sigma_2)$]{
\begin{tikzpicture}[scale=2]
\draw[very thick] (0,0)--(0.608,0) node[left,above]{$z_1$};
\draw[very thick] (0.608,0)--(0.821,0.135) node[right]{$z_2$};
\draw[very thick] (0.608,0)--(0.821,-0.135) node[right]{$z_3$};
\draw[very thick] (0,0)--(-0.304,0.527) node[right]{$\omega z_1$};
\draw[very thick] (-0.304,0.527)--(-0.527,0.643) node[left]{$\omega z_2$};
\draw[very thick] (-0.304,0.527)--(-0.294,0.778) node[right]{$\omega z_3$};
\draw[very thick] (0,0)--(-0.304,-0.527) node[right]{$\omega^2 z_1$};
\draw[very thick] (-0.304,-0.527)--(-0.527,-0.643) node[left]{$\omega^2 z_3$};
\draw[very thick] (-0.304,-0.527)--(-0.294,-0.778) node[right]{$\omega^2 z_2$};
\draw[dashed](0,0)--(0.456,0.7905);
\draw[dashed](0,0)--(0.456,-0.7905);
\draw[dashed](0,0)--(-0.912,0);
\draw  (-1.0,-1.0) rectangle (1.5,1.0);
\end{tikzpicture}
} 
\subfigure[$\mathcal R_3 = \mathbb C \setminus \Sigma_2$]{
\begin{tikzpicture}[scale=2]
\draw[dashed](0,0)--(0.456,0.7905);
\draw[dashed](0,0)--(0.456,-0.7905);
\draw[dashed](0,0)--(-0.912,0);
\draw  (-1.0,-1.0) rectangle (1.5,1.0);
\end{tikzpicture}
} 

\caption{The three sheets $\mathcal R_1$, $\mathcal R_2$ and $\mathcal R_3$ of the Riemann surface \label{fig:three-sheets}}
\end{figure}

The restriction to the three sheets defines three functions that have the
asymptotic behavior
	\begin{equation} \label{eq:xiatinfinity}
\begin{aligned}
	\xi_1(z) & = z^2 + t z^{-1} + O(z^{-4}),  && \text{ as } z \to \infty, \\
	\xi_2(z) & = z^{1/2}  - \frac{1}{2} t z^{-1} + O(z^{-5/2}), && \text{ as } z \to \infty, \, z \in S_0, \\
	\xi_3(z) & = -z^{1/2} - \frac{1}{2} t z^{-1} + O(z^{-5/2}), && \text{ as } z \to \infty, \, z \in S_0. 
	\end{aligned}
\end{equation}
The first two sheets are connected by the cuts
\begin{align} \label{eq:Sigma1} 
	\Sigma_1 & = \Sigma_1^o \cup \Sigma_1^w, \qquad 
	\Sigma_1^o  = \bigcup_{j=0}^2 [0,\omega^j z_1]
	\end{align}
where $\Sigma_1^w$ are cuts that connect $\omega^j z_1$ with $\omega^j z_2$ and $\omega^j z_3$
for $j=0,1,2$. The sheets $\mathcal R_2$ and $\mathcal R_3$ are connected via the cut
\begin{align} \label{eq:Sigma2}
	\Sigma_2 = \{ z \in \mathbb C \mid z^3 \in \mathbb R^-\}.
	\end{align}

We refer to $\Sigma_1^w$ as  whiskers that stick out of the branch points $\omega^j z_1$ for $j=0,1,2$,
see Figure \ref{fig:whiskers}.
The whiskers that connect the branch points $\omega^j z_1$, $\omega^j z_2$ and $\omega^j z_3$ 
for $j=1,2,3$, are arbitrary at this point,
but it turns out that they can be defined in a special way if the Boutroux condition is satisfied.
This is part of our main result which we state as follows.

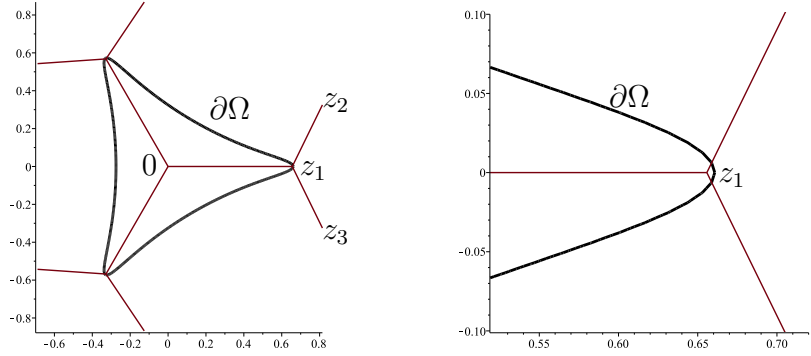
\begin{figure}[!t] 
\centering
\begin{overpic}[width=5cm,height=5cm]{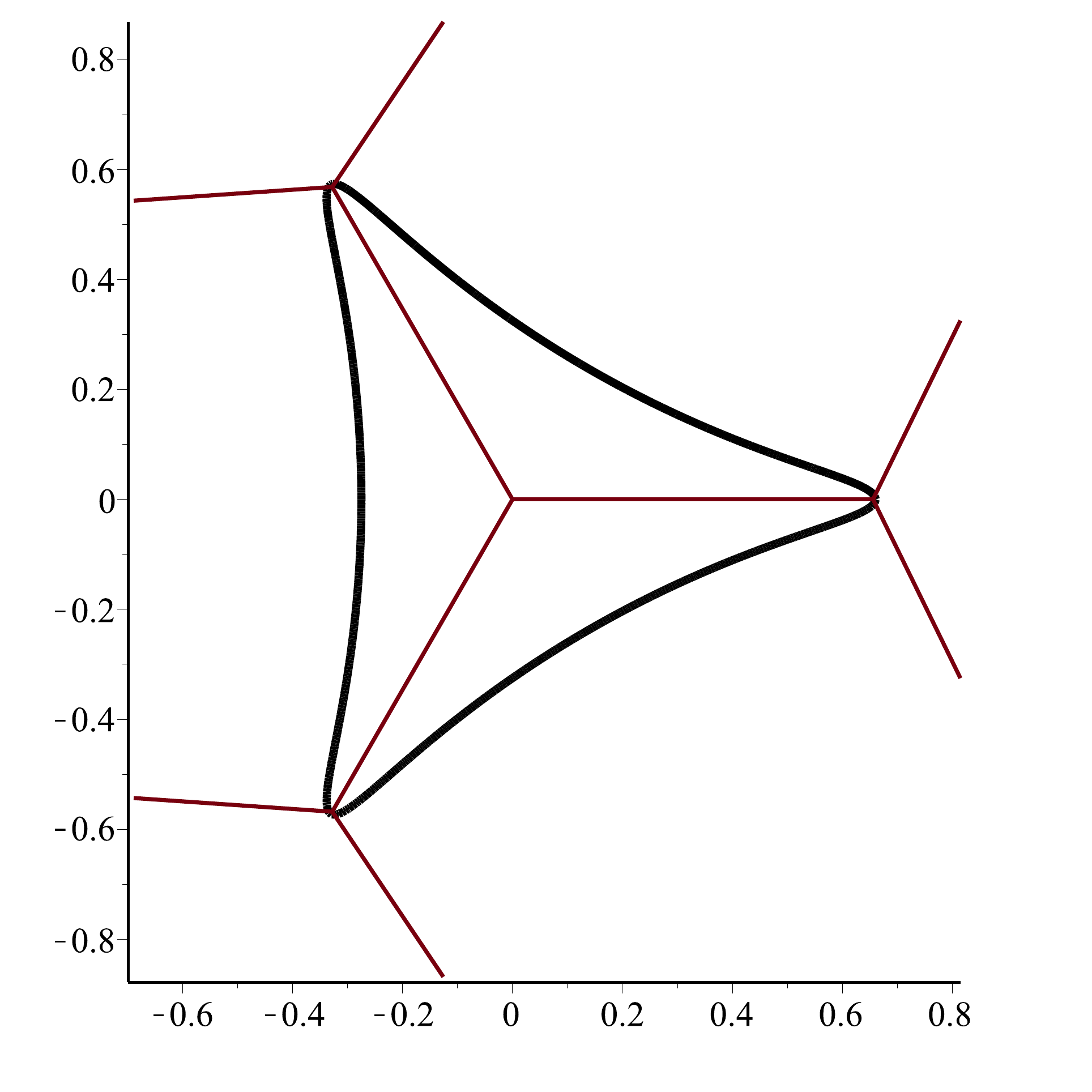}
	\put(40,52){$0$}
	\put(82,52){$z_1$}
	\put(88,70){$z_2$}
	\put(88,35){$z_3$}
	\put(58,67){$\partial \Omega$}
\end{overpic}
\qquad
\begin{overpic}[width=5cm,height=5cm]{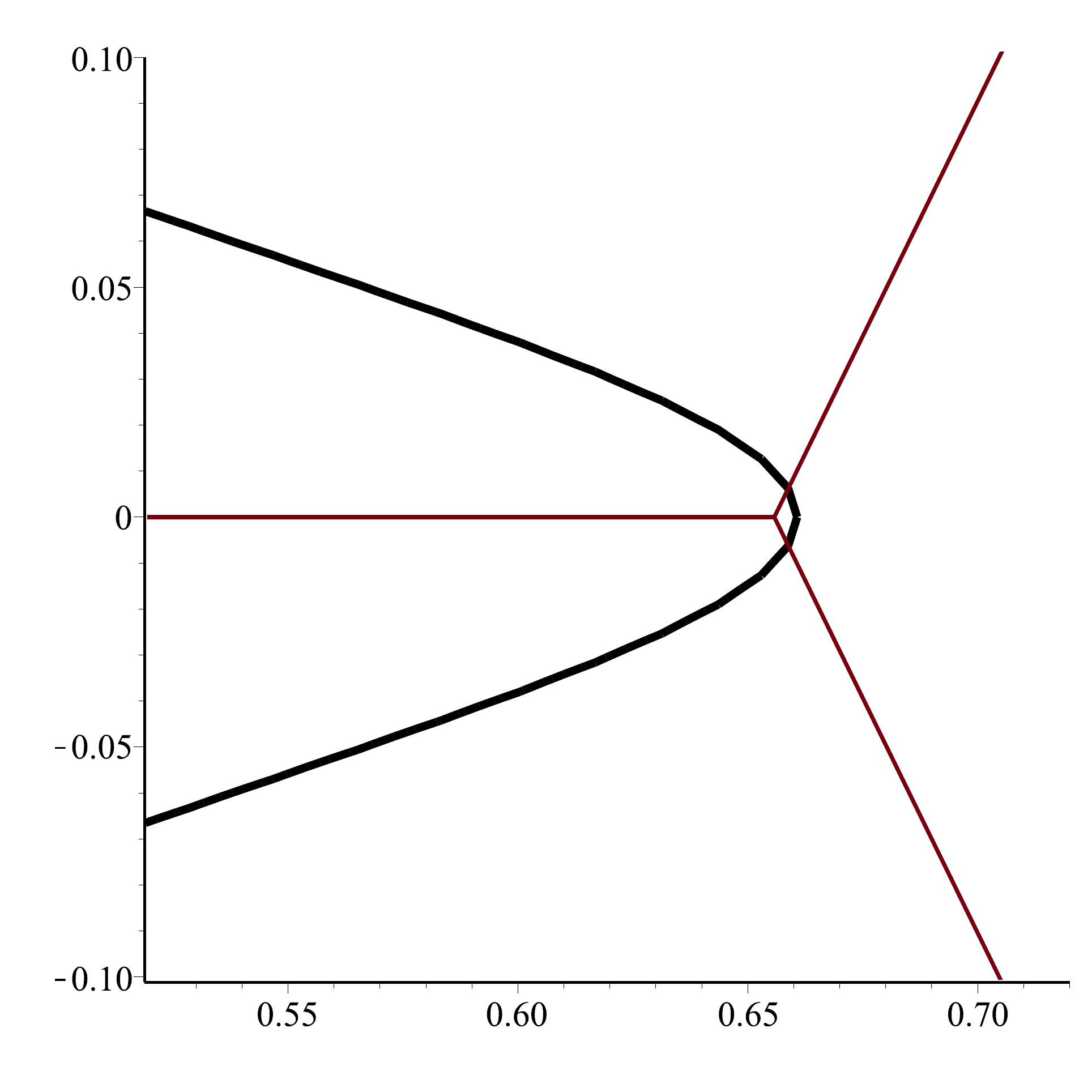}
	\put(74,50){$z_1$}
	\put(45,70){$\partial \Omega$}
\end{overpic}
\caption{Laplacian growth in the supercritical regime. Left panel shows
the domain  $\Omega(t)$ and the set $\Sigma_1(t)$ for $t=0.2$.
The right panel zooms in on a neighborhood of  $z_1$, which lies inside $\Omega(t)$. Most of the
whiskers stick out of $\Omega(t)$.
 \label{fig:whiskers}}
\end{figure}  

\begin{theorem} \label{theorem1}
\begin{enumerate}
{\item[\rm (a)]}
There exist $t_{**} > t_*$ and a unique continuous function $A : [t_*, t_{**}) \to (0, \infty)$ 
with $A(t_*) = A_1(t_*) = \frac{27}{256}$ and
\begin{equation} \label{boundaryA} 
 \lim_{t \to t_{**}-} A(t) = 0
	\end{equation}
such that the differential $\xi dz$ has the Boutroux condition
for every $t \in (t_*, t_{**})$ and $A = A(t)$. Moreover, for every $t \in (t_*, t_{**})$ and $A = A(t)$:
\item[\rm (b)] There is a simple analytic arc $\gamma_{1,2}$
from $z_1$ to $z_2$ lying in the sector $S_0$ such that
\begin{equation} \label{trajgamma12} 
\Re\int_{z_1}^{z} (\xi_1(s) - \xi_2(s)) ds =0, \qquad \text{for every } z \in \gamma_{1,2}. 
	\end{equation}
We define $\gamma_{1,3} = \overline{\gamma_{1,2}}$. Then $\gamma_{1,3}$
is a simple analytic arc  from $z_1$ to $z_3$ such that
\begin{equation} \label{trajgamma13} 
\Re\int_{z_1}^{z} (\xi_1(s) - \xi_2(s)) ds =0, \qquad \text{for every } z \in \gamma_{1,3}.
	\end{equation}
\item[\rm (c)] Let  $\Sigma_1$ be as in \eqref{eq:Sigma1} with
\begin{equation} \label{eq:Sigma1w} 
	\Sigma_1^w = \bigcup_{j=0}^2 \omega^j (\gamma_{1,2} \cup \gamma_{1,3}) 
	\end{equation}
where $\gamma_{1,2}$ and $\gamma_{1,3}$ are as in part  {({b})}. Then
\begin{equation} \label{eq:mu} 
	d\mu_1(s) = \frac{1}{2\pi i t} (\xi_{2,+}(s)-\xi_{1,+}(s)) ds, \qquad s \in \Sigma_1 
	\end{equation}
defines a probability measure on $\Sigma_1$. Here all parts in $\Sigma_1$ are oriented outwards, 
that is, away from the origin,
$ds$ is the complex line element that is compatible with this orientation 
and $\xi_{1,+}$, 
$\xi_{2,+}$ denote the limiting values of $\xi_{1,2}$ if we approach $\Sigma_1$ from the left 
(and $\xi_{1,-}$, $\xi_{2,-}$ are the limiting values from the right).
\item[\rm (d)] The equation 
\begin{equation} \label{dOmegasuper} 
	\partial \Omega(t) : \quad 2 \Re(z^3) - |z|^4 - (1+t)|z|^2 + A = 0 
	\end{equation}
 defines the boundary of a domain $\Omega(t)$ that is such that 
 the points $ \omega^j z_1$, $j=0,1,2$ are inside and the points $\omega^j z_2$, $\omega^j z_3$, $j=0,1,2$, 
are outside $\Omega(t)$, and
\begin{equation} \label{eq:areak} 
 \frac{1}{2\pi i} \int_{\partial \Omega(t)} \frac{\overline{s}}{s^k} ds\textsc{} + t \int_{\Sigma_1^w} \frac{d\mu_1(s)}{s^k} 
	= \begin{cases} t & \text{ for } k = 0, \\
		1 & \text{ for } k = 3, \\
		0 & \text{ otherwise.} \end{cases} 
		\end{equation}
\end{enumerate}
In particular the case $k=0$ in \eqref{eq:areak} reduces (by Green's theorem) to
\begin{equation} \label{eq:area}
	t = \frac{1}{\pi} \area \Omega (t)  + t  \mu_1(\Sigma_1^w).  
	\end{equation}
\end{theorem}
The proof of Theorem \ref{theorem1} is in section \ref{section:proof1}.
The value of $t_{**}$ has been calculated numerically and it is approximately
\[ t_{**} = 3.6 \cdots \]

The equations \eqref{eq:areak} represent the continuation of the Laplacian growth after
criticality. In fact the equations \eqref{eq:areak} also holds before criticality,
since then $\Sigma_1^w = \emptyset$ and \eqref{eq:areak} reduces to the equations
that characterize the Laplacian growth in terms of the exterior harmonic moments of the 
droplet \cite{WZ}. After criticality these equations are modified by the extra term $t \int_{\Sigma_1^w} \frac{d\mu_1(s)}{s^k}$
representing the contribution from the whiskers. See Figure \ref{fig:whiskers} for an illustration.

\begin{figure}[t] \centering
\includegraphics[width=6cm,height=6cm]{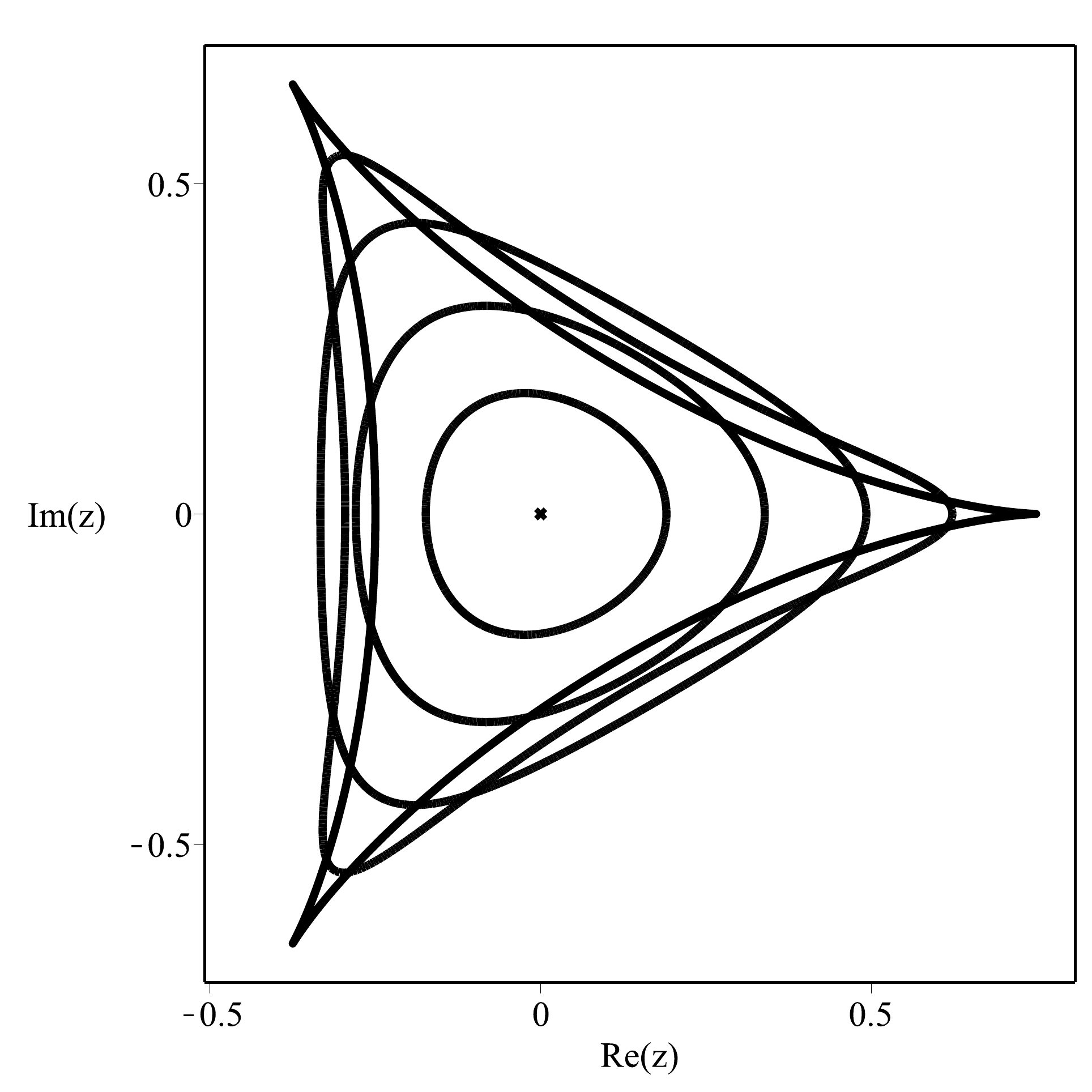}
\caption{Evolution of the domain $\Omega(t)$ in the supercritical case
for the values $t=0.125, 0.3, 1.0, 2.0, 3.0$ and $3.6$.
The domain does not grow with $t$ anymore, but starts to shrink and eventually
collapses to a point. 
\label{fig:Omega-supercritical}}
\end{figure}

In the supercritical case, the domain $\Omega(t)$ is no longer growing with $t$, 
see Figure \ref{fig:Omega-supercritical}. In fact, the following corollary shows that
 $t=t_{**}$ represents a second criticality, see also Figure \ref{fig:secondcrit}.
 
\begin{corollary} \label{cor-1}
 At $t=t_{**}$ we have $z_1(t_{**})=0$, whereas $z_{2}(t_{**})$,  $z_{3}(t_{**})$  
are solutions of 
\begin{equation} \label{eq-zt2}
 z^6+ \frac{t^2+20t-8}{4} z^3 +(1+t)^3=0
 \end{equation}
{with $t=t_{**}$.}
  Moreover,  at  $t = t_{**}$  the domain $\Omega(t)$ reduces to a point {at the origin.} 
 \end{corollary}
\begin{proof} To prove Corollary \ref{cor-1}, it is sufficient to  observe that $A(t_{**})=0$ reduces: 
i) the discriminant $Q(w)$, given by \eqref{def:Q}, to 
\begin{equation} \label{Q(t_**)}
Q(w) = w[4w^2 + (t_{**}^2 + 20t_{**}  -8) w  + 4 (1+t_{**})^3 ],
\end{equation}		
and; ii) the expression \eqref{dOmegasuper} for $\partial \Omega(t_{**})$  in polar coordinates $z = re^{i\theta}$  
to 
\begin{equation} \label{Omega(t_**)}
	r(r^2-2r\cos 3\theta +1+t_{**})=0,
\end{equation}
with only solution $r = 0$, since the expression in brackets is positive.
\end{proof}

Solving  \eqref{eq-zt2}  we find the following explicit values for  $z_{2,3} = z_{2,3}(t_{**})$,
\[ z_{2,3} = \frac{1}{2} \left[8-20t_{**} - t_{**}^2 \pm i \sqrt{t_{**}(8-t_{**})^3} \right]^{\frac{1}{3}}. \]

\begin{figure}[t] 
\centering
\begin{overpic}[scale=0.6]{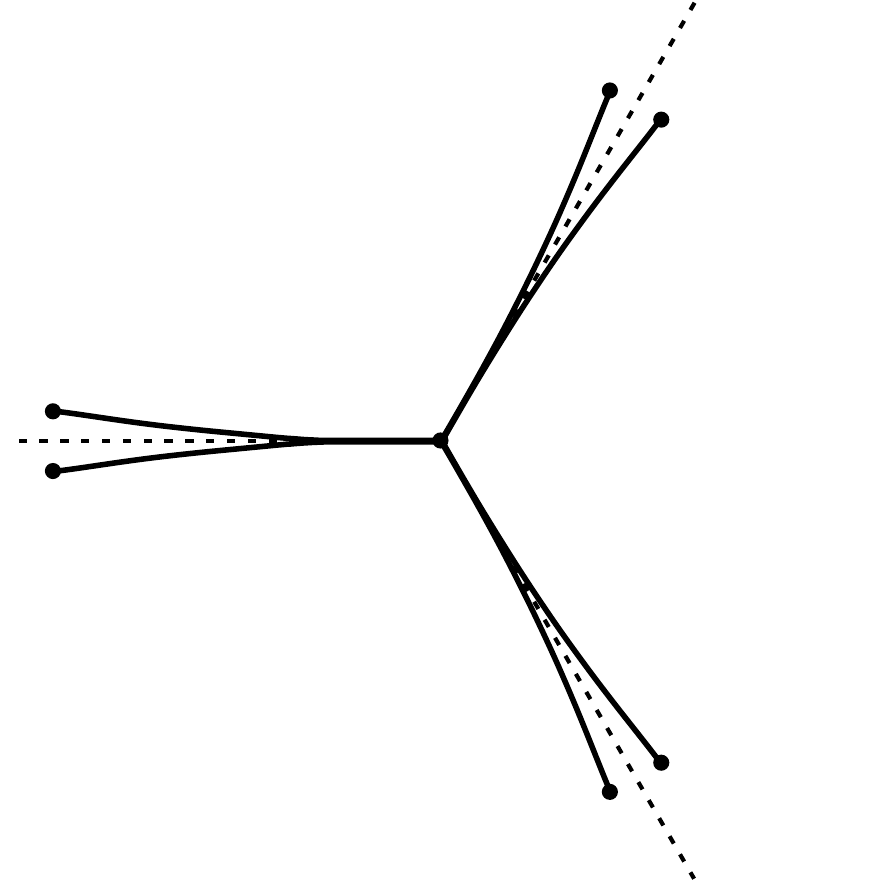}
	\put(52,49){$z_1 =0$}
	\put(76,84){$z_2$}
	\put(76,13){$z_3$}
\end{overpic}
 \caption{At the second critical time $t_{**}$, the domain $\Omega(t)$ shrinks to a point,
and the set $\Sigma_1$ consists of whiskers only, as shown in the figure. 
The figure also shows the rays $\arg z = \pm \pi/3$, $\arg z = \pi$ that make up the set $\Sigma_2$ (dotted lines).
The branch points $z_2$ and $z_3$ remain in the sector $S_0$ and do not come to $\Sigma_2$.
\label{fig:secondcrit}}
\end{figure}

The importance of the Boutroux condition is that it gives rise to the
probability measure \eqref{eq:mu} on $\Sigma_1$ when $t \in (t_*, t_{**})$,
which is a continuation of the probability measure $\mu_1$ on $\Sigma_1$
for subcritical $t \in (0, t_*)$ described in \cite{BK}. The equations \eqref{eq:areak} may
be viewed as characterizing the Laplacian growth after criticality.

The set $\Sigma_2$  carries the measure, see \cite[formula (4.12)]{BK},
\begin{equation} \label{defmu2} 
	d\mu_2(s) = \frac{1}{2\pi i t} \left( \pm 2s^{1/2} + \xi_{3,+}(s) - \xi_{2,+}(s)\right) ds, \qquad s \in \Sigma_2,
	\end{equation}
where the choice of sign $\pm$ is such that $\pm 2s^{1/2} + \xi_{3,+}(s) - \xi_{2,+}(s) = O(s^{-1})$  as $s \to \infty$.
As in the subcritical case, the measure $\mu_2$ is real and positive with total mass $1/2$.
In \cite[Theorems 2.5 and 2.6]{BK} the two measures $(\mu_1, \mu_2)$ were found as the minimizer of a vector
equilibrium problem. 

This vector equilibrium problem is still relevant in the supercritical case,
as it can be shown that the measures  jointly minimize the energy functional
\begin{multline} \label{vectorenergy}
	E(\mu_1,\mu_2) = \iint \log \frac{1}{|x-y|}d\mu_1(x)d\mu_1(y)
	+  \iint \log \frac{1}{|x-y|}d\mu_2(x)d\mu_2(y) \\
	- \iint \log \frac{1}{|x-y|}d\mu_1(x)d\mu_2(y) \\
	+ \frac{1}{t}  \sum_{j=0}^2 
		\int_{\Sigma_1 \cap S_j} \Re \left(  \frac{2}{3\sqrt{3}} (\omega^{-j} s)^{3/2} - \frac{1}{3} s^3 \right) d\mu_1(s)
		\end{multline}
among all vectors of measures $(\mu_1, \mu_2)$ with $\mu_j$ on $\Sigma_j$ and $\mu_j(\Sigma_j) = \frac{1}{j}$
for $j=1,2$. Here $S_j = \omega^j S_0$ for $j=1,2$. 
See \cite{Deift,HK,SaTo} for more on
(vector) equilibrium problems for logarithmic potentials.

We conjecture that the optimal set $\Sigma_1$ is characterized by a max-min property 
for the energy \eqref{vectorenergy}.
Let $\mathcal T$ denote the collection of contours $\Sigma$ such that:
\begin{enumerate}
\item[(a)] $\Sigma$ is a connected contour (= finite union of analytic arcs) that
is symmetric in the real axis, and invariant under the $\mathbb Z_3$ rotational
symmetry $z \mapsto \omega z$.
\item[(b)] The part of $\Sigma$ in $S_0$ connects $0$ to infinity in the directions $\arg z = \pm \frac{\pi}{3}$.
\end{enumerate}

For each  $\Sigma \in \mathcal T$ we define $\mathcal E(\Sigma)$
as the infimum of $E(\mu_1, \mu_2)$ where $\mu_1$ is a measure on $\Sigma$ with $\mu_1(\Sigma) = 1$
and $\mu_2$ is a measure on $\Sigma_2$ with $\mu_2(\Sigma_2) = \frac{1}{2}$.
Then we conjecture that
\[ \mathcal E(\Sigma_1) = \sup_{\Sigma \in \mathcal T} \mathcal E(\Sigma). \]
Such a characterization would be analogous to the $S$-curves that play
a role in rational approximation and complex non-Hermitian orthogonality 
\cite{KS,Rak}.

\subsection{Orthogonal polynomials} \label{sec:orthopoly}

A basic ingredient in the study of the normal matrix model before criticality are the orthogonal polynomials
with respect to the scalar product \eqref{innerproduct}, or with respect to
a modified version involving a cut-off. The cut-off approach does not work in the supercritical case.

In \cite{BK} a different regularization of the scalar product \eqref{innerproduct} was proposed
which for the model with cubic potential $V(z) = \frac{1}{3} z^3$ leads to the Hermitian form
\begin{equation} \label{Hform} 
	\langle f, g \rangle = \frac{1}{2\pi i} \sum_{j=0}^2 \sum_{k=0}^2 \epsilon_{j,k} 
	 \int_{\Gamma_j} dz \int_{\overline{\Gamma}_k} dw 
		f(z) \overline{g}(w) e^{- \frac{n}{t} (wz - \frac{1}{3} (w^3 + z^3))}, 
		\end{equation}
defined on polynomials $f$ and $g$, where $\Gamma_j$ is an unbounded contour 
stretching out to infinity from $e^{(2j-1)\pi i/3} \infty$ to $e^{(2j+1) \pi i /3} \infty$ for $j=0,1,2$, and
\[ \epsilon_{j,k} = \begin{cases}
			0  & \text{ if } j = k, \\
			1  &  \text{ if } j \equiv k + 1 \mod{3}, \\
			-1 &  \text{ if } j \equiv k -1 \mod{3}.
			\end{cases} \]
The Hermitian form \eqref{Hform} satisfies the  identity
\[  t \langle f, g' \rangle - n \langle zf, g \rangle + n \langle f, V'g \rangle=0 \]
which is also satisfied by \eqref{innerproduct} with the cut-off regularization if one forgets about boundary terms,
see \cite{BK}.

The Hermitian form \eqref{Hform} depends on $n$.
The orthogonal  polynomial $P_{k,n}$ is a monic polynomial of degree $k$ that satisfies
\begin{equation} \label{Horthogonal} 
	\langle P_{k,n}, z^j \rangle = 0, \qquad \text{for } j=0,1, \ldots, k-1.
	\end{equation}
In the subcritical case, it was shown in \cite{BK} that the zeros of the diagonal polynomials $P_{n,n}$
accumulate on the motherbody $\Sigma_1$ as $n \to \infty$, with $\mu_1$ as limit
of the normalized zero counting
measures. This result followed from a steepest descent analysis of the Riemann-Hilbert (RH)
problem  that characterizes the orthogonal polynomials with respect to \eqref{Hform}.
The RH problem has size $3 \times 3$. It results in a strong asymptotic formula
\[ P_{n,n}(z) = M_{11}(z) e^{n g_1(z)} ( 1 + O(1/n)), \qquad \text{ for } z \in \mathbb C \setminus \Sigma_1 \]
where $g_1(z) = \int \log(z-s) d\mu_1(s)$ and $M_{11}(z)$ is a prefactor that arises as the $11$-entry
of a global parametrix $M$ that is used in the steepest descent analysis. 

The Hermitian form \eqref{Hform} also makes sense in the supercritical case, and so does
the characterization of the orthogonal polynomials by means of the RH problem. 
We could do the steepest descent analysis also in this case. The outcome is a strong
asymptotic formula for $P_{n,n}$ as in Theorem \ref{theorem2}, namely
\[ P_{n,n}(z) = M_{n,11}(z) e^{n g_1(z)} ( 1 + O(1/n)), \qquad \text{ for } z \in \mathbb C \setminus \Sigma_1 \]
where now the prefactor $M_{n,11}(z)$ varies with $n$ and contains an elliptic theta function,
see \eqref{def-Mn11} below.

The prefactor $M_{n,11}$ has at most three zeros in $\mathbb C \setminus \Sigma_1$ 
that do not tend to $\Sigma_1$ as $n \to \infty$. They correspond to spurious zeros of $P_{n,n}$ 
{and} perform a quasi-periodic motion on the Riemann surface {$\mathcal R$}. 
This is the same phenomenon as happens for usual orthogonal polynomials
with an orthogonality measure that is supported on several intervals.

In the present case, the quasi-periodic motion takes place on the real part of $\mathcal R$ which
can be identified with the cycle $c_R$ that starts at $z_1$ and goes along
the intervals $[z_1, \infty)$ and $(-\infty,0]$ on the first sheet $\mathcal R_1$,
then along $[0, \infty)$ on the third sheet $\mathcal R_3$, 
and, finally, from right to left along the interval $(\infty, z_1]$ on the second sheet $\mathcal R_2$. 

For a precise description of $M_{n,11}$ we need some more notions related to the 
Riemann surface $\mathcal R$, which as we recall, has genus three in the supercritical case
$t \in (t_{*}, t_{**})$. There is a unique holomorphic 
differential $\omega_R$ on $\mathcal R$ that is $\mathbb Z_3$-invariant
and that is normalized such that
\begin{equation} \label{omegaRnormalized} 
	\oint_{c_R} \omega_R = 1. 
	\end{equation}
It can be explicitly given by
\begin{equation} \label{omegaR} 
	\omega_R =   \frac{3C}{3\xi^2 - 2z^2 \xi - (1+t)z} dz 
	\end{equation}
with the constant 
\[ C = \frac{1}{3} \left[ \oint_{c_R} \frac{1}{3\xi^2 - 2z^2 \xi - (1+t)z} dz \right]^{-1} \]

There is another cycle $a_R$ going around  $\Sigma_1^w \cap S_0$ on the first sheet in
counterclockwise direction (the cycle passes through the branch point $z_1$). By symmetry in the real axis we have
\begin{equation} \label{def-tau} 
	\tau := \oint_{a_R} \omega_R \in i \mathbb R^+ 
	\end{equation}
The theta function with  $q=e^{\pi i(1 + \tau)/2}$ (and quasi-period $(1+\tau)/2$) is defined by
\begin{equation} \label{theta3-def} \theta(s) = \sum_{n = - \infty}^{\infty} q^{n^2} e^{2n\pi is}
	= \sum_{n=-\infty}^{\infty} e^{\pi i n^2 \frac{1+\tau}{2} + 2\pi i ns}. \end{equation}
It gives an entire function 
in the complex $s$-plane with periodicity properties 
\begin{equation} \label{theta-periodicity} 
	\begin{aligned}
	\theta(s+1) & = \theta(s) = \theta(-s),  \\
	\theta(s+ \tfrac{1+\tau}{2}) & = e^{-\pi i \frac{1+\tau}{2}} e^{-2\pi i s} \theta(s), \\
	\quad \theta(s+ \tfrac{1-\tau}{2}) & = e^{-\pi i \frac{1+\tau}{2}} e^{2\pi i s} \theta(s).
	\end{aligned}
	\end{equation} 
The theta function has a simple zero at the values $s_0 + k + \frac{1+\tau}{2} l$, $k, l \in \mathbb Z$,
where
\begin{equation} \label{def-s0} 
	s_0 = \frac{-1+\tau}{4},  
	\end{equation}
and no other zeros.

We further define
\begin{equation} \label{def-beta}
	\beta := \frac{1}{6} \mu_1(\Sigma_1^w) \end{equation}
and for $\varepsilon >0$,
\begin{equation} \label{Nepsilon} 
	\mathbb N_{\varepsilon} = \{ n \in \mathbb N \mid 
	\dist_{\mathbb R \slash \mathbb Z} \left(n \beta, \tfrac{1}{2} + \tfrac{\tau}{2\pi i} \log 2 +
		\int_{\infty_1}^{-A^{1/3}} \omega_R\right) \geq \varepsilon \}
		\end{equation}
where $A = A(t)$ as before, and $-A^{1/3}$ denotes the point $z = -A^{1/3}$, $\xi = 0$ that is
on the first sheet of the Riemann surface, see \eqref{def:speccurve}.

\begin{theorem} \label{theorem2}
Suppose $t \in (t_*, t_{**})$. Let $\varepsilon >0$. Then for any large $n \in \mathbb N_{\varepsilon}$
the polynomials $P_{n,n}$ exist and
\begin{equation}\label{Pnn-asymp} 
	P_{n,n}(z) =  (M_{n,11}(z) + O(1/n)) e^{n g_1(z)} \qquad \text{as } n \in \mathbb N_{\varepsilon},~n \to \infty  
	\end{equation}
uniformly for $z$ in compact subsets of $\mathbb C \setminus \Sigma_1$, where: 
\begin{equation} \label{def-g1} g_1(z) = \int_{\Sigma_1} \log(z-s) d\mu_1(s),
\end{equation}
and 
\begin{multline}  \label{def-Mn11} 
	M_{n,11}(z) = 2^{2 \int_z^{\infty_1} \omega_R}
	\frac{\theta(s_0 + \int_{-A^{1/3}}^{-\infty_1} \omega_R)}{\theta(s_0 + \int_{-A^{1/3}}^z \omega_R)} 
		\frac{\theta(s_0 + \int_{-A^{1/3}}^z \omega_R + n\beta - \frac{\tau}{2\pi i} \log 2 - 1/2)} 
		{\theta(s_0 + \int_{-A^{1/3}}^{\infty_1} \omega_R + n \beta - \frac{\tau}{2\pi i} \log 2 - 1/2)} \\
		\times
		  \frac{\xi_1(z)}{(3\xi_1^2(z) - 2 z^2 \xi_1(z) - (1+t)z )^{1/2}}. 
			\end{multline}
The error term $O(1/n)$ in \eqref{Pnn-asymp}, which depends on $\varepsilon$, is uniform
on compact subsets of $(t_*,t_{**})$.
\end{theorem}
The proof of Theorem \ref{theorem2} can be found in section \ref{section:proof2}.

The theta function $\theta(s_0 + \int_{-A^{1/3}}^z \omega_R)$ has zeros at $z= - A^{1/3}$ and
also at $ - \omega^j A^{1/3}$ for $j=1,2$, since  by the $\mathbb Z_3$ symmetry of $\omega_R$ 
one has that $\int_{-A^{1/3}}^{- \omega^j A^{1/3}} \omega_R=0$.
However, these zeros in the denominator are cancelled by the fact that the $\xi_1(z)$ has 
zeros at these same values. If $n \not\in \bigcup_{\varepsilon > 0} \mathbb N_{\varepsilon}$ then
\[ \theta\left(s_0 + \int_{-A^{1/3}}^{\infty_1} \omega_R + n \beta - \frac{\tau}{2\pi i} \log 2 - 1/2\right) = 0 \]
and then the right-hand side of \eqref{def-Mn11} is not well-defined.

For each $n$ there is a point $Q_n$ on the cycle $c_R$ on the Riemann surface, such that
\[  \int_{-A^{1/3}}^{Q_n} \omega_R = - n\beta + \frac{\tau}{2\pi i} \log 2 + 1/2  \quad \mod{\mathbb Z}. \]
If $Q_n$ happens to be on the first sheet, say $Q_n = (x_n, \xi_1(x_n)) \in \mathcal R_1$, then
\[ \theta\left(s_0 + \int_{-A^{1/3}}^{x_n} \omega_R + n\beta - \frac{\tau}{2\pi i} \log 2 - 1/2\right) = 0 \]
and so $M_{n,11}(x_n) = 0$ by formula \eqref{def-Mn11}. By symmetry we then also have $M_{n,11}(\omega^j x_n) = 0$ for $j=1,2$.
 It then follows from the asymptotic formula \eqref{Pnn-asymp} and Hurwitz's theorem from complex analysis 
that $P_{n,n}$ has a simple zero near each of $\omega^j x_n$, $j=0,1,2$, if $n \in \mathbb N_{\varepsilon}$ is large enough.
These are spurious zeros of $P_{n,n}$. The other zeros of $P_{n,n}$ are non-spurious zeros.

\begin{corollary}
Let $n \to \infty$ with $n \in \mathbb N_{\varepsilon}$.
Then the non-spurious zeros of $P_{n,n}$ tend to $\Sigma_1$ and
$\mu_1$ is the limit of the normalized zero counting measures.
\end{corollary}
\begin{proof}
This follows from the asymptotic formula \eqref{Pnn-asymp} and the above consideration on the spurious zeros.
Indeed, from \eqref{Pnn-asymp} and \eqref{def-g1} we have
\[ \lim_{n \to \infty \atop n \in \mathbb N_{\varepsilon}}  \frac{1}{n} \lim \log |P_{n,n}(z)| =  \int_{\Sigma_1} \log |z-s| d\mu_1(s), \]
almost everywhere in $\mathbb C \setminus \Sigma_1$, which by standard arguments from
logarithmic potential theory, see \cite{SaTo}, yields that the non-spurious zeros of $P_{n,n}$
tend to $\Sigma_1$ with $\mu_1$ as limiting distribution.
\end{proof}

\subsection{Remark on perturbation analysis around the critical regime} 

Although the focus of this paper is not on the critical regime, it is interesting
to see how the branch points $z_j$, $j=1,2,3$ approach the cusp point as $\Delta t = t-t_* > 0$ tends
to zero, and to compare this with  \cite{LTW3}. 

For the critical values $t=t_* = 1/8$, $A = A_* = 27/256$ the cubic  equation \eqref{def:speccurve} 
for the spectral curve
has the branch point $z = z_* = 3/4$, $\xi = \xi_* = 3/4$, which is a triple zero of the discriminant
of \eqref{def:speccurve}.

Let us  introduce  small $\Delta t > 0$ and define $t  = t_* + \Delta t$. Direct perturbation analysis
of \eqref{def:speccurve} indicates the following asymptotic behavior:
\begin{equation} 	\label{pert-values}
\begin{aligned}
	A & = A_* + \frac{9}{16} \Delta t - k (\Delta t)^{3/2} + O\left( (\Delta t)^2 \right), \cr
	z & = z_* + x (\Delta t)^{1/2} + O\left(\Delta t \right), \cr
	\xi & = \xi_* + x (\Delta t)^{1/2} + \frac{4}{3\sqrt{3}} y (\Delta t)^{3/4}  + O\left( (\Delta t)^2 \right),
\end{aligned}
\end{equation}
where the constants $x,y,k$ are to be defined.
Plugging \eqref{pert-values} into \eqref{def:speccurve} we find 
\[ \left(y^2 - x^3 - \frac{3}{2} x  - k\right) (\Delta t)^{3/2} + O\left((\Delta t)^{7/4}\right) =0, \]
which  means that the  equation
\begin{equation} \label{eq-cub}
y^2 = x^3 + \frac{3}{2} x + k  
\end{equation}
should hold true. The equation \eqref{eq-cub} defines an elliptic Riemann surface that is obtained from 
the blow up at the criticality. 

Recall that the choice of $A$ as a function of $t$ is dictated by the fact that $\xi dz$ has the Boutroux condition.
Then in the new $xy$- variables  it means that $y dx$ has the Boutroux condition, 
see Definition \ref{def-Boutroux},  on the surface \eqref{eq-cub}. This is a condition on $k$ in \eqref{eq-cub}. 
The elliptic Riemann surface has one real branch point $x_1$ and two non-real branch points $x_2$ and $x_3 = \overline{x}_2$
with $\Im x_2 > 0$.  There are non-trivial cycles $\alpha$ and $\beta$ analogous to the cycles
$\alpha_0$ and $\beta_0$ on $\mathcal R$, see Figure \ref{fig:cycles} below. The period $\oint_{\alpha} ydx$
is always purely imaginary. The mapping $k \mapsto \Re (\oint_{\beta} ydx)$ is
strictly increasing (this is analogous to Lemma \ref{lemma:hincreasing} below) and there is a unique
value of $k$ with $\Re (\oint_{\beta} ydx) = 0$. This value for $k$ is approximately  
\begin{equation} \label{eq:k}
k = 0.647 \cdots .
\end{equation}
After $k$ is determined, the  branch points $x_j$, $j=1,2,3$, can be calculated from the cubic equation 
\eqref{eq-cub} as they are the zeros of $x^3 + \frac{3}{2}x + k$. 
We find the  approximate values 
\begin{equation} \label{val-x}
x_1 \approx -0.391, \qquad x_2 \approx 0.196 + 1.27 i, \qquad x_3 \approx 0.196 - 1.27 i 
\end{equation}
which by \eqref{pert-values} determine approximations for $z_j$ for $j=1,2,3$.
In particular
\begin{equation} \label{z1super}
	z_1 \approx z_* - 0.391 (\Delta t)^{1/2}  \qquad \text{ as } \Delta t \to 0+.
\end{equation}

Recall from \cite[formula (2.19)]{BK} that for $t < t_*$ one
has $z_1 = \frac{3}{4} \left(1- \sqrt{1-8t}\right)^{2/3}$
and so if $\Delta t = t-t_* < 0$,
\begin{equation} \label{z1sub} 
	z_1 = z_* - \sqrt{2} (-\Delta t)^{1/2} + O\left( \Delta t\right) \qquad \text{ as } \Delta t \to 0-.
	\end{equation}

In the paper \cite{LTW3} of Lee, Teodorescu and Wiegmann the following equations are given for the motion
of the fingertip (formula (28) in \cite{LTW3})
\begin{equation} \label{LTWtip} 
	e(T) = \begin{cases} -2 \sqrt{-T}, & T < 0, \\
		- 0.553594  \sqrt{T}, & T > 0.
		\end{cases} \end{equation}
To  compare the results one should identify $e(T) = z_1 - z_*$ and $T = \frac{1}{2} \Delta t$. 
Then $e(T) = - 2\sqrt{-T}$ corresponds to the leading behavior in \eqref{z1sub} for $\Delta t < 0$,
and for $T > 0$,
\[ e(T) = -0.553594 \sqrt{T} = - \frac{0.553594}{\sqrt{2}} (\Delta t)^{1/2} = 0.39145 (\Delta t)^{1/2},\]
which corresponds to \eqref{z1super}.

It is also of interest to consider the intersection point $\widehat{z}$ of $\partial \Omega$ with the positive real line.
If we substitute  $t=t_*+\Delta t$, $\widehat{z}=z_*+ x (\Delta t)^{1/2} + O(\Delta t)$,
and the expression for $A$ from \eqref{pert-values}
into \eqref{dOmegasuper}, then we find
\[ -\left(x^3 + \frac{3}{2} x + k\right) (\Delta t)^{3/2} + O(\Delta t)^2 = 0. \]
This means that $x=x_1$, where $x_1$ is the real solution of $x^3 + \frac{3}{2} x + k = 0$.  
Thus, the distance from $\widehat{z}$ to $z_1$ is only $O(\Delta t)$ as $\Delta t \to 0+$,
and this is much smaller than the distance from $z_1$ to the other branch points $z_2$ and $z_3$, which is $O(\Delta t)^{1/2}$. 
Therefore, in the leading order $O(\Delta t)^{1/2}$, the branch point $z_1$ is indistinguishable from 
$\widehat{z}$. This effect is clearly visible in the left panel of Figure \ref{fig:whiskers}.
		
\section{Proof of Theorem \ref{theorem1}} \label{section:proof1}

\subsection{The Riemann surface}

We start by investigating the Riemann surface that is associated with the cubic
equation \eqref{def:speccurve}. 
The discriminant of \eqref{def:speccurve} with respect to the variable $\xi$
is a polynomial in $z$ of degree $9$,
which because of the $\mathbb Z_3$ symmetry takes the form 
\[ D(P)(z) = Q(z^3) \]
with a cubic polynomial $Q(w)$ (these and other calculations were made 
with the help of Maple)
\begin{equation} \label{def:Q} 	
	Q(w) = 4w^3 + (t^2 + 20t + 4A -8) w^2  
		+ (4 (1+t)^3 + 18 At -36 A) w - 27 A^2. 
		\end{equation}
		
The discriminant of \eqref{def:Q} with respect to $w$ has the remarkable factorization 
\begin{equation} 
	\label{def:Deltaq} 
		D(Q) = 16 (t^2 - 7 t - 8 + 3A)^3  (t (1+t)^3 - 20 A t + 16 A^2 - A),
		\end{equation}
which is zero for $A = A_j(t)$, $j = 1,2,3$, where $A_1(t)$ is given by \eqref{def:A1}, 
\begin{align} 
	\label{def:A2}
	A_2(t) = \frac{1}{32} \left(1 + 20 t - 8 t^2 + (1-8 t)^{3/2} \right), \quad t < t_*,
\end{align}
and
\begin{align}
  A_3(t) & = \frac{1}{3} (1+t)(8-t). \label{def:A3}
\end{align}

The cubic polynomial \eqref{def:Q} has three zeros $w_1, w_2, w_3$, and at least one of
them, say $w_1$, is real. If $A \neq 0$ then $w_1 > 0$. 
Then $D(P)(z)$ has nine zeros, namely
\[ z_1 = w_1^{1/3}, \quad z_2 = w_2^{1/3}, \quad z_3 = w_3^{1/3}, \]
and their rotations $\omega z_j$, $\omega^2 z_j$, for $j=1,2,3$,
where $z_1 > 0$ is real, and $-\frac{\pi}{3} \leq \arg z_3 \leq 0 \leq \arg z_2  \leq \frac{\pi}{3}$.

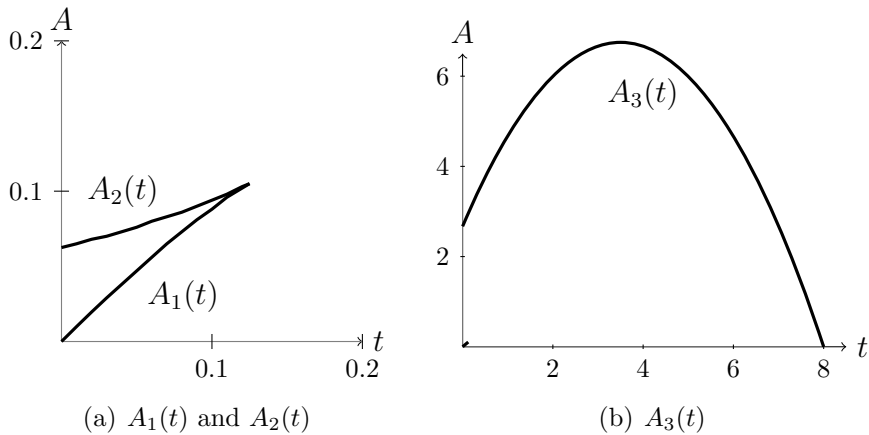
\begin{figure}[t]
\centering
\subfigure[$A_1(t)$ and $A_2(t)$]{
\begin{tikzpicture}[scale=20]
\draw[->](0,0)--(0.2,0) node[right]{$t$};
\draw[->](0,0)--(0,0.2) node[above]{$A$};
\draw[help lines] (0,0.2)--(0,0)--(0.2,0);
\draw[very thick] (0,0)--(0.01,0.0099)--(0.02,0.0196)--(0.03,0.029)--(0.04,0.038)--(0.05,0.047)--(0.06,0.056)--(0.07,0.065)--(0.08,0.073)--(0.09,0.081)--(0.10,0.088)--(0.11,0.096)--(0.12,0.102)--(0.125,0.105);
\draw[very thick] (0,0.0625)--(0.01,0.065)--(0.02,0.068)--(0.03,0.070)--(0.04,0.073)--(0.05,0.076)--(0.06,0.080)--(0.07,0.083)--(0.08,0.086)--(0.09,0.090)--(0.10,0.094)--(0.11,0.098)--(0.12,0.103)--(0.125,0.105);
\draw (0.05,0.03)  node[right]{$A_1(t)$};
\draw (0.01,0.1) node[right]{$A_2(t)$};
\draw (0.1,0.005)--(0.1,-0.005) node[font=\footnotesize,below]{$0.1$};
\draw (0.2,0.005)--(0.2,-0.005) node[font=\footnotesize,below]{$0.2$};
\draw (0.005,0.1)--(-0.005,0.1) node[font=\footnotesize,left]{$0.1$};
\draw (0.005,0.2)--(-0.005,0.2) node[font=\footnotesize,left]{$0.2$};
\end{tikzpicture}
}
\subfigure[$A_3(t)$]{
\begin{tikzpicture}[scale=0.6]
\draw[->](0,0)--(8.5,0) node[right]{$t$};
\draw[->](0,0)--(0,6.5) node[above]{$A$};
\draw[help lines] (0,6)--(0,0)--(8,0);
\draw[very thick] (0,2.67) parabola bend (3.5,6.75)  (8,0);
\draw[very thick] (0,0)--(0.125, 0.105);
\draw (4,5)  node[above]{$A_3(t)$};
\draw (2,0.05)--(2,-0.05) node[font=\footnotesize,below]{$2$};
\draw (4,0.05)--(4,-0.05) node[font=\footnotesize,below]{$4$};
\draw (6,0.05)--(6,-0.05) node[font=\footnotesize,below]{$6$};
\draw (8,0.05)--(8,-0.05) node[font=\footnotesize,below]{$8$};
\draw (0.05,2)--(-0.05,2) node[font=\footnotesize,left]{$2$};
\draw (0.05,4)--(-0.05,4) node[font=\footnotesize,left]{$4$};
\draw (0.05,6)--(-0.05,6) node[font=\footnotesize,left]{$6$};
\end{tikzpicture}
}
\caption{Graphs of $A_1(t)$, $A_2(t)$ and $A_3(t)$. Note the difference
in scale between the two plots.} 
\end{figure}

If $A = A_j(t)$ for some $j=1,2,3$, then $Q(w)$ has one simple and one double zero.
This gives three simple zeros and three  double zeros for the discriminant of $P$.
In case $A= A_1(t)$ or $A=A_2(t)$ we have that the double zeros are nodes 
and there is no branching at these points. Thus, keeping
in mind the branch point at infinity, we have only four branch points, which by the
Riemann-Hurwitz formula, see e.g.\ \cite{FK,Mir}, gives that the genus is zero.

For $A = A_3(t)$ the double zeros are branch points located on the
rays with angles  $\pm \pi/3$ and $\pi$, that connect all three sheets.
Then the genus is three.

Finally, if $A \in \mathbb R^+ \setminus \{ A_1(t), A_2(t), A_3(t)\}$ then all zeros are simple zeros 
of the discriminant of $P$. Being simple
zeros, these give rise to nine branch points of the Riemann surface. Taking
note that there is also branching at infinity (between sheets $\mathcal R_2$ and $\mathcal R_3$)
we have in total 10 branch points on the Riemann surface when viewed as a three-fold cover 
of the $z$-plane. The Riemann-Hurwitz formula then tells us that the genus is three.

At the critical value $t = t_*$, $A = A_* = A_1(t_*)$, we have triple zeros of the discriminant at $z_1$, $\omega z_1$, $\omega^2 z_1$
for some $z_1 > 0$.
The sheet structure of the Riemann surface is then given by 
$\mathcal R_1 = \mathbb C \setminus \Sigma_1$, $\mathcal R_2 = \mathbb C \setminus (\Sigma_1 \cup \Sigma_2)$,
$\mathcal R_3 = \mathbb C \setminus \Sigma_3$ where
$\Sigma_1 = \bigcup_j [0, \omega^j z_1]$ and $\Sigma_3 = \{ z \in \mathbb C \mid z^3 \in \mathbb R^-\}$.
The sheets $\mathcal R_1$ and $\mathcal R_2$ are connected via $\Sigma_1$,
and $\mathcal R_2$ and $\mathcal R_3$ are connected via $\Sigma_2$.

If we increase $t$ and move in parameter space from $(t_*, A_*)$ to $(t_*+\delta, A_*)$
then the triple zero $z_1$ of the discriminant splits into three zeros $z_1, z_2, z_3$, where $z_1$ is real
and $z_2$ and $z_3$ are non-real and each other complex conjugates. We take
$\Im z_2 > 0$.
Then by continuity we find the sheet structure of the Riemann surface as shown in Figure \ref{fig:three-sheets}.

Also by continuity the same sheet structure will hold throughout the region in the  first quadrant of the $(t,A)$  plane
that is bounded by the three curves $A = A_j(t)$, $j=1,2,3$, and $A > 0$. This region consists of the three pieces
\begin{enumerate}
\item[(1)] $0 < t \leq t_*$, $0 < A < A_1(t)$,
\item[(2)] $0 < t \leq t_*$, $A_2(t) < A < A_3(t)$,
\item[(3)] $t_* < t < t_{**}$, $0 < A < A_3(t)$.
\end{enumerate} 

For each such $(t,A)$  we have three branch points $z_1, z_2, z_3$  in sector $S_0$ with
$z_1 > 0$, $\Im z_2 > 0$ and $z_3 = \overline{z}_2$. There are six other branch points $\omega^j z_k$, $j=1,2$,
$k=1,2,3$, and we use the sheet structure of $\mathcal R$ as in Figure \ref{fig:three-sheets}.

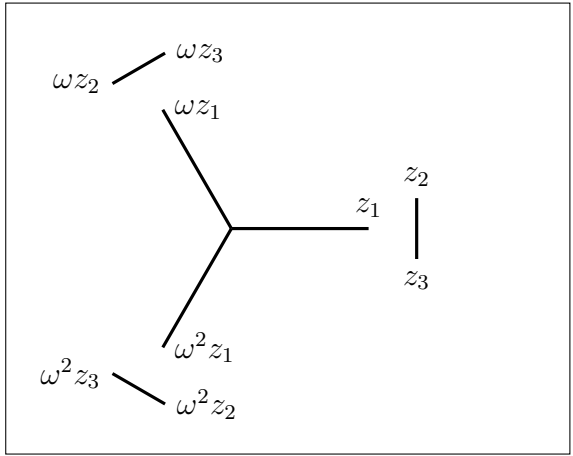
\begin{figure}[!t]
\centering
\begin{tikzpicture}[scale=3]
\draw[very thick] (0,0)--(0.608,0) node[above]{$z_1$};
\draw[very thick] (0.821,0.135)--(0.821,-0.135) node[below]{$z_3$};
\draw[very thick] (0.821,0.135) node[above]{$z_2$};

\draw[very thick] (-0.527,0.643)--(-0.294,0.778) node[right]{$\omega z_3$};
\draw[very thick] (-0.527,0.643) node[left]{$\omega z_2$};
\draw[very thick] (0,0)--(-0.304,0.527) node[right]{$\omega z_1$};
\draw[very thick] (-0.527,-0.643)--(-0.294,-0.778) node[right]{$\omega^2 z_2$};
\draw[very thick] (0,0)--(-0.304,-0.527) node[right]{$\omega^2 z_1$};
\draw[very thick] (-0.527,-0.643) node[left]{$\omega^2 z_3$};
\draw  (-1.0,-1.0) rectangle (1.5,1.0);
\end{tikzpicture}
  

\caption{The first sheet $\mathcal R_1$ after deformation of the whiskers $\Sigma_1^w$ to 
arcs joining $\omega^j z_2$ to $\omega^j z_3$ disjoint from $\omega^j z_1$, for $j=0,1,2$. \label{fig:sheet1}}

\end{figure}

\subsection{Proof of Theorem \ref{theorem1} (a)}

We take $t > t_*$ and $0 < A < A_3(t)$.  From the previous section we then know
the sheet structure of the Riemann surface. The surface has genus three, and we  
take the following  canonical homology  basis 
\begin{equation} \label{homologybasis}
		\{ \alpha_0, \beta_0, \alpha_1, \beta_1, \alpha_2, \beta_2  \} 
		\end{equation}
where $\alpha_0$ is a nontrivial cycle going around $\Sigma_1^w \cap S_0$ on the first sheet with counterclockwise
orientation,
$\beta_0$ is a cycle from $z_1$ to $z_2$ on the first sheet, lying to the right of $\Sigma_1^w$, 
and back from $z_2$ to $z_1$ on the second sheet. 
The other cycles are obtained by rotation over angles $2\pi/3$ and $4\pi/3$, i.e.,
\[ \alpha_j = \omega^j \alpha_0, \qquad \beta_j = \omega^j \beta_0, \qquad j=1,2. \]

To visualize these cycles it is convenient to first deform the whiskers $\Sigma_1^w$ to cuts that connect
the branch points $\omega^j z_2$ and $\omega^j z_3$ and are disjoint from $\omega^j z_1$, as in 
Figure \ref{fig:sheet1}.
Then homotopic versions of the cycles \eqref{homologybasis} are shown in Figure \ref{fig:cycles}.

\begin{figure}[!t]
\centering
\begin{tikzpicture}[scale=3]
\draw[thick] (0,0)--(0.608,0) node[below]{$z_1$};
\draw[thick] (0.821,0.135)--(0.821,-0.135) node[below]{$z_3$};
\draw[thick] (0.821,0.135) node[above]{$z_2$};

\draw[thick] (-0.527,0.643)--(-0.294,0.778); 
\draw[thick] (-0.527,0.643); 
\draw[thick] (0,0)--(-0.304,0.527); 
\draw[thick] (-0.527,-0.643)--(-0.294,-0.778);  
\draw[thick] (0,0)--(-0.304,-0.527); 
\draw[thick] (-0.527,-0.643); 
\draw  (-1.0,-1.0) rectangle (1.5,1.0);

\draw[ultra thick] (0.821,0) ellipse (2pt and 15pt);
\draw (1.0,0.14) node[above]{$\alpha_0$};
\draw[ultra thick] (0.608,0)--(0.821,0.135);
\draw (0.608,0) node[above]{$\beta_0$};

\draw (-0.75,0.4) node[right]{$\alpha_1$};
\draw (-0.4,0.57) node[below]{$\beta_1$};
\draw (-0.7,-0.45) node[right]{$\alpha_2$};
\draw (-0.2, -0.5) node[below]{$\beta_2$};

\draw[ultra thick,rotate=120] (0.821,0) ellipse (2pt and 15pt);
\draw[ultra thick,rotate=-120] (0.821,0) ellipse (2pt and 15pt);

\draw[ultra thick,rotate=120] (0.608,0)--(0.821,0.135);
\draw[ultra thick,rotate=-120] (0.608,0)--(0.821,0.135);
\end{tikzpicture}
\caption{Cycles $\alpha_j$, $\beta_j$, for $j=0,1,2$ on the first sheet of the Riemann surface
after deformation of the cuts. The $\beta$ cycles
also have a part on the second sheet. The $\alpha$ cycles are oriented counterclockwise and
the cycle $\beta_j$ is oriented from $\omega^j z_1$ to $\omega^j z_2$ on the first sheet  \label{fig:cycles} }  
\end{figure}
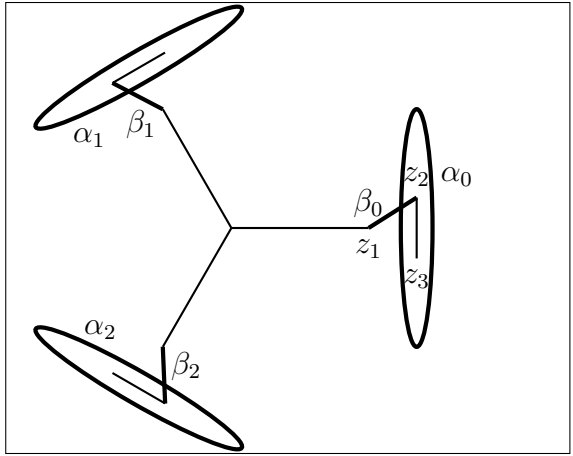

\begin{lemma}  \label{lemma:Boutrouxbeta}
$\xi dz$ has the Boutroux condition if and only if
\begin{equation} \label{eq:Boutrouxbeta0} 
	\Re \left( \oint_{\beta_0} \xi dz \right) = 0. 
	\end{equation}
\end{lemma}
\begin{proof}
First observe that the residues of $\xi dz$ in the poles (the two points at infinity) are real;
in fact they are $\pm t$, as can be deduced from \eqref{eq:xiatinfinity}.

Because of symmetry $\xi_1(\overline{z}) = \overline{\xi_1(z)}$ in the real axis,
we have 
\begin{equation} \label{alphasymmetry1} 
	\oint_{\alpha_0} \xi dz = \oint_{\alpha_0} \xi_1(z) dz   \in i \mathbb R. 
	\end{equation}
By the rotational $\mathbb Z_3$ symmetry $\xi_1(\omega z) = \omega^2 \xi_1(z)$, we then also have
\begin{equation} \label{alphasymmetry2}  
	\oint_{\alpha_j} \xi dz = \oint_{\alpha_0} \xi dz \in i \mathbb R, \qquad j=1,2.  
	\end{equation}
Similarly,
\begin{equation} \label{betasymmetry}  \oint_{\beta_j} \xi dz = \oint_{\beta_0} \xi dz, \qquad j=1,2. 
\end{equation}
and therefore we have the Boutroux condition if and only if \eqref{eq:Boutrouxbeta0} 
is satisfied.
\end{proof}

We define for $t \geq t_*$, $0 < A \leq A_3(t)$, 
\begin{equation} \label{def:h} 
	h(t,A) = \Re \left( \oint_{\beta_0} \xi dz \right) = \Re \left( \int_{z_1}^{z_2} \left(\xi_1(z) - \xi_2(z) \right) dz \right).
	\end{equation}
The second identity in \eqref{def:h} comes from the definition of $\beta_0$ as a path
from $z_1$ to $z_2$ on the first sheet and back from $z_2$ to $z_1$ on the second sheet. Note that
\begin{equation} \label{htstar} 
	h(t_*, A_1(t_*)) = 0,
	\end{equation}
since if we approach the critical values $t=t_*$, $A = A_1(t_*)$, the cycle $\beta_0$ shrinks to a point.

\begin{lemma} \label{lemma:hnotzero} 
For each $t\in [t_*, 8]$ 	we have 
\[ h(t, A_3(t)) > 0. \]
\end{lemma}
\begin{proof}
For $A = A_3(t)$ it is easy to calculate from \eqref{def:Q} and \eqref{def:A3} that
the three zeros of $Q$  are $w_1 = \frac{1}{12}(t-8)^2$ and $w_2 = w_3 = -3(1+t)$. 
Thus
\begin{equation} \label{eq:z2forA3}
\begin{aligned} 
	 z_1 & = 12^{-1/3} (8-t)^{2/3}, \\
	z_2 & = 3^{1/3} (1+t)^{1/3} e^{\pi i/3}, \\
	z_3 & = 3^{1/3} (1+t)^{1/3} e^{-\pi i/3} = \omega^2 z_2.
	\end{aligned} \end{equation}
Then $z_2$, $\omega z_2 = -|z_2|$ and $\omega^2 z_2 = z_3$ are double branch points 
that connect all three sheets, and the three values $\xi_1$, $\xi_2$, $\xi_3$ coincide 
for these values of $z$. 
From the spectral curve equation \eqref{def:speccurve} we have
\[ \xi_1 + \xi_2 + \xi_3 = z^2 \]
and it follows that 
\begin{align}  \label{eq:xi123equal}
	\xi_1(-|z_2|) = \xi_2(-|z_2|) = \xi_3(-|z_2|) = \frac{1}{3} |z_2|^2 
	= 3^{-1/3} (1+t)^{2/3}, 
	\end{align}
see also \eqref{eq:z2forA3}.

Now we calculate \eqref{def:h} by integrating $\xi_1 - \xi_2$ from $z_1$ to $0$
and then from $0$ to $z_2$ along the ray $\arg z = \pi/3$. The integral from
$z_1$ to $0$ does not contribute to the real part, since $\xi_1$ and $\xi_2$
are complex conjugates of  each other there. What remains is
\begin{align} \nonumber
		h(t,A_3(t)) & = \Re \left( \int_0^{z_2} ( \xi_1(z) - \xi_{2,-}(z)) dz \right) \\
 	 & = \Re \left( e^{\pi i/3} \int_0^{|z_2|} [\xi_1(r e^{\pi i/3}) - \xi_{2,-}(r e^{\pi i/3})] dr \right)
	\label{eq:htA2} 
	\end{align}
	where we put $z = r e^{\pi i/3}$.
	The value $\xi_{2,-}$ in \eqref{eq:htA2} denotes the limit of $\xi_2$ as we approach $\arg z = \pi/3$ from 
	the sector $S_0$.
This is the limiting value from the right if we orient $\arg z = \pi/3$ from $0$ to $\infty$.

By the symmetry $\xi_j(z) = \omega \xi_j( \omega z)$, we have $e^{\pi i/3} \xi_1(re^{\pi i/3}) = -\xi_1(-r)$,
$e^{\pi i/3} \xi_{2,-}(r e^{\pi i/3}) = -\xi_{2,+}(-r)$, where $\xi_{2,+}$ denotes the limit from the upper
half plane. Thus  by putting $s = -r$ in \eqref{eq:htA2} we obtain
\begin{align} 
	h(t,A_3(t)) = \Re  \left( \int_{-|z_2|}^0 (\xi_{2,+}(s) - \xi_1(s)) ds  \right)
	\label{eq:htA3} \end{align}
where the negative real axis is oriented from left to right.
For $s < 0$ we have that $\xi_1(s)$ is real and $\xi_{2,+}(s)$ and $\xi_{3,+}(s)$
are complex conjugate.
Since $\xi_1 + \xi_2 + \xi_3 = z^2$, we have $\Re \xi_{2,+}(s) = \frac{1}{2}(s^2 - \xi_1(s))$ for $s < 0$,
and so by \eqref{eq:htA3}
\begin{align}
	h(t,A_3(t)) = \frac{1}{2} \int_{-|z_2|}^0 (s^2 - 3\xi_1(s) ) ds 
	\label{eq:htA4}
\end{align}

Putting $z=s < 0$, $\xi = \frac{1}{3} s^2$, $A = A_3(t)$, in the spectral curve equation \eqref{def:speccurve}
we obtain
\[ P(\tfrac{1}{3} s^2, s) = - \frac{1}{27}(2s^3 + 3t-24)(s^3 + 3t + 3) \]
which, since $t \leq 8$, has exactly one zero for $s < 0$, namely
\[ s = - 3^{1/3}(1+t)^{1/3} = - |z_2|, \]
see \eqref{eq:z2forA3}.
It follows that
\begin{equation} \label{eq:xi1equality} 
	\xi_1(s) = \frac{1}{3} s^2  \qquad \text{if and only if } s = - |z_2|. 
	\end{equation}
[The equality $\xi(s) = \frac{1}{3} s^2$ for $s= - |z_2|$ is also immediate
from  \eqref{eq:z2forA3} and \eqref{eq:xi123equal}.]
As $s \to 0-$ we have $\xi_1(s) \to - A^{1/3}$ (this follows from \eqref{def:speccurve})
and so $\xi_1(s) < \frac{1}{3} s^2$ for $s$ close to $0$. From \eqref{eq:xi1equality} we then
get
\[ 	\xi_1(s) < \frac{1}{3} s^2  \qquad \text{for } s \in (-|z_2|, 0) \] 
and then \eqref{eq:htA4} tells us that $h(t,A_3(t)) > 0$, which proves the lemma.
\end{proof}

\begin{lemma} \label{lemma:hincreasing} 
For all $t \geq t_*$ and $0 < A < A_3(t)$	we have
\begin{equation} \label{eq:partialhA} 
	\frac{\partial h}{\partial A} > 0. 
	\end{equation}
\end{lemma}

\begin{proof} Note that by \eqref{def:speccurve}
\begin{equation} \label{eq:partialxiA} 
	\frac{\partial \xi}{\partial A} =  - \frac{ \frac{\partial P}{\partial A}}{\frac{\partial P}{\partial \xi}} =
	 - \frac{1}{3\xi^2 - 2z^2 \xi - (1+t) z} 
	\end{equation}
is a meromorphic function on the Riemann surface with simple poles at the branch points 
$\omega^j z_k$, $j,k=1,2,3$, since
these are the points where $\frac{\partial P}{\partial \xi} = 0$. 
The meromorphic differential $dz$ has a zero at the branch points, and therefore
$\frac{\partial \xi}{\partial A} dz$ is a meromorphic differential whose only possible
poles are at the points at infinity.

Because of \eqref{eq:xiatinfinity} and \eqref{eq:partialxiA} we have
\begin{align*} 
	\frac{\partial \xi_1}{\partial A}(z) & = O(z^{-4}) && \text{ as } z \to \infty, \\
  \frac{\partial \xi_2}{\partial A}(z) & = O(z^{-5/2}) && \text{ as } z \to \infty, \, z \in S_0,
	\end{align*}
Then it easily follows that the singularities at the points at infinity are removable, 
and so $\frac{\partial \xi}{\partial A} dz$ is a holomorphic differential.
It is
a multiple of the holomorphic differential $\omega_R$ given in \eqref{omegaR}.
The holomorphic differential has a double zero at both points at infinity, and these are the only
zeros of $\frac{\partial \xi}{\partial A} dz$, since the genus is $3$.

Because of the symmetries \eqref{alphasymmetry2}--\eqref{betasymmetry}, we have  
for $j=1,2$,
\begin{equation} \label{eq:Riemann} \begin{aligned} 
	\oint_{\alpha_j} \frac{\partial \xi}{\partial A} dz & = 
	\oint_{\alpha_0} \frac{\partial \xi}{\partial A}dz \in i \mathbb R,  \\
	\oint_{\beta_j} \frac{\partial \xi}{\partial A} dz 
	& = \oint_{\beta_0} \frac{\partial \xi}{\partial A}dz.
	\end{aligned} \end{equation}
Not all periods of a non-zero holomorphic differential can be 
purely imaginary, see e.g.\ \cite[Proposition III.3.3]{FK}. 
Hence from \eqref{eq:Riemann} and \eqref{def:h} we conclude that 
\[ \frac{\partial h}{\partial A} = \Re 	\left( \oint_{\beta_0} \frac{\partial \xi}{\partial A} dz \right) \neq 0. \]

Then by continuity in the parameters $t$ and $A$, we either have
$\frac{\partial h}{\partial A} > 0$,  or $\frac{\partial h}{\partial A} < 0$,
for all $t,A$ with $t \geq t_*$ and $0 < A < A_3(t)$.  
Since $h(t,A_3(t)) > 0$ by Lemma \eqref{lemma:hnotzero} and $h(t_*, A_1(t_*)) = 0$ by \eqref{htstar},  we have the 
first possibility, and the lemma is proved.
\end{proof}

After these preparations it is easy to prove part (a) of Theorem \ref{theorem1}.

\begin{proof}[Proof of Theorem \ref{theorem1} (a)]

Since $A \mapsto h(t_*,A)$ is increasing because of Lemma \ref{lemma:hincreasing},
we have $h(t_*,0) < h(t_*,A_1(t_*)) = 0$, see \eqref{htstar}.
By Lemma \ref{lemma:hnotzero}, we have  $h(8,0) >  0$ and so by continuity there is $t_{**} \in (t_*, 8)$
such that $h(t_{**}, 0) = 0$ and $h(t,0) < 0$ for every $t \in [t_{*}, t_{**})$. 
Let $t \in (t_*, t_{**})$. Since $h(t,0) < 0$ and $h(t,A_3(t)) > 0$ (see Lemma \ref{lemma:hincreasing})
there is a value $A = A(t) \in (0, A_3(t))$ such that $h(t,A(t))=0$. This value for $A$ is unique because
of Lemma \ref{lemma:hincreasing}, and so $t \mapsto A(t)$ is continuous with $A(t_{**}) = 0$. 
For this value of $A=A(t)$ we have the Boutroux condition by Lemma \ref{lemma:Boutrouxbeta}
and the definition \eqref{def:h} of $h(t,A)$. 
The fact that $h(t_{**}, 0) = 0$ implies \eqref{boundaryA}. This proves part (a).
\end{proof}

\subsection{Proof of Theorem \ref{theorem1} (b)}

Let $t \in (t_*, t_{**})$ and $A = A(t)$.
We define
\begin{equation} \label{def:H}
	H(z) = \Re \int_{z_1}^z (\xi_1(s)  - \xi_2(s)) ds, \qquad z \in S_0 \setminus [0,z_1] 
	\end{equation}
with a path of integration in $S_0 \setminus [0,z_1]$. 
Because of the Boutroux condition \eqref{cond:Boutroux} $H(z)$ is well-defined, and it is independent of the path
from $z_1$ to $z$. Indeed, if we take two paths $\gamma_1$ and $\gamma_2$ in $S_0$ with corresponding values $H_1$ and $H_2$ then
\[ H_1(z) - H_2(z) =  \Re \int_{\gamma_2^{-1} \circ \gamma_1} (\xi_1(s) - \xi_2(s)) ds  \]
which can be identified as the real part of $\oint_{\gamma} \xi ds$ for a closed curve $\gamma$
on the Riemann surface and the real part is zero because of \eqref{cond:Boutroux}.

Then $H$ is a well-defined harmonic function on $S_0 \setminus [0,z_1]$ and it extends to
a continuous function on $\overline{S}_0$. Its level sets $H(z) = c$ are the trajectories
of the quadratic differential $- (\xi_1-\xi_2)^2 ds^2$, see \cite{Str}.  
Also by Lemma \eqref{lemma:Boutrouxbeta} and \eqref{def:h}, \eqref{def:H} we have $H(z_2) = 0$.
By the local theory of quadratric differential near a simple zero \cite{Str},
there are three trajectories emanating from $z_2$ that are on the zero level set $H(z) = 0$.
Similarly, $H$ is zero on three trajectories from $z_3$.

\begin{lemma} \label{lemma:Honray}
$H$ has exactly one zero on the half-ray $\arg z = \pi/3$.
\end{lemma}

\begin{proof}
For $z = x e^{\pi i/3}$, $x > 0$, we integrate from $z_1$ to $z$ by first going from
$z_1$ to $0$ and then from $0$ to $z$ along the half ray of angle $\pi/3$. 
Then as in the proof of Lemma \ref{lemma:hnotzero}  we find from \eqref{def:H} that
\begin{align} \nonumber 
	H(x e^{\pi i/3}) & = \Re \left( \int_0^{z} (\xi_1(s) - \xi_{2,-}(s)) ds \right) \\
		& = \frac{1}{2} \int_{-x}^0  \left(s^2 - 3\xi_1(s) \right) ds, 
		\label{eq:Hintegral}
		\end{align}
see, in particular, \eqref{eq:htA4}.

The function $s^2 - 3 \xi_1(s)$ satisfies  
\begin{equation} \label{eq:dHbehavior}
\begin{aligned} s^2 - 3 \xi_1(s) & = - 2s^2 + O(s^{-1}), &&  \text{ as } s \to -\infty, \\
	s^2 - 3 \xi_1(s) & \to 3 A^{1/3} >0 &&  \text{ as } s \to 0-, 
	\end{aligned}
	\end{equation}.
	Hence it changes sign at least once on the negative real axis.
Suppose $s^*$ is a zero of $s^2 - 3 \xi_1(s)$ on the negative real axis. Then,
inserting $\xi_1(s^*) = \frac{1}{3} (s^*)^2$ into the cubic equation \eqref{def:speccurve}, we find
that $w^* = (s^*)^3$ is a zero of 
\[ w^2 - \frac{9(2-t)}{2} w - \frac{27}{2} A.  \]
This quadratic polynomial has one positive zero and one negative zero (since $A >0$).
Thus $w^*$ is the unique negative root, and then $s^*$ is unique as
the negative real solution of  $s^2-3\xi_1(s)$ and together with \eqref{eq:dHbehavior}
we find  
\[ s^2 - 3 \xi_1(s) \begin{cases} < 0 & \text{ for } s < s^* \\
	   > 0 & \text{ for } s^* < s < 0. \end{cases} \]
Then by \eqref{eq:Hintegral} we find that $x \mapsto H(x e^{\pi i/3})$ is strictly increasing for $0 < x < -s^*$ and 
strictly decreasing for $x > -s^*$. Then there is exactly one zero on the half ray $\arg z = \pi/3$, since
\begin{align*}
	H(x e^{\pi i/3}) & = - \frac{1}{3} x^3 + O(\log x),  \quad \text{ as } x \to +\infty, \\
	H(x e^{\pi i/3}) & \to 0, \qquad \text{ as } x \to 0+,
	\end{align*}
	see \eqref{eq:Hintegral} and \eqref{eq:dHbehavior}. 
\end{proof}

\begin{lemma} \label{lemma:Hatinfinity} As $z \to \infty$ in $S_0$, 
\[ H(z) = \Re \left( \frac{1}{3} z^3 - \frac{2}{3} z^{3/2} + O(\log|z|) \right). \]
\end{lemma}
\begin{proof}
This is immediate from \eqref{eq:xiatinfinity} and \eqref{def:H}.
\end{proof}

From Lemma \ref{lemma:Hatinfinity} it follows that there are two unbounded branches of the zero level set 
$H(z) = 0$
in $S_0$.  They tend to infinity at asymptotic angles $\pm \pi/6$,
respectively.

\begin{lemma} \label{lemma:Hforxreal}
For $x$ real and $x > z_1$ we have $H(x) > 0$.
\end{lemma}
\begin{proof}
There is no solution of $\xi_1(s) = \xi_2(s)$ for $s > z_1$, since such a solution
would show up as a zero of the discriminant of $P$, and in the supercritical case
the zeros $z_2$ and $z_3$ are not real.  
Since  $\xi_{1,2}(s)\in \mathbb R$ when $s > z_1$ we conclude that $\xi_1(s) - \xi_2(s)$ has a constant
sign for $s > z_1$. Since $\xi_1(s) - \xi_2(s)  = s^2 + O(s^{1/2})$ as $s \to +\infty$,
see \eqref{eq:xiatinfinity}, the sign is positive. 
Then $H(x) > 0$ for $x > z_1$ by the definition \eqref{def:H} of $H$.
\end{proof}

With the help of these lemmas we can now prove part (b) of Theorem \ref{theorem1}.

\begin{proof}[Proof of Theorem \ref{theorem1} (b)]
We already noted that the trajectories of the quadratic differential $-(\xi_1-\xi_2)^2 ds^2$ that emanate
from $z_2$ are contained in the level set $H(z) = 0$ of $H$.
There are three such trajectories, and we follow them in the sector
\[ S_0^+ = \{ z \in \mathbb C \mid 0 < \arg z < \pi/3 \}. \]

The trajectories cannot exit $S_0^+$ along $(z_1, +\infty)$ by Lemma \ref{lemma:Hforxreal}.
They cannot go to a point in $[0, z_1)$ either, since the interval $[0,z_1]$ is a trajectory as well,
and trajectories do not intersect, except possibly at zeros of the quadratic differential.
Thus, if one of the trajectories from $z_2$ comes to the real axis, it will come to $z_1$.

There is at most one trajectory that comes to the ray $\arg z = \pi/3$, because of Lemma~\ref{lemma:Honray}.
Any unbounded trajectory that stays inside $S_0$ has to go to infinity with
asymptotic angle $\pi/6$. This follows from Lemma \ref{lemma:Hatinfinity}. 
There can be at most one such trajectory from $z_2$.

Combining all this we see that 
the only possible topology of the trajectories emanating from $z_2$ (that is, of
zero level sets of $H(z)$) is:
one of the trajectories from $z_2$ has to come to $z_1$, the other one intersects the 
ray $\arg z=\frac \pi 3$ and the third one goes to infinity, see Figure \ref{levelcurves}. 
The trajectory connecting  $z_1$ and $z_2$ forms the analytic arc $\gamma_{1,2}$, which
proves part (b) of Theorem \ref{theorem1}.
\end{proof}

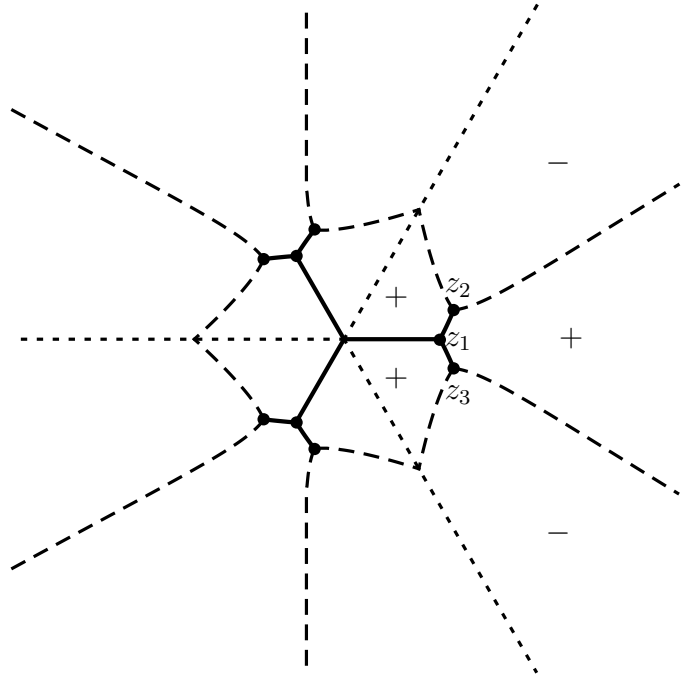
\begin{figure}[t] 
\centering
\begin{overpic}[scale=1.0]{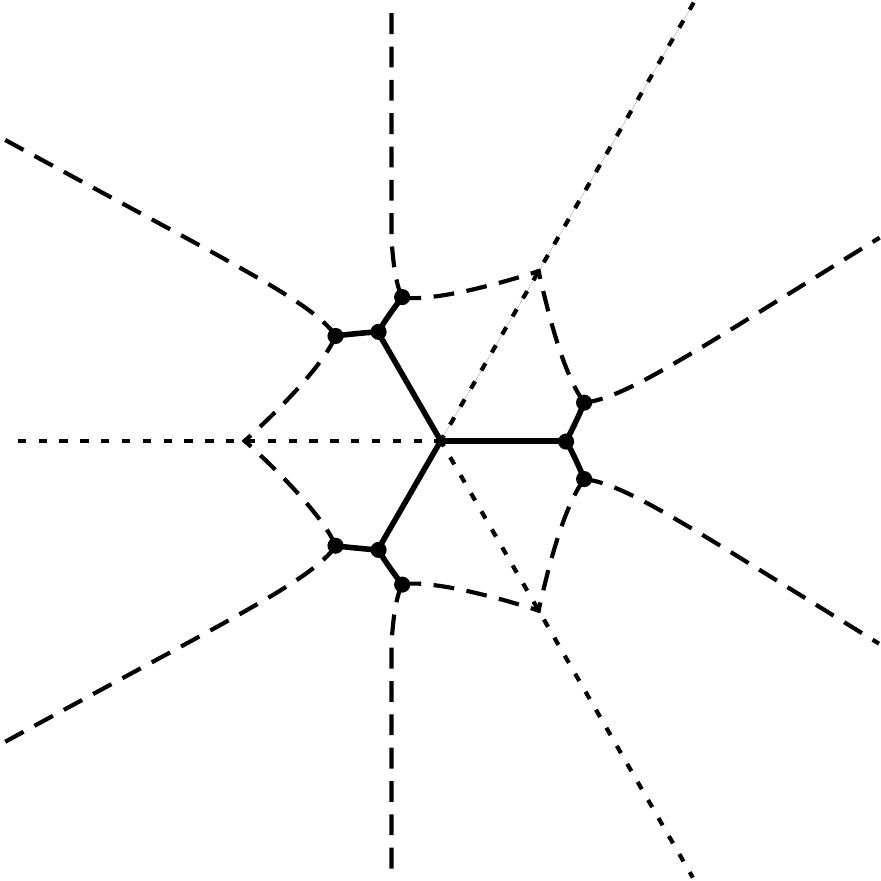}
	\put(65,49){$z_1$}
	\put(65,57){$z_2$}
	\put(65,41){$z_3$}
	\put(80,75){$\Huge -$}
	\put(82,49){$\Huge +$}
	\put(56,55){$\Huge +$}
	\put(56,43){$\Huge +$}
	\put(80,20){$\Huge -$}
\end{overpic}
 \caption{The zero level set of $H(z)$ given by \eqref{def:H} consists of the part of $\Sigma_1$ (solid lines) in $S_0$,
 and two trajectories that emanate from each of the branch points $z_2$, $z_3$
(dashed lines). The figure also shows the  sign of $H(z)$ in the sector $S_0$. 
 The set $\Sigma_2$ is shown with dotted lines. The figure is based on numerical calculations for the value
$t=0.155$. \label{levelcurves}}
\end{figure}

\begin{remark} \label{rem-signs}
In the proof of Theorem \ref{theorem1} (b) we, in fact, have shown that $H(z)>0$
on both sides of $\Sigma_1\cap S_0$.
\end{remark}

\subsection{Proof of Theorem \ref{theorem1} (c)}

The measure $\mu_1$ defined by \eqref{eq:mu} is real on the whiskers in $S_0$ because
of the properties \eqref{trajgamma12} and \eqref{trajgamma13}.
It is also real on $[0,z_1]$, since $\xi_{2,+}(s)$ and $\xi_{1,+}(s)$ are each others
complex conjugates for $s \in [0,z_1]$. Hence $\mu_1$ is a real measure on $\Sigma_1 \cap S_0$.
Because of $\mathbb Z_3$ symmetry it then also follows that $\mu_1$ is real on the full $\Sigma_1$.

Since $\xi_{2,+}(s)$ and $\xi_{1,+}(s)$ are only the same at the branch points, we then see
that the density of $\mu_1$ does not change sign on any arc $[0,z_1]$, $\gamma_{1,2}$ and $\gamma_{1,3}$,
which implies that $\mu_1^o$ and $\mu_1^w$ are either positive or negative measures,
where $\mu_1^o$ and $\mu_1^w$ denote the restriction of $\mu_1$ to $\Sigma_1^o$ and $\Sigma_1^w$,
respectively. For $t = t^*$ we know from \cite{BK} that $\mu_1^o = \mu_1$ is positive and so by 
continuity in $t$, $\mu_1^o$ is positive for every $t \in (t_{*}, t_{**})$. 

Now note that $\xi_1(x) > \xi_2(x)$ for $x \in (z_1, \infty)$ (we saw this in the proof of Lemma \ref{lemma:Hforxreal})
and 
\begin{equation} \label{xi12-atz1} 
	\xi_1(s) - \xi_2(s) = c(s-z_1)^{1/2}(1 + O(s-z_1)) \qquad \text{ as } s \to z_1, \, s \in S_0 
	\end{equation}
for some positive constant $c > 0$.
Then for $z \in \gamma_{1,2}$,
\begin{align*} \int_{z_1}^z d\mu_1(s) & = 
	 \frac{1}{2\pi it} \int_{z_1}^z (\xi_{1,-}(s) - \xi_{2,-}(s)) ds \\
	& = 	\frac{1}{2\pi it} \int_{z_1}^z c (s-z_1)^{1/2} ds \, ( 1 + O(z-z_1)) \\
	& = 		\frac{1}{2\pi it} \frac{2c}{3} (z-z_1)^{3/2} \, (1+O(z-z_1)) \qquad \text{ as } z \to z_1
	\end{align*}
which is positive for $z \in \gamma_{1,2}$ close enough to $z_1$,
since $\arg(z-z_1) \to \pi/3$ as $z \to z_1$. Thus $\mu_1$ is not a negative measure on $\gamma_{1,2}$
and thus has to be positive on $\gamma_{1,2}$. Similarly it is positive on $\gamma_{1,3}$.
Thus $\mu_1^w$ is a positive measure as well.

The total mass of $\mu_1$ is $\frac{1}{2\pi i t} \int_{\Sigma_1} (\xi_{1,-}(s) - \xi_{1,+}(s)) ds$
by \eqref{eq:mu}, which can be written as a contour integral
\[ \int d\mu_1 = \frac{1}{2\pi i t} \oint_{C} \xi_1(s) ds, \]
where $C$ is a contour that goes around $\Sigma_1$ in counterclockwise direction.
By deforming the contour to infinity, and noting the expansion \eqref{xi1:asymptotics}
with the residue $t$ at infinity, we find that $\int d\mu_1 = 1$.
Thus $\mu_1$ is indeed a probability measure on $\Sigma_1$.

\subsection{Proof of Theorem \ref{theorem1} (d)}

We start by studying the equation \eqref{dOmegasuper} and show that it
is indeed the boundary of domain $\Omega(t)$. 
\begin{lemma} \label{lemma:dOmegapolar}
Let $t > \frac{1}{8}$ and $A > 0$.
Then for each $\theta \in [-\pi, \pi]$ there is a unique $r > 0$ such that
$z=r e^{i\theta}$ satisfies the equation \eqref{dOmegasuper}. 
\end{lemma}
\begin{proof}
Putting $z = re^{i \theta}$ into the left hand-side \eqref{dOmegasuper} we find
\begin{equation} \label{dOmegapolar} 
	2 r^3 \cos(3\theta) - r^4 - (1+t) r^2 + A, 
	\end{equation}
whose derivative with respect to $r$ is
\begin{align*} 
	6 r^2 \cos(3 \theta) - 4 r^3 - 2(1+t) r 
	& \leq 6 r^2 - 4r^3 - 2(1+t) r \\
	& = \frac{r}{4} \left[1-8t - (3-4r)^2 \right] 
	\end{align*}
which is $< 0$ for all $r > 0$ since $t > \frac{1}{8}$. 
Thus \eqref{dOmegapolar} is strictly decreasing in $r$.  The value for $r =0$
is equal to $A > 0$, and it tends to $-\infty$ for $r \to \infty$. There
is a unique $r > 0$ for which \eqref{dOmegapolar} is zero, which proves the lemma.
\end{proof}

Lemma \ref{lemma:dOmegapolar} shows that the curve  given in
polar coordinates by 
\[	2 r^3 \cos(3\theta) - r^4 - (1+t) r^2 + A = 0 \]
is  the boundary of a  starshaped domain that contains the origin.
If $A = A(t)$ then it agrees with the equation \eqref{dOmegasuper} for
$\partial \Omega(t)$, which arises from putting
\begin{equation} \label{xiiszbar} 
	\xi = \overline{z} 
	\end{equation}
 in the algebraic equation \eqref{def:speccurve}. For 
$t = t_*$ and $A = A_1(t_*)$ we know from  \cite{BK} that $\partial \Omega(t)$
is given by $\xi_1(z) = \overline{z}$. That is, if $z \in \partial \Omega(t)$,
then the relevant solution of \eqref{def:speccurve} which gives rise to \eqref{xiiszbar}
is the solution $\xi_1(z)$, which is defined on the first sheet of the Riemann surface.

\begin{lemma} For each $t\in(t_*,t_{**})$ we have $z_1 \in \Omega(t)$ and $z_2, z_3 \in \mathbb C \setminus \Omega(t)$.
\end{lemma}

\begin{proof}
Since $z_1$ is a branch point of \eqref{def:speccurve},
\[ \frac{\partial P}{\partial \xi} = 3\xi^2 -2 z^2 \xi - (1+t) z = 0 \]
for $z = z_1$ and $\xi = \xi_1(z_1) = \xi_2(z_1)$. Note that $\xi_2(x) > 0$ for large positive $x$ due
to \eqref{eq:xiatinfinity}, and $\xi_2$ does not have any zeros on $[z_1, \infty)$ as can be easily
checked from the algebraic equation \eqref{def:speccurve}. Thus $\xi_2(z_1) > 0$.
Solving the quadratic equation, we then get $\xi_2(z_1) = F(z_1)$ where
\[ F(x) =  \frac{1}{3} \left[ x^2 + \sqrt{x^4 + 3(1+t) x} \right]. \] 
It is an easy calculus exercise to show that for $t > \frac{1}{8}$ we have
\[ F(x) > \frac{1}{3} \left[ x^2 + \sqrt{x^4 + \frac{27}{8} x} \right] \geq x, \qquad x > 0. \]
Thus 
\[ \xi_2(z_1) > z_1. \]
Now note that $\xi_2(x)$ is real for real $x > z_1$ and it behaves like $x^{1/2}$ as $x \to \infty$
by \eqref{eq:xiatinfinity}, so that clearly $\xi_2(x) < x$ for large enough $x > z_1$.
Thus by continuity there exists $\tilde{x} > z_1$ such that $\xi_2(\tilde{x}) = \tilde{x}$.
Then $\tilde{x}$ belongs to $\partial \Omega(t)$ since \eqref{xiiszbar} is satisfied, 
and noting Lemma \ref{lemma:dOmegapolar},
we conclude that it is the unique intersection point of $\partial \Omega(t)$ with the positive real axis.
Since $\tilde{x} > z_1$ we find  that $z_1$ lies in $\Omega(t)$.

We subdivide $\partial \Omega$ into pieces
\[ (\partial \Omega)_j = \{ z \in \mathbb C \mid \xi_j(z) = \overline{z} \}, \qquad j = 1,2,3, \]
where from now on we drop the $t$-dependance from the notation.
We visualize $\partial \Omega$ on the Riemann surface, by putting the part $(\partial \Omega)_j$
on sheet $j$.
For $t \leq 1/8$ we have $\partial \Omega = (\partial \Omega)_1$, and so $\partial \Omega$
is fully on the first sheet. For $t = 1/8$, the curve contains the branch points $\omega^j z_1^*$, $j=1,2,3$, 
and the structure of $\Sigma_1$ changes near these branch points if we move into the supercritical regime $t > 1/8$. 
Then part of $\partial \Omega$ may move to the second sheet, and this happens indeed since
we just proved that $\tilde{x} \in \partial \Omega$ with $\xi_2(\tilde{x}) = \tilde{x}$.
So the part of $\partial \Omega$ near the real axis is on the second sheet, as well as the parts near
the halfrays at angles $\pm 2\pi/3$. The remaining parts are on the first sheet, and in particular the parts near
the angles $\pm \pi/3$ and $\pi$, where the branch cut  from the second to the third sheet is. 
The curve can then never move to the third sheet, and it follows that $(\partial \Omega)_3 = \emptyset$
for all $t \in (t_*, t_{**})$. The curve intersects the whiskers $\gamma_{1,2}$ and $\gamma_{1,3}$.
Each intersection will mean a change from $(\partial \Omega)_1$ to $(\partial \Omega)_2$ or vice versa.
So there are an odd number of intersections\footnote{There is probably only one intersection,
but we have not been able to prove this.}, and this means that $z_2$ is outside of $\Omega$.
By symmetry with respect to complex conjugation, also $z_3 \not\in \Omega(t)$ and the lemma is proved.
\end{proof}

We are now ready for the proof of part (d).

\begin{proof}[Proof of Theorem \ref{theorem1} (d)]

We give the proof under the assumption that $\partial \Omega$ has one intersection point
with  $\gamma_{1,2}$. The proof can be modified to cover the hypothetical situation of more than
one intersection point (which probably does not occur).

Since $\xi_1(z) = z^2 + t z^{-1} + O(z^{-4})$ as $z \to \infty$, we have for a nonnegative
integer $k$,
\[ t_k := \frac{1}{2 \pi i} \oint_C \frac{\xi_1(s)}{s^k} ds = \begin{cases}
	t & \text{ if } k = 0, \\
	1 & \text{ if } k = 3, \\
	0 & \text{ otherwise,} \end{cases}  \]
where $C$ is a contour that encircles $\Sigma_1$ once in counterclockwise direction.

We deform $C$ inwards so that it consists of $\partial \Omega$ and the plus and
minus sides of the parts of $\Sigma_1^w$ that are outside of $\Omega$. 
We use $(\partial \Omega)_1$ and $(\partial \Omega)_2$ as in the proof of the last lemma.
Then we have
\begin{align} \nonumber t_k & = \frac{1}{2\pi i} \int_{(\partial \Omega)_1 \cup (\partial \Omega)_2} \frac{\xi_1(s)}{s^k} ds {+}
			\frac{1}{2\pi i} \int_{\Sigma_1^w \setminus \Omega} \frac{( \xi_{1,-} - \xi_{1,+})(s)}{s^k} ds \\
	& = \frac{1}{2\pi i} \int_{\partial \Omega} \frac{\overline{s}}{s^k} ds
	+ \frac{1}{2\pi i} \int_{(\partial \Omega)_2} \frac{\xi_1(s)-\xi_2(s)}{s^k} ds
		+	\frac{1}{2\pi i} \int_{\Sigma_1^w \setminus \Omega} \frac{( \xi_{2,+} - \xi_{1,+})(s)}{s^k} ds.
		\label{tksums}
		\end{align}
The integral over $(\partial \Omega)_2$ can be deformed to 
an integral over $\Sigma_1^w \cap \Omega$, which results in
\begin{equation} \label{tksums2}  
	\frac{1}{2\pi i} \int_{(\partial \Omega)_2} \frac{\xi_1(s)-\xi_2(s)}{s^k} ds 
	=   \frac{1}{2\pi i} \int_{\Sigma_1^w \cap \Omega} \frac{\xi_{2,+}(s)-\xi_{1,+}(s)}{s^k} ds. 
	\end{equation}
From \eqref{tksums}--\eqref{tksums2} we get
\[ t_k = \frac{1}{2\pi i} \int_{\partial \Omega} \frac{\overline{s}}{s^k} ds
	+ \frac{1}{2\pi i} \int_{\Sigma_1^w} \frac{( \xi_{2,+} - \xi_{1,+})(s)}{s^k} ds,
\]
which gives \eqref{eq:areak} in view of \eqref{eq:mu}. 
\end{proof}

The proof of Theorem \ref{theorem1} is now complete.

\section{Proof of Theorem \ref{theorem2}} \label{section:proof2}

\subsection{The Riemann-Hilbert problem in the supercritical case }\label{sec-RHP}

The proof of Theorem \ref{theorem2} is based on a steepest descent analysis of a RH problem for
the orthogonal polynomials $P_{n,n}$ that we take from \cite{BK}.
The RH problem arises from the fact that $P_{n,n}$ can be viewed as a multiple
orthogonal polynomial with respect to two weight functions $w_{0,n}$
and $w_{1,n}$ on $\bigcup_{j=1}^3 \Gamma_j$, see Lemma 5.1 of \cite{BK}.

We are free to move the contours  in the complex plane, as long as we respect
the starting and ending directions at infinity. We can then move them to
a contour
\[ \Gamma_Y = \Sigma_1 \cup \bigcup_{j=0}^2 C_j^{\pm} \]
where $C_j^{+}$ denotes the continuation of $\Sigma_1$ from $\omega^j z_2$
and $C_j^{-}$ is the continuation of $\Sigma_1$ from $\omega^j z_3$, for $j=0,1,2$,
see Figure \ref{fig:ContoursForY}.
Then the weights are expressed in terms of the Airy function $\Ai$ and
its derivative (see Definition 5.3 of \cite{BK}).  
They are given on the parts of $\Gamma_Y$ in the sector $S_0$ by
\[ 
		\begin{aligned}
		&\begin{cases} 
		w_{0,n}(x)  = \Ai( \frac{n^{2/3}}{t^{2/3}} x) e^{ \frac{n}{3t} x^3} \\
		w_{1,n}(x)  = \Ai'(\frac{n^{2/3}}{t^{2/3}} x) e^{\frac{n}{3t} x^3} 
		\end{cases} && \quad \text{ on } [0, z_1], \\
		& \begin{cases} 
		w_{0,n}(z) = \frac{1}{3}( \Ai( \frac{n^{2/3}}{t^{2/3}} z) - \omega \Ai(\omega \frac{n^{2/3}}{t^{2/3}} z))  e^{ \frac{n}{3t} z^3} \\
		w_{1,n}(z) = \frac{1}{3}( \Ai'( \frac{n^{2/3}}{t^{2/3}} z - \omega^2 \Ai'(\omega \frac{n^{2/3}}{t^{2/3}} z)) e^{\frac{n}{3t} z^3} 
		\end{cases} && \quad \text{ on } \gamma_{1,2}  \cup C_0^+, \\
	  & \begin{cases} 
		w_{0,n}(x) = \frac{1}{3}( \Ai( \frac{n^{2/3}}{t^{2/3}} z) - \omega^2 \Ai(\omega^2 \frac{n^{2/3}}{t^{2/3}} z))  e^{ \frac{n}{3t} z^3} \\
		w_{1,n}(z) = \frac{1}{3}( \Ai'( \frac{n^{2/3}}{t^{2/3}} z - \omega \Ai'(\omega^2 \frac{n^{2/3}}{t^{2/3}} z)) e^{\frac{n}{3t} z^3} 
		\end{cases} && \quad \text{ on } \gamma_{1,3} \cup C_0^-,
		\end{aligned} \]
and then continued to other sectors by the symmetry property
\[ w_{0,n}(\omega z)  = \omega^2 w_{0,n}(z), \quad 
	   w_{1,n}(\omega z) = \omega w_{1,n}(z), \quad \text{ for } z \in \Gamma_Y. \]

\begin{figure}[t] 
\centering
\begin{overpic}[width=8cm,height=8cm]{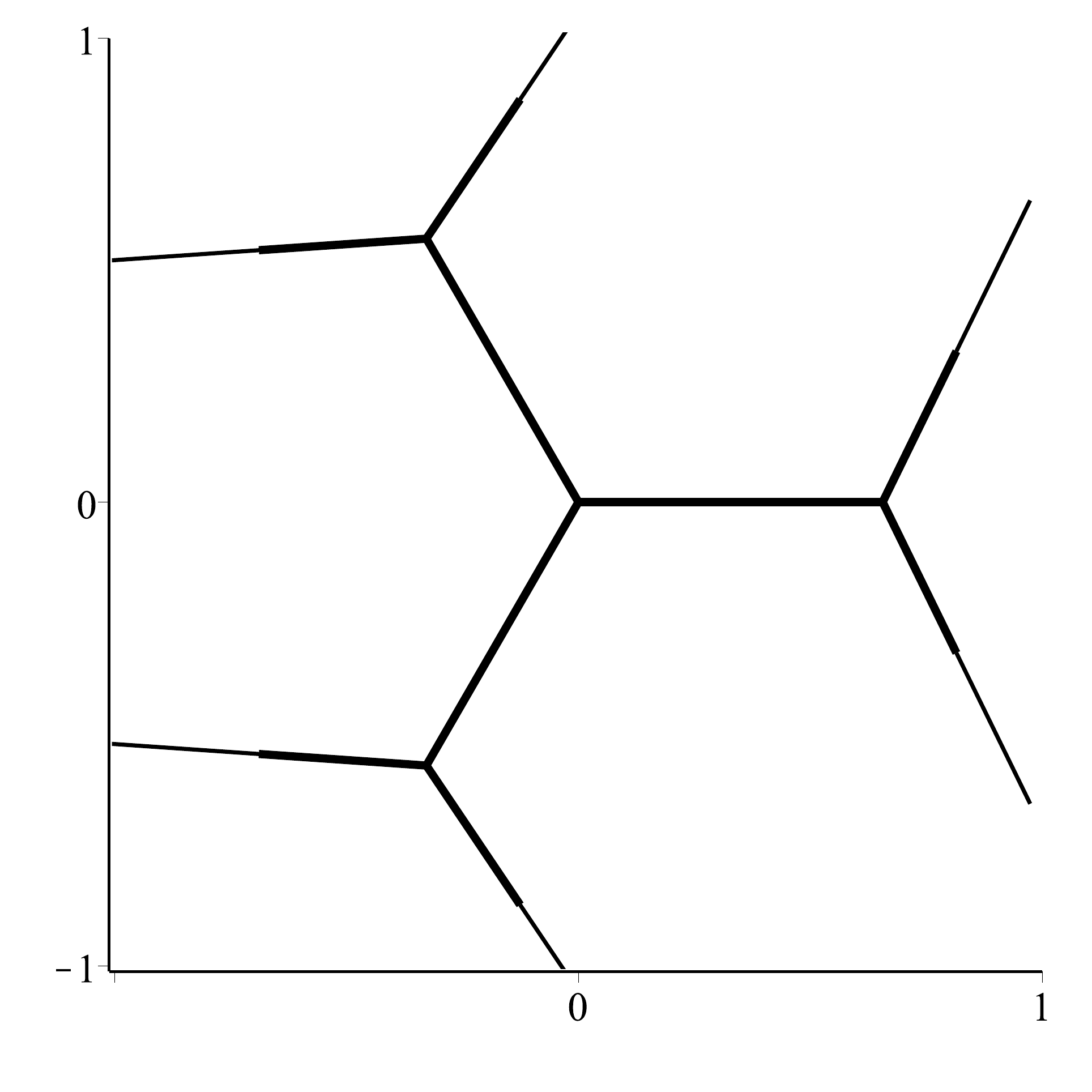}
	\put(49,52){$0$}
	\put(82,52){$z_1$}
	\put(88,66){$z_2$}
	\put(88,39){$z_3$}
	\put(60,56){$\Sigma_1$}
	\put(49,89){$\omega z_3$}
	\put(41,76){$\omega z_1$}
	\put(20,79){$\omega z_2$}
	\put(49,17){$\omega^2 z_2$}
	\put(41,28){$\omega^2 z_1$}
	\put(20,26){$\omega^2 z_3$}
	\put(86,78){$C_0^+$}
	\put(86,26){$C_0^-$}
	\put(44,94){$C_1^-$}
	\put(13,71){$C_1^+$}
	\put(13,33){$C_2^-$}
	\put(42,12){$C_2^+$}
\end{overpic}
\caption{Contour $\Gamma_Y$ in the RH problem for $Y$. \label{fig:ContoursForY}}
\end{figure}
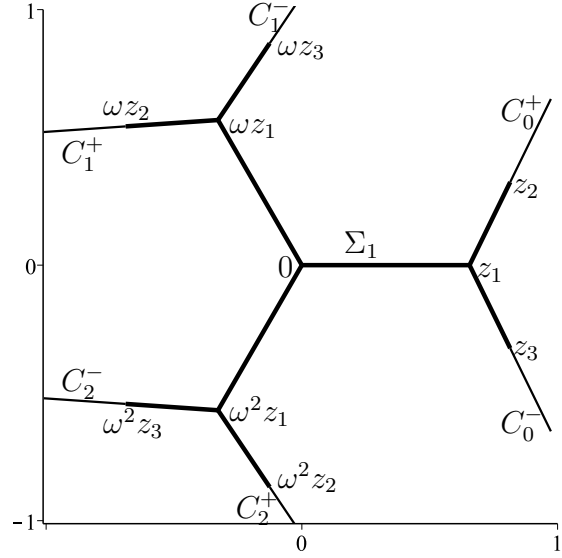
		
The RH problem then is: 
\begin{rhproblem} \label{RHforY} 
The matrix-valued function $Y : \mathbb C \setminus \Gamma_Y \to \mathbb C^{3 \times 3}$ satisfies
\begin{itemize}
\item $Y$ is analytic.
\item $Y_+(z) = Y_-(z) \begin{pmatrix} 1 & w_{0,n}(z) & w_{1,n}(z) \\ 0 & 1  & 0 \\ 0 & 0 & 1 \end{pmatrix}$ for $z \in \Gamma_Y$.
\item $Y(z) = (I + O(1/z)) \diag\left(z^n, z^{-\lceil n/2 \rceil}, z^{-\lfloor n/2 \rfloor} \right)$ as $z \to \infty$.	
\end{itemize}
\end{rhproblem}
The RH problem \ref{RHforY} has a unique solution if and only if $P_{n,n}$ uniquely exist and
in that case		
\begin{equation} \label{Y11} 
	Y_{11}(z) = P_{n,n}(z).
	\end{equation}
In what follows we assume that $n$ is even, since it will simplify the exposition. With minor
modifications, the analysis can be done in the case of odd $n$ as well.

The asymptotic analysis of the RH problem \ref{RHforY} will be 
obtained through a chain of transformations. In the end, we will obtain
a ``model'' RH problem   \ref{RHforM} for a matrix $M(z)$
(also known as a global parametrix), that gives the leading pre-exponential behavior of $Y(z)$ 
as $n \to \infty$, and a number of 
local parametrices (RH problem \ref{RHforP}), that are used to obtain the error estimate.
This chain of transformations is fully described in  
\cite{BK} for the subcritical case $t<t_*$ and constitutes a substantial part of that paper. 
In order to limit the size of the present paper, we will not repeat the full description of
transformations from 
\cite{BK} here, but rather underline  the distinctions of the supercritical case,
caused by the presence of the whiskers.
 
The transformations $Y \mapsto X \mapsto V$ are the same as in \cite{BK}. The precise form of these transformations is irrelevant
for the present paper. We only need to know the resulting RH problem, and
the fact that the $11$-entry remains unchanged, so that by \eqref{Y11},
\begin{equation} \label{V11}
	V_{11}(z) = Y_{11}(z) = P_{n,n}(z).
	\end{equation}
	To state the RH problem for $V$ we introduce
\begin{align} \label{defQ1}
	Q_1(z) & = \begin{cases}  \frac{2}{3t} z^{3/2} - \frac{1}{3 t} z^3,  &  z \in S_0, \\
	 - \frac{2}{3t} z^{3/2} - \frac{1}{3t} z^3, & z \in S_1 \cup S_2,
	\end{cases} \\ \label{defQ2}
	Q_2(z) & =  \begin{cases} \frac{4}{3t} z^{3/2}, & 0 <  \arg z < 2\pi /3 \text{ or } -\pi < \arg z < -2\pi/3, \\
	- \frac{4}{3t} z^{3/2}, & - 2\pi/3 < \arg z < 0 \text{ or } 2\pi/3 < \arg z < \pi, 
		\end{cases} 
		\end{align} 
with principal branches of the fractional powers.
The definitions are such that $ Q_j(\omega z) = Q_j(z)$ for $j=1,2$.
We also define the constant 
\begin{equation} \label{def:alpha} 
	\alpha = \frac{1}{2} + \frac{i}{6} \sqrt{3} = \frac{\sqrt{3}}{3} e^{\pi i/6}.
	\end{equation}
Then the RH problem for $V$ is, see \cite[RH problem 6.7]{BK}:

\begin{rhproblem} \label{RHforV}
$V : \mathbb C \setminus \Gamma_V \to \mathbb C^{3\times 3}$, where $\Gamma_V = \Gamma_Y \cup \Sigma_2$
and $\Sigma_2$ is given by \eqref{eq:Sigma2}, satisfies
\begin{itemize}
\item $V$ is analytic
\item $V_+ = V_- J_V$ on $\Gamma_V$ with
\begin{equation} \label{JV1} 
	J_V(z)  = \begin{cases} \begin{pmatrix} 1 & e^{-n Q_1(z) } & 0 \\ 
	0 & 1 & 0 \\ 0 & 0 & 1 \end{pmatrix},  & 
 	   z \in \Sigma_{1}^o = \bigcup_{j=0}^2 [0, \omega^j z_1] \\
 \begin{pmatrix} 1 & \alpha e^{-n Q_1(z) } & 0 \\ 
	0 & 1 & 0 \\ 0 & 0 & 1 \end{pmatrix},  & 
 	   z \in \bigcup_{j=0}^2 \omega^j ( \gamma_{1,2} \cup C_0^+)  \\
 \begin{pmatrix} 1 & \overline{\alpha} e^{-n Q_1(z) } & 0 \\ 
	0 & 1 & 0 \\ 0 & 0 & 1 \end{pmatrix},  & 
 	   z \in \bigcup_{j=0}^2 \omega^j (\gamma_{1,3} \cup C_0^-)  \\
\begin{pmatrix} 1 & 0 &  \overline{\alpha}  e^{-n Q_{1,-}(z)} \\ 
	0 & \omega^2 e^{n Q_2(z)}  & 1 \\ 
		0 & 0 & \omega e^{-n Q_2(z)} \end{pmatrix},  & z \in \Sigma_2. \\
	  \end{cases}
\end{equation}
\item $ V(z) = (I + O(1/z)) A(z) \begin{pmatrix} z^n & 0 & 0 \\ 0 & z^{-n/2} & 
	0 \\ 0 & 0 & z^{-n/2}   \end{pmatrix}$ as $z \to \infty$, 
	where $A(z)$ is given by, 
	\begin{equation} \label{defAz}
	A(z) = \begin{pmatrix} 1 & 0 & 0 \\ 0 & z^{1/4} & 0 \\ 0 & 0 & z^{-1/4} \end{pmatrix}
	\times \frac{1}{\sqrt{2}}
	\begin{cases}
	\begin{pmatrix} \sqrt{2} & 0 & 0 \\ 0 & 1 & - i \\ 0 &  -i & 1 \end{pmatrix}, &
		z \in S_0, \\
	\begin{pmatrix} \sqrt{2} & 0 & 0 \\ 0 & i & 1 \\ 0 &  -1 & i \end{pmatrix}, &
		z \in S_1, \\	
	\begin{pmatrix} \sqrt{2} & 0 & 0 \\ 0 & -i & - 1 \\ 0 &  1 & i \end{pmatrix}, &
		z \in S_2,
		\end{cases}
	\end{equation}
		see also formula (6.9) in \cite{BK}.
\end{itemize}	 
\end{rhproblem}

The next transformation $V \mapsto U$ uses the measures $\mu_1$ and $\mu_2$, see \eqref{defmu2},
and their $g$ functions 
\begin{equation} \label{defgj} 
	g_j(z) = \int \log(z-s) d\mu_j(s), \qquad j = 1,2. 
	\end{equation}
The branches for the logarithm are chosen as in \cite[Section 6.3]{BK}.
It means that $g_1$ is
defined and analytic in $\mathbb C \setminus (\Sigma_1 \cup \mathbb R^-)$,
$g_2$ is analytic in $\mathbb C \setminus \Sigma_2$ with symmetries
\[ g_1(\omega^{\pm} z) =  g_1(z) \pm 2 \pi i/3, \quad
	 g_2(\omega^{\pm} z) =  g_2(z) \pm  \pi i/3, \qquad \text{for } z \in S_0. \]
In addition $g_1(x)$ and $g_2(x)$ are real for real $x > z_1$,
see \cite[Section 6.3]{BK}.

Then, as in \cite[Lemma 6.8]{BK}, there is a constant $\ell$ such that
\[ g_{1,+}(z) + g_{1,-}(z) - g_2(z) = 
		\frac{2}{3t} z^{3/2} - \frac{1}{3t} z^3 + \ell, \qquad z \in [0,z_1]. \]
On the whiskers, we then have 
\begin{equation}\label{g1jump} 
	g_{1,+}(z) + g_{1,-}(z) - g_2(z) = 
		\frac{2}{3t} z^{3/2} - \frac{1}{3t} z^3 + \ell 
			\begin{cases} + 2\pi i \beta,  & \quad z \in \gamma_{1,2}, \\
       - 2\pi i \beta, &  \quad z \in \gamma_{1,3}, 
			\end{cases} \end{equation}
where $\beta = \mu_1(\gamma_{1,2}) = \mu_1(\gamma_{1,3})$, which by the symmetry
is the same as $\beta = \frac{1}{6} \mu_1(\Sigma_1^w)$, see \eqref{def-beta}.

The definition of $U$ is as in Definition 6.9 of \cite{BK}, namely
\begin{equation} \label{defU} 
	U(z) = \begin{pmatrix} e^{-n \ell} & 0 & 0 \\ 0 & 1 & 0 \\ 0 & 0 & 1 \end{pmatrix}
	V(z) \begin{pmatrix} e^{-n (g_1(z)-\ell)} & 0 & 0 \\ 0 & e^{n(g_1(z) - g_2(z))} & 0 \\ 0 & 0 & e^{n g_2(z)} \end{pmatrix}.
	\end{equation}
It then follows that 
\begin{equation} \label{U11}
	U_{11}(z) = P_{n,n}(z) e^{-n g_1(z)}, \qquad z \in \mathbb C \setminus \Gamma_U,
	\end{equation}
with $\Gamma_U = \Gamma_V = \Gamma_Y \cup \Sigma_2$.

The jumps in the RH problem
for $U$ are conveniently expressed in terms of the functions
\begin{align} \label{phi1}
	\varphi_1(z) & = \frac{1}{2t} \int_{\omega^j z_1}^z (\xi_1(s) - \xi_2(s)) ds && 
	\text{for } z \in S_j \setminus \Sigma_1, \, j = 0,1,2, \\
	 \label{phi2}
	\varphi_2(z) & = \frac{1}{2t} \int_0^z (\xi_2(s) - \xi_3(s)) ds  \mp \frac{\pi i}{6} &&
	\text{for } z \in S_0 \cap \mathbb C^{\pm} \setminus \Sigma_1,
	\end{align}
	with formulas similar to \eqref{phi2} in $S_1 \setminus \Sigma_1$ and $S_2 \setminus \Sigma_1$
	as in \cite[formula (6.23)]{BK}.
The path of integration in \eqref{phi1}	goes from $\omega^j z_1$ to $\omega^j (z_1 + \varepsilon)$ for some $\varepsilon > 0$
and then continues to $z$ in the domain $S_j \setminus \Sigma_1$. The path of integration in \eqref{phi2} is
in $S_0^{\pm} \setminus \Sigma_1$, where $S_0^{\pm} =S_0 \cap \mathbb C^{\pm}$.
Note that 
\begin{equation} \label{phi-H}
\Re \varphi_1(z)=\frac{1}{2t}H(z)
\end{equation}
where $H$ was given by \eqref{def:H}.

We remind, see \eqref{def-beta}, that
\begin{equation} \label{beta}
	\beta = \frac{1}{6} \mu_1(\Sigma_1^w) =  \mu_1(\gamma_{1,2})  \in (0, \tfrac{1}{6}).
	\end{equation} 
This number appears in the RH problem for $U$. The reason for this is
that it appears in the jump relation \eqref{g1jump}, and it also
comes in the jump of $\varphi_1$ on $\Sigma_1$.  
Indeed, from the definition of $\varphi_1$ (in particular the choice of the path from $\omega^j z_1$ to $z$) 
and \eqref{eq:mu}, \eqref{beta}, one finds 
\begin{equation} \label{phi1-jump} 
	\varphi_{1,+}(z) + \varphi_{1,-}(z) = 
	\begin{cases} 0 & \text{ on } \Sigma_1^o, \\
	 2 \pi i \beta & \text{ on } \bigcup_{j=0}^2 \omega^j \gamma_{1,2}, \\
	 - 2\pi i \beta & \text{ on } \bigcup_{j=0}^2 \omega^j \gamma_{1,3}. 
		\end{cases} \end{equation}
We will not give the details for the
(sometimes tedious) calculations that produce a list of identities that 
lead to the following RH problem.

\begin{rhproblem} \label{RHforU}
$U$ is the solution of the following RH problem.
\begin{itemize}
\item $U$ is analytic in $\mathbb C \setminus \Gamma_U$, where $\Gamma_U = \Gamma_V$.
\item $U_+ = U_- J_U$ on $\Gamma_U$ with
\begin{equation} \label{JU} 
	J_U(z) = \begin{cases} 
	\begin{pmatrix} e^{-2n \varphi_{1,+}(z)} & 1 & 0 \\ 
	0 & e^{-2n \varphi_{1,-}(z)} & 0 \\ 0 & 0 & 1 \end{pmatrix}, &  z \in \Sigma_{1}^o, \\
	\begin{pmatrix} e^{-2n \varphi_{1,+}(z) + 2 \pi in \beta} & \alpha e^{2\pi in \beta} & 0 \\ 
		0 & e^{-2n \varphi_{1,-}(z) +2\pi in \beta} & 0 \\ 0 & 0 & 1 \end{pmatrix}, &
		z \in \bigcup_{j=0}^2 \omega^j \gamma_{1,2}, \\
	\begin{pmatrix} e^{-2n \varphi_{1,+}(z) - 2\pi in \beta} & \overline{\alpha} e^{-2\pi in \beta} & 0 \\ 
		0 & e^{-2n \varphi_{1,-}(z) -2\pi in \beta} & 0 \\ 0 & 0 & 1 \end{pmatrix}, &
		z \in \bigcup_{j=0}^2 \omega^j \gamma_{1,3}, \\
		\begin{pmatrix} 1 & 0 &  \overline{\alpha}  e^{2n \varphi_{1,-}(z)} \\ 0 & \omega^2 e^{-2n \varphi_{2,+}(z)}  & 1 \\ 
		0 & 0 & \omega e^{-2n \varphi_{2,-}(z)} \end{pmatrix}, &  z \in \Sigma_2, 
	  \end{cases} \end{equation}
and
\begin{equation} \label{JU2}
			J_U(z) = \begin{cases} \begin{pmatrix} 1 &  \alpha  e^{2n \varphi_1(z)} & 0 \\ 0 & 1 & 0 \\ 0 & 0 & 1 \end{pmatrix},
	  	& z \in \bigcup_{j=0}^2 C_j^+ \\
	  \begin{pmatrix} 1 &  \overline{\alpha}  e^{2n \varphi_1(z)} & 0 \\ 0 & 1 & 0 \\ 0 & 0 & 1 \end{pmatrix},
	  	& z \in \bigcup_{j=0}^2 C_j^-.
	  \end{cases} \end{equation}
\item $U(z) = (I + O(1/z)) A(z)$ as $z\to \infty$, where $A(z)$ is given by \eqref{defAz}.
\end{itemize}
\end{rhproblem}

Further transformations $U  \mapsto T \mapsto S$ are again as in \cite{BK}.
The transformation $U \mapsto T$ is effective in the domain bounded by $C_0$, $C_1$, $C_2$
and the whiskers. It simplifies the jumps, as it kills for example the $13$-entry in
the jump matrix on $\Sigma_2$. The number $\alpha$ also disappears from the jumps,
and only its real part $\Re \alpha = \frac{1}{2}$, see  \eqref{def:alpha}, 
remains in the jump matrices. It leads to the following RH problem.

\begin{rhproblem} \label{RHforT}
$T$ is the solution of the following RH problem.
\begin{itemize}
\item $T$ is analytic in $\mathbb C \setminus \Gamma_T$, where $\Gamma_T = \Gamma_U$.
\item $T_+ = T_- J_T$ on $\Gamma_T$ with
\begin{equation} \label{JT} 
	J_T(z) = \begin{cases} 
	\begin{pmatrix} e^{-2n \varphi_{1,+}(z)} & 1 & 0 \\ 
	0 & e^{-2n \varphi_{1,-}(z)} & 0 \\ 0 & 0 & 1 \end{pmatrix}, &  z \in \Sigma_{1}^o, \\
	\begin{pmatrix} e^{-2n \varphi_{1,+}(z) + 2 \pi in \beta} & \frac{1}{2} e^{2\pi in \beta} & 0 \\ 
		0 & e^{-2n \varphi_{1,-}(z) +2\pi in \beta} & 0 \\ 0 & 0 & 1 \end{pmatrix}, &
		z \in \bigcup_{j=0}^2 \omega^j \gamma_{1,2}, \\
	\begin{pmatrix} e^{-2n \varphi_{1,+}(z) - 2\pi in \beta} & \frac{1}{2} e^{-2\pi in \beta} & 0 \\ 
		0 & e^{-2n \varphi_{1,-}(z) -2\pi in \beta} & 0 \\ 0 & 0 & 1 \end{pmatrix}, &
		z \in \bigcup_{j=0}^2 \omega^j \gamma_{1,3}, \\
		\begin{pmatrix} 1 & 0 &  0 \\ 0 & \omega^2 e^{-2n \varphi_{2,+}(z)}  & 1 \\ 
		0 & 0 & \omega e^{-2n \varphi_{2,-}(z)} \end{pmatrix}, &  z \in \Sigma_2, 
		 \end{cases} \end{equation}
and
\begin{equation} \label{JT2}
	J_T(z) = \begin{cases} 
	  \begin{pmatrix} 1 &  \frac{1}{2}  e^{2n \varphi_1(z)} & 0 \\ 0 & 1 & 0 \\ 0 & 0 & 1 \end{pmatrix},
	  	& z \in \bigcup_{j=0}^2 C_j^+ \\
	  \begin{pmatrix} 1 &  \frac{1}{2}  e^{2n \varphi_1(z)} & 0 \\ 0 & 1 & 0 \\ 0 & 0 & 1 \end{pmatrix},
	  	& z \in \bigcup_{j=0}^2 C_j^-.
	  \end{cases} \end{equation}
\item $T(z) = (I + O(1/z)) A(z)$ as $z\to \infty$, where $A(z)$ is given by \eqref{defAz}.
\end{itemize}
\end{rhproblem}

\begin{figure}[t] 
\centering
\begin{overpic}[width=8cm,height=8cm]{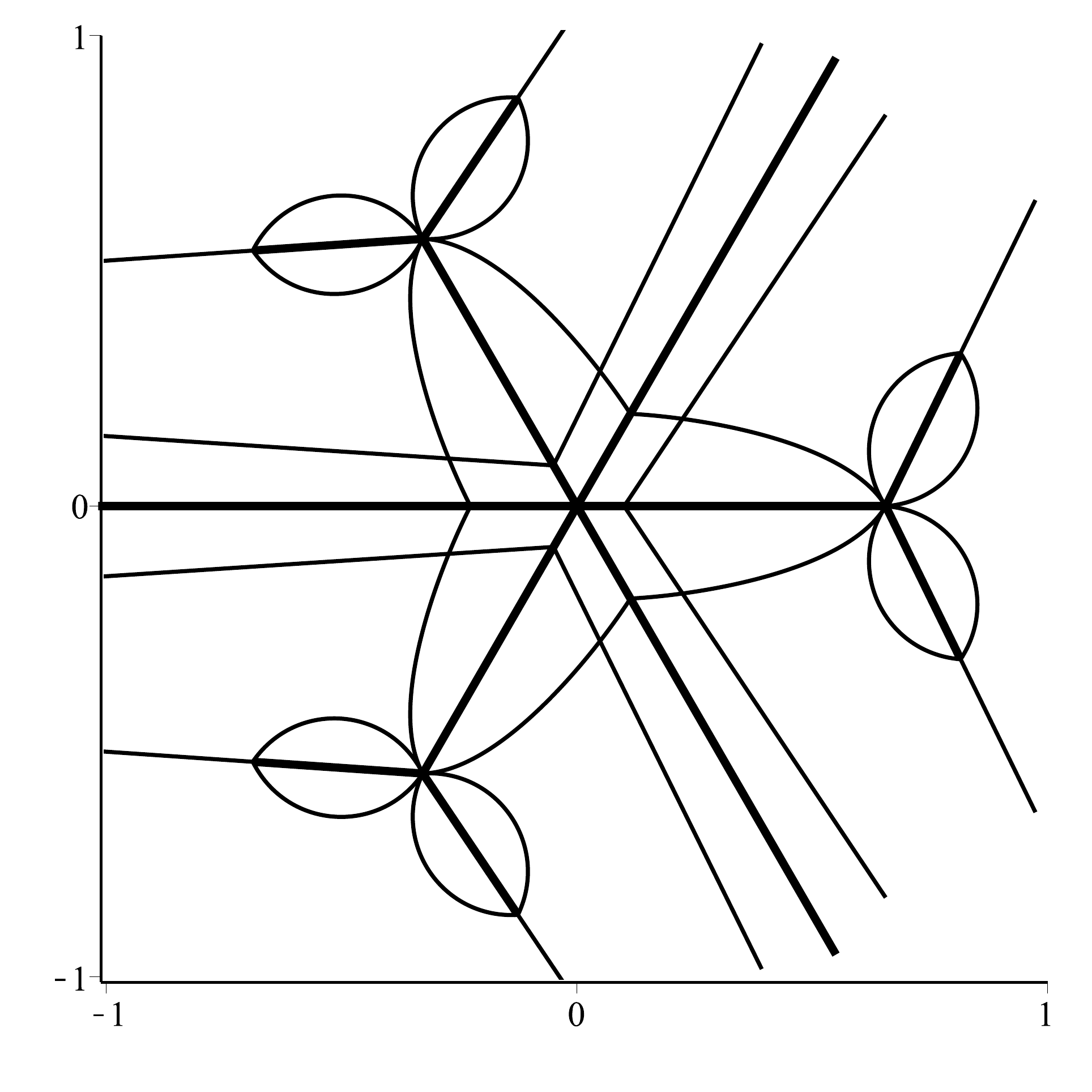}
	\put(85,52){$z_1$}
	\put(88,66){$z_2$}
	\put(88,38){$z_3$}
\end{overpic}
\caption{Contours $\Gamma_S$ in the RH problem for $S$. \label{fig:ContoursForS}}
\end{figure}

In the step $T \mapsto S$ lenses are opened around $\Sigma_1$ and $\Sigma_2$. 
It is based on a factorization of the jump matrices on these contours. 
This is the same as in \cite{BK}, except that we now also have the whiskers in $\Sigma_1$.
The jump matrix $J_T$ on $\bigcup_j \omega^j \gamma_{1,2}$, see \eqref{JT}, has the following factorization (we only list the nontrivial
$2 \times 2$ block)
\begin{multline}  \begin{pmatrix} e^{-2n \varphi_{1,+}(z) + 2 \pi in \beta} & \frac{1}{2} e^{2\pi in \beta} \\ 
		0 & e^{-2n \varphi_{1,-}(z) +2\pi in \beta}  \end{pmatrix} 
		\\ =
			\begin{pmatrix} 1 & 0 \\ 2 e^{-2n \varphi_{1,-}(z)} & 1 \end{pmatrix}
			\begin{pmatrix} 0 & \frac{1}{2} e^{2\pi in \beta} \\ 2 e^{-2\pi in \beta} & 0
			\end{pmatrix}
			\begin{pmatrix} 1 & 0 \\ 2 e^{-2n \varphi_{1,+}(z)} & 1 \end{pmatrix}
		\end{multline}
and there is a similar factorization fo the jump matrix $J_T$ on $\bigcup_j \omega^j \gamma_{1,3}$.

The definition of $S$ (which we do not specify here in detail) leads
to a RH problem for $S$ on a complicated set of contours, see Figure \ref{fig:ContoursForS}.

\begin{rhproblem} \label{RHforS}
$S$ is the solution of the following RH problem.
\begin{itemize}
\item $S$ is analytic in $\mathbb C \setminus \Gamma_S$, where $\Gamma_S$ consists of $\Gamma_U$ and lenses around
$\Sigma_1$ and $\Sigma_2$.
\item $S_+ = S_- J_S$ on $\Gamma_S$ with
\begin{equation} \label{JS} 
	J_S = \begin{cases} 
	\begin{pmatrix} 0 & 1 & 0 \\ -1 & 0 & 0 \\ 0 & 0 & 1 \end{pmatrix}  &  \text{ on } \Sigma_{1}^o, \\
	\begin{pmatrix} 0 & \frac{1}{2} e^{2\pi in \beta} & 0 \\ -2 e^{-2\pi in\beta} & 0 & 0 \\ 0 & 0 & 1 \end{pmatrix} 
	& \text{ on } \bigcup_{j=0}^{2} \omega^j \gamma_{1,2} \\
	\begin{pmatrix} 0 & \frac{1}{2} e^{-2\pi in \beta} & 0 \\ -2 e^{2\pi in\beta} & 0 & 0 \\ 0 & 0 & 1 \end{pmatrix} 
	& \text{ on } \bigcup_{j=0}^2 \omega^j \gamma_{1,3} \\
	\begin{pmatrix} 1 & 0 & 0 \\ 0 & 0 & 1 \\ 0 & -1 & 0 \end{pmatrix} 
	& \text{on } \Sigma_2, 
		\end{cases} \end{equation}
and	
\begin{equation} \label{JS2} 
	J_S = \begin{pmatrix} 1 & \frac{1}{2} e^{2 n \varphi_1} & 0 \\ 0 & 1 & 0 \\ 0 & 0 & 1 \end{pmatrix} \qquad
	 \text{ on } \bigcup_{j=0}^{2} C_j^{\pm}
	 \end{equation}
on $\Gamma_S$, and on the lips of the lenses it is
\begin{equation} \label{JSlips}
	J_S = \begin{cases}
		\begin{pmatrix} 1 & 0 & 0 \\ e^{-2n \varphi_1} & 1 & 0 \\ 0 & 0 & 1 \end{pmatrix} 
			& \begin{array}{l} \text{ on lips around $\Sigma_1^o$} \\ \text{ outside lens around $\Sigma_2$} \end{array} \\
		\begin{pmatrix} 1 & 0 & 0 \\ 2 e^{-2n \varphi_1} & 1 & 0 \\ 0 & 0 & 1 \end{pmatrix} 
			& \text{ on lips around $\Sigma_1^w$} \\
		\begin{pmatrix} 1 & 0 & 0 \\ 0 & 1 & 0 \\ 0 & \omega^{\mp} e^{-2n \varphi_2} & 1 \end{pmatrix}
		  &  \begin{array}{l} \text{ on lips around $\Sigma_2$} \\ \text{ outside lens around $\Sigma_1$} \end{array} \\
		\begin{pmatrix} 1 & 0 & 0 \\ e^{-2n \varphi_1} & 1 & 0 \\ \pm e^{-2n (\varphi_1 + \varphi_2)} & 0 & 1 \end{pmatrix}
		  &  \begin{array}{l} \text{ on lips around $\Sigma_1$} \\ \text{ inside lens around $\Sigma_1$} \end{array} \\
		\begin{pmatrix} 1 & 0 & 0 \\ 0 & 1 & 0 \\ \omega^{\pm} e^{-2n (\varphi_1 + \varphi_2)} & \omega^{\mp} e^{-2n \varphi_2} & 1 \end{pmatrix}
		  &  \begin{array}{l} \text{ on lips around $\Sigma_2$} \\ \text{ inside lens around $\Sigma_1$.} \end{array} 
			\end{cases}
	\end{equation}
\item $S(z) = (I + O(1/z)) A(z)$ as $z\to \infty$, where $A(z)$ is given by \eqref{defAz}.
\end{itemize}
\end{rhproblem}
Notice that by \eqref{U11} and the transformations $U \mapsto T \mapsto S$,
\begin{equation} \label{S11}
	S_{11}(z) = T_{11}(z) = U_{11}(z) = P_{n,n}(z) e^{-n g_1(z)} 	\quad \text{for $z$ outside the lenses around $\Sigma_1$.}
	\end{equation}

As it was shown in \cite{BK}, the jumps on the lips of the lenses are such that $J_S = I + O(e^{-cn})$ as $n \to \infty$
uniformly if we stay away from the branch points $\omega^j z_k$, $j=0,1,2$, $k=1,2,3$.
Even though the lenses around whiskers were not considered in \cite{BK},
the above estimate is valid on these lenses due to \eqref{phi-H} and 
Remark \ref{rem-signs}.
Ignoring the jumps that are close to the identity matrix, we arrive at
a model RH problem.

\subsection{The model problem (outer parametrix)} \label{sec-modelRHP}

We ignore the jumps \eqref{JSlips} on the lenses, as well as the jump \eqref{JS2} on $\bigcup_j C_j^{\pm}$, 
and we find the following model RH problem for a matrix $M = M_n$.

\begin{rhproblem} \label{RHforM}
\begin{itemize}
\item $M$ is analytic in $\mathbb C \setminus (\Sigma_1 \cup \Sigma_2)$, 
\item $M_+ = M_- J_M$ on $\Sigma_1 \cup \Sigma_2 $ with
\begin{equation} \label{JM}
	J_M = \begin{cases} \begin{pmatrix} 0 & 1 & 0 \\ -1 & 0 & 0 \\ 0 & 0 & 1 \end{pmatrix} & \quad \text{on } 
		\Sigma_1^o  \\
		\begin{pmatrix} 0 & \frac{1}{2} e^{2\pi in \beta} & 0 \\ -2 e^{-2\pi in\beta} & 0 & 0 \\ 0 & 0 & 1 \end{pmatrix} 
	& \quad \text{on } \bigcup_{j=0}^2 \omega^j \gamma_{1,2} \\
	\begin{pmatrix} 0 & \frac{1}{2} e^{-2\pi in \beta} & 0 \\ -2 e^{2\pi in\beta} & 0 & 0 \\ 0 & 0 & 1 \end{pmatrix} 
	& \quad \text{on } \bigcup_{j=0}^2 \omega^j \gamma_{1,3} \\
	\begin{pmatrix} 1 & 0 & 0 \\ 0 & 0 & 1 \\ 0 & -1 & 0 \end{pmatrix} 
	& \quad \text{on } \Sigma_2 
 	\end{cases} \end{equation}
\item $M(z)  = (I + O(1/z)) A(z)$ as $z \to \infty$ where $A(z)$
is given by \eqref{defAz}
\item For $j=0,1,2$ and $k=1, 2, 3$ we have
	\[ M(z) =  O\left((z- \omega^j z_k)^{-1/4}\right) \qquad \text{ as } z \to \omega^j z_k. \]
\end{itemize}
\end{rhproblem}

The model problem depends on $n \beta \mod{\mathbb Z}$ with $n \in \mathbb N$ and $0 < \beta \frac{1}{6}$.  
The asymptotic condition \eqref{defAz} is compatible with the jump of $M$ on the unbounded
contour $\Sigma_2$, since
\[ A_+ = A_- \begin{pmatrix} 1 & 0 & 0 \\ 0 & 0 & 1 \\ 0 & -1 & 0 \end{pmatrix} \quad \text{ on } \Sigma_2. \]
as can be checked from the definition \eqref{defAz}.

We write
\begin{equation} \label{def-betastar}
\beta^* = \frac{1}{2} + \frac{\tau}{2\pi i} \log 2 + \int_{\infty_1}^{-A^{1/3}} \omega_R
\end{equation}
see \eqref{omegaR}, \eqref{def-tau} and \eqref{Nepsilon} for
the definitions of $\omega_R$, $\tau$,  and $-A^{1/3}$.

\begin{proposition} \label{prop:modelRHP}
The following holds:
\begin{enumerate}
\item[\rm (a)]  The model RH problem is  solvable
if and only if 
\begin{equation} \label{betacongruent} 
	n \beta \neq \beta^*  \qquad \mod{\mathbb Z}
	\end{equation}
and if \eqref{betacongruent} holds then there is a unique solution $M = M_n$.
\item[\rm (b)] 
Let $\varepsilon > 0$ and $r > 0$, and let $\| \cdot \|$ denote any matrix norm.
Then there is a constant $K > 0$ such that
\[ \| M_n^{-1}(z)\| \leq K, \qquad \| M_n(z) \| \leq K \]
holds for all $z$ such that
$ |z- \omega^j z_k| \geq r$ for all $j=0,1,2$, $k = 1,2,3$, and all
$n \in \mathbb N$ such that $\dist(n \beta - \beta^*, \mathbb Z) \geq \varepsilon$.
\item[\rm (c)] The $11$-entry of $M_n$ has the form given in \eqref{def-Mn11}.
\end{enumerate}
\end{proposition}

The proof of Proposition \ref{prop:modelRHP} is rather long. We decided
to put it in a separate section~\ref{sec:modelRHP}.

\subsection{The local parametrix} 
Let $r > 0$ be a small number such that the disks $D(\omega^j z_k,r)$, $j = 0,1,2$, $k=1,2,3$, 
are all disjoint and they do not intersect with the lens  around $\Sigma_2$. Let $D$ denote
the union of these nine disks.

The local parametrix $P$ is defined in each of these disks and it should satisfy
\begin{rhproblem} \label{RHforP}
$P$ satisfies the following:
\begin{itemize}
\item $P$ is analytic in $D \setminus \Gamma_S$, 
\item $P_+ = P_- J_P$ on $\Gamma_S \cap D$ with $J_P = J_S$ is given by \eqref{JS}.
\item $P(z) = (I + O(1/n)) M_n(z)$ as $n \to \infty$ uniformly for $z$ on the boundary of each of the disks.
\end{itemize}
\end{rhproblem}

We will be able to constuct $P$ with the desired jumps, but the required matching can be done
only if $n$ is restricted to 
\[ \mathbb N_{\varepsilon} = \{ n \in \mathbb N \mid \dist(n \beta - \beta^*, \mathbb Z) \geq \varepsilon \} \]
for some $\varepsilon > 0$. Note that  this agrees with the earlier definition \eqref{Nepsilon},
in view of \eqref{def-betastar}.

The construction can be done in each of the disks with Airy functions. We will do it in some detail for the point $z_1$
because of the somewhat unusual fact that three pieces of $\Sigma_1$ are connected at $z_1$ and, as a result of 
the opening of lenses, that nine curves from $\Gamma_S$ come together at $z_1$. 
The construction of $P$ at $\omega z_1$ and $\omega^2 z_1$ follows by symmetry and the construction
at the other branch points is of a standard form, see e.g. \cite{Deift}.

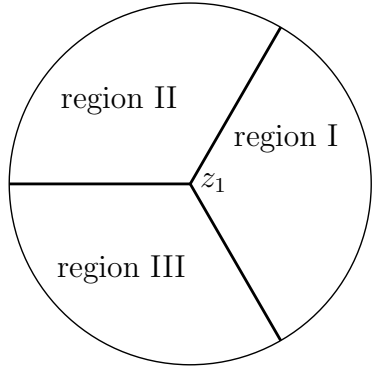
\begin{figure}[t] 
\centering
\begin{overpic}[width=5cm,height=5cm]{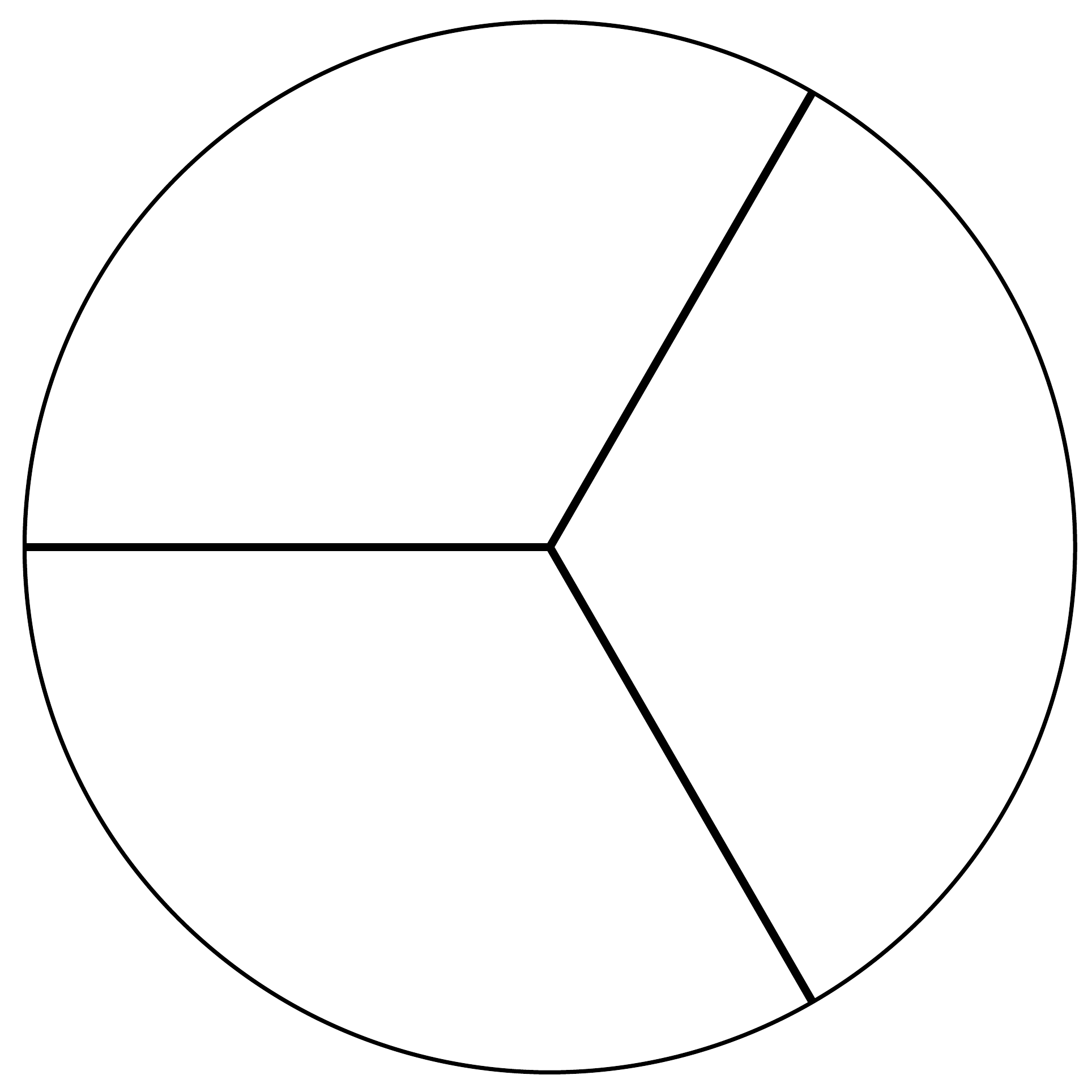}
	\put(53,49){$z_1$}
	\put(62,60){region I}
	\put(16,70){region II}
	\put(15,25){region III}
\end{overpic}
\caption{Disk $D(z_1,r)$ around $z_1$ and the parts of $\Sigma_1$ inside this disk. \label{fig:DiskDz1}}
\end{figure}

The set $D(z_1,r) \setminus \Sigma_1$  consists of three parts, that we call regions I, II, III,
as shown in Figure \ref{fig:DiskDz1}.  
Note that by \eqref{phi1-jump} we have that the function
\[ \widehat{\varphi}_1(z) = \begin{cases} \varphi_1(z) & \text{ in region I,} \\
		2 \pi i \beta -  \varphi_1(z) & \text{ in region II,} \\
		-2\pi i \beta - \varphi_1(z)  & \text{ in region III,} \end{cases} \]
is analytic in $D(z_1, r) \setminus [0, z_1]$,  By \eqref{phi1} it is real for real $z > z_1$ 
and for some constant $c > 0$,
\[ \widehat{\varphi}_1(z) = c(z-z_1)^{3/2} + O(z-z_1)^{1/2} \qquad \text{ as } z \to z_1, \]
see \eqref{xi12-atz1}. It then follows that 
\begin{equation} \label{def-f1} f_1(z) = \left[ \frac{3}{2} \widehat{\varphi}_1(z) \right]^{2/3}, \qquad z \in D(z_1,r)  
\end{equation}
is a conformal map from the disk $D(z_1,r)$ to a neighborhood of the origin, in such a way that $\Sigma_1 \cap D(z_1,r)$
is mapped by $f_1$ into the three half rays with angles $\pm \pi/3$ and $\pi$. We have the
freedom to open the lenses around $\Sigma_1$ in such a way that the lips of the lenses 
in $D(z_1,r)$ are mapped by $f_1$ into half rays as well.
In order to reduce the number of jump contours we open the lenses such that the lips  of two consecutive lenses 
coincide in the disk $D(z_1,r)$ and are mapped by $f_1$ into the half rays with 
the angles  $0$ and $\pm 2\pi/3$. 
Note that 
\[ \varphi_1(z) = \begin{cases}	 \frac{2}{3} f_1(z)^{3/2}  & \text{ in region I}, \\
	- \frac{2}{3} f_1(z)^{3/2} + 2\pi i \beta & \text{ in region II},  \\
	- \frac{2}{3} f_1(z)^{3/2} - 2 \pi i \beta & \text{ in region III},  \end{cases} \]
and we find from \eqref{JS} and \eqref{JSlips}, after some calculations, 
that the jump matrices for $P$ can be written as, 
\begin{multline} \label{JPnearz1}
J_P = J_S = \\
	\begin{cases} 
	\begin{pmatrix} e^{n \varphi_{1,-}} & 0 & 0 \\ 0 & e^{-n \varphi_{1,-}} & 0 \\ 0 & 0 & 1 \end{pmatrix} 
	\begin{pmatrix} 0 & 1 & 0 \\ -1 & 0 & 0 \\ 0 & 0 & 1 \end{pmatrix} 
	\begin{pmatrix} e^{-n \varphi_{1,+}} & 0 & 0 \\ 0 & e^{n \varphi_{1,+}} & 0 \\ 0 & 0 & 1 \end{pmatrix} 
	& \text{ if } \arg f_1(z) = \pi \\
	\begin{pmatrix} e^{n \varphi_{1,-}} & 0 & 0 \\ 0 & e^{-n \varphi_{1,-}} & 0 \\ 0 & 0 & 1 \end{pmatrix} 
	\begin{pmatrix} 0 & \frac{1}{2} & 0 \\ -2 & 0 & 0 \\ 0 & 0 & 1 \end{pmatrix} 
	\begin{pmatrix} e^{-n \varphi_{1,+}} & 0 & 0 \\ 0 & e^{n \varphi_{1,+}} & 0 \\ 0 & 0 & 1 \end{pmatrix}
	& \text{ if } \arg f_1(z) = \pm \pi/3 \\
	\begin{pmatrix} e^{n \varphi_{1}} & 0 & 0 \\ 0 & e^{-n \varphi_{1}} & 0 \\ 0 & 0 & 1 \end{pmatrix}
	\begin{pmatrix} 1 & 0 & 0 \\ -1 & 1 & 0 \\ 0 & 0 & 1 \end{pmatrix} 	
	\begin{pmatrix} e^{-n \varphi_{1}} & 0 & 0 \\ 0 & e^{n \varphi_{1}} & 0 \\ 0 & 0 & 1 \end{pmatrix} 
	& \text{ if } \arg f_1(z) = \pm 2\pi/3  \\
	\begin{pmatrix} e^{n \varphi_{1}} & 0 & 0 \\ 0 & e^{-n \varphi_{1}} & 0 \\ 0 & 0 & 1 \end{pmatrix} 
	\begin{pmatrix} 1 & 0 & 0 \\ 4 & 1 & 0 \\ 0 & 0 & 1 \end{pmatrix}
	\begin{pmatrix} e^{-n \varphi_{1}} & 0 & 0 \\ 0 & e^{n \varphi_{1}} & 0 \\ 0 & 0 & 1 \end{pmatrix} 
	& \text{ if } \arg f_1(z) = 0. 
	\end{cases}
	\end{multline}
Here the contours that are mapped to angles $ \pm 2\pi/3$ and $\pi$ are oriented towards $z_1$ and the
others are oriented away from $z_1$.

We then look for a matrix valued function  $\Psi$ defined and analytic in
an auxiliary $\zeta$ plane except with cuts at the half rays $\arg \zeta = k \pi/3$ for $k=-2, \ldots, 3$, such that
$ \Psi_+ = \Psi_- J_{\Psi} $ with
\begin{equation} J_{\Psi}  =
	\begin{cases} 
	\begin{pmatrix} 0 & 1 & 0 \\ -1 & 0 & 0 \\ 0 & 0 & 1 \end{pmatrix} 
	& \text{ if } \arg \zeta = \pi, \\
	\begin{pmatrix} 0 & \frac{1}{2} & 0 \\ -2 & 0 & 0 \\ 0 & 0 & 1 \end{pmatrix} 
	& \text{ if } \arg \zeta = \pm \pi/3, \\
	\begin{pmatrix} 1 & 0 & 0 \\ -1 & 1 & 0 \\ 0 & 0 & 1 \end{pmatrix} 	
	& \text{ if } \arg \zeta = 2\pi/3,  \\
	\begin{pmatrix} 1 & 0 & 0 \\ 4 & 1 & 0 \\ 0 & 0 & 1 \end{pmatrix}
	& \text{ if } \arg \zeta = 0. 
	\end{cases}
	\end{equation}
	
Then for any analytic prefactor $E_n$ we will have that
\[ E_n(z) \Psi(n^{2/3} f_1(z)) \begin{pmatrix} e^{-n \varphi_1(z)} & 0 & 0 \\ 0 & e^{n \varphi_1(z)} & 0 \\ 0 & 0 & 1 \end{pmatrix} \]
satisfies the required jumps for the local parametrix.

The matrix $\Psi$ is constructed with the Airy functions
\[ y_0(\zeta) = \Ai(\zeta), \quad y_1(\zeta) = \omega \Ai(\omega \zeta), \quad y_2(\zeta) = \omega^2 \Ai(\omega^2 \zeta) \]
in the following way
\[ \Psi(\zeta) = \begin{cases}
	\begin{pmatrix} -2y_2 & \frac{1}{2}y_0 & 0 \\ -2y_2' & \frac{1}{2} y_0' & 0 \\ 0 & 0 & 1 \end{pmatrix} & 
		\text{ for }  0 <  \arg \zeta <  \pi/3, \\
	\begin{pmatrix} -y_0 & -y_2 & 0 \\ -y_0' &  -y_2' & 0 \\ 0 & 0 & 1 \end{pmatrix} & 
		\text{ for }  \pi/3 <  \arg \zeta < 2 \pi/3, \\
	\begin{pmatrix} y_1 & -y_2 & 0 \\ y_1' & -y_2' & 0 \\ 0 & 0 & 1 \end{pmatrix} & 
		\text{ for } 2 \pi/3 <  \arg \zeta < \pi, \\
	\begin{pmatrix} 2y_1 &  \frac{1}{2}y_0 & 0 \\ 2y_1' & \frac{1}{2} y_0' & 0 \\ 0 & 0 & 1 \end{pmatrix} & 
		\text{ for }  - \pi/3  <  \arg \zeta < 0, \\
	\begin{pmatrix} y_0 & -y_1 & 0 \\ y_0' &  -y_1' & 0 \\ 0 & 0 & 1 \end{pmatrix} & 
		\text{ for }  -2 \pi/3 <  \arg \zeta < - \pi/3, \\
	\begin{pmatrix} -y_2 & -y_1 & 0 \\ -y_2' & -y_1' & 0 \\ 0 & 0 & 1 \end{pmatrix} & 
		\text{ for }  - \pi < \arg \zeta < -2 \pi/3.
		\end{cases}
		\]
		
From the known asymptotic behavior of the Airy functions we find the behavior
\begin{multline} 
\Psi(n^{2/3} f_1(z)) \begin{pmatrix} e^{-n \varphi_1(z)} & 0 & 0 \\ 0 & e^{n \varphi_1(z)} & 0 \\ 0 & 0 & 1 \end{pmatrix} 
	= 
	\frac{1}{2\sqrt{\pi}} \begin{pmatrix} n^{-1/6} f_1(z)^{-1/4} & 0 & 0 \\ 0 & n^{1/6} f_1(z)^{1/4} & 0 \\ 0 & 0 & 1 \end{pmatrix} \\
	\times 
	\begin{cases} 
		\begin{pmatrix} 2i & \frac{1}{2} & 0 \\ 2i & -\frac{1}{2} & 0 \\ 0 & 0 & 1 \end{pmatrix} 
			& \text{ if } -\pi/3 < \arg f_1(z) < \pi/3, \\
			\begin{pmatrix} -1 & i & 0 \\ 1 & i & 0 \\ 0 & 0 & 1 \end{pmatrix}
				\begin{pmatrix} e^{-2\pi i n \beta} & 0 & 0 \\ 0 & e^{2\pi i n \beta} & 0 \\ 0 & 0 & 1 \end{pmatrix} 
				& \text{ if } \pi/3 < \arg f_1(z) < \pi,  \\
				\begin{pmatrix} 1 & -i & 0 \\ -1 & -i & 0 \\ 0 & 0 & 1 \end{pmatrix} 
				\begin{pmatrix} e^{2\pi i n \beta} & 0 & 0 \\ 0 & e^{-2\pi i n \beta} & 0 \\ 0 & 0 & 1 \end{pmatrix} 
				& \text{ if } -\pi < \arg f_1(z) < -\pi/3.  
		\end{cases} \\
		\times (I + O(n^{-1})) \qquad \text{ as $n \to \infty$}
		\end{multline}
uniformly for $z$ on the circle $|z-z_1| = r$.
		
The matching with $M_n$ is then provided if we define
\begin{multline} E_n(z) = \sqrt{\pi} M_n(z)  
	\times \begin{cases} 
		\begin{pmatrix} - \frac{1}{2} i & - \frac{1}{2} i & 0 \\ 2 & -2 & 0 \\ 0 & 0 & 1 \end{pmatrix} 
			& \text{ if } -\pi/3 < \arg f_1(z) < \pi/3 \\
				\begin{pmatrix} e^{2\pi i n \beta} & 0 & 0 \\ 0 & e^{-2\pi i n \beta} & 0 \\ 0 & 0 & 1 \end{pmatrix}
			\begin{pmatrix} -1 & 1 & 0 \\ - i & - i & 0 \\ 0 & 0 & 1 \end{pmatrix} 
			& \text{ if } \pi/3 < \arg f_1(z) < \pi  \\
				\begin{pmatrix} e^{-2\pi i n \beta} & 0 & 0 \\ 0 & e^{2\pi i n \beta} & 0 \\ 0 & 0 & 1 \end{pmatrix}
				\begin{pmatrix} 1 & -1 & 0 \\ i & i & 0 \\ 0 & 0 & 1 \end{pmatrix} 
			& \text{ if } -\pi < \arg f_1(z) < -\pi/3
		\end{cases} 
		\\ \times \begin{pmatrix} n^{1/6} f_1(z)^{1/4} & 0 & 0 \\ 0 & n^{-1/6} f_1(z)^{-1/4} & 0 \\ 0 & 0 & 1 \end{pmatrix}
\end{multline}
It is straightforward to check that $E_n$ is analytic in a full neighborhood of $z_1$.

Then if we define 
\[ P(z) = E_n(z) \Psi(n^{2/3} f_1(z)) \begin{pmatrix} e^{-n \varphi_1(z)} & 0 & 0 \\ 0 & e^{n \varphi_1(z)} & 0 \\ 0 & 0 & 1 \end{pmatrix} \]
we get $ P M_n^{-1} = M_n (I + O(1/n)) M_n^{-1} = I + M_n O(1/n) M_n^{-1}$ as $n \to \infty$.
If $n \to \infty$ with $n \in \mathbb N_{\varepsilon}$ then by Proposition \ref{prop:modelRHP} (b), 
both $M_n$ and $M_n^{-1}$ are uniformly
bounded on the circle $|z-z_1 | = r$. This leads to the matching condition
$PM_n^{-1} = I + O(1/n)$ as $n \in \mathbb N_{\varepsilon} \to \infty$, 
uniformly for $z \in \partial D(z_1, r)$ as required in the RH problem \ref{RHforP}.

\subsection{Proof of Theorem \ref{theorem2}}

In the final transformation we define
\begin{equation} \label{def-R}
	R(z) = \begin{cases} S(z) P(z)^{-1} & \text{ in } D, \\
		S(z) M(z)^{-1} & \text{ elsewhere}.
		\end{cases}
		\end{equation}
Recall that $D$ denotes the union of nine disks $D(\omega^j z_k)$, $j=0,1,2$, $k=1,2,3$.

Then $R$ is defined and analytic outside $\Gamma_S \cup \partial D$ and it has an analytic
extension across $\Sigma_1$, $\Sigma_2$ and across the parts of $\Gamma_S$ that are in $D$. 
Thus $R$ is analytic in $\mathbb C \setminus \Gamma_R$ where $\Gamma_R$ consists of $\partial D$,
and the parts of $\Gamma_S \setminus (\Sigma_1 \cup \Sigma_2)$ that are outside the disks,
and we have the following RH problem.
\begin{rhproblem} \label{RHforR}
\begin{itemize} 
\item $R$ is analytic in $\mathbb C \setminus \Gamma_R$.
\item $R_+ = R_- J_R$ on $\Gamma_R$ where 
\[  J_R(z) = \begin{cases} M(z)^{-1} P(z) & \text{ for $z$ on the circles,} \\
	M(z)^{-1} J_S(z) M(z) & \text{ elsewhere on $\Gamma_R$.}
	\end{cases} \]
\item $R(z) = I + O(z^{-1})$ as $z \to \infty$.
\end{itemize} 
\end{rhproblem}

As a result of the matching condition in the RH problem \ref{RHforP} for $P$ we have
\[ J_R(z) = I + O(1/n) \qquad \text{ as } n \to \infty, n \in \mathbb N_{\varepsilon}. \]
On the remaining parts of $\Gamma_R$ the jumps for $R$ are exponentially close to the identity matrix:
\[ J_R(z) = I + O(e^{-cn}) \qquad \text{ elsewhere on } \Gamma_R \]
for some $c > 0$.
This follows from the formulas \eqref{JS} and \eqref{JSlips} for the jump matrices $J_S$
which are of the form $I + O(e^{-cn})$ and we also use Proposition \ref{prop:modelRHP} (b),
which says that $M(z)$ and $M(z)^{-1}$ are uniformly bounded on $\Gamma_R$. For $z \to \infty$,
the estimate can be sharpened to
\[ J_R(z) = I + O(e^{-cn |z|^3}) \qquad \text{ elsewhere on } \Gamma_R. \]

Then by standard estimates on RH problems \cite{Deift}, we have
\begin{equation} \label{R-estimate} 
	R(z) = I + O\left( \frac{1}{n (1+|z|)} \right)   \qquad \text{ as } n \to \infty, n \in \mathbb N_{\varepsilon} 
	\end{equation}
uniformly for $z \in \mathbb C \setminus \Gamma_R$. The estimate \eqref{R-estimate} is the
final result of the steepest descent analysis of the RH problem.

\medskip

\begin{proof}[Proof of Theorem \ref{theorem2}]

We can now prove Theorem \ref{theorem2} by following the steps $Y \mapsto X \mapsto V \mapsto U \mapsto T \mapsto S \mapsto R$
to see the effect on the polynomial $P_{n,n}$. This is similar to the proof of Lemma 6.1 in \cite{BK}.
As in that proof we find
\[ P_{n,n}(z) = S_{11}(z) e^{n g_1(z)}, \qquad z \in \mathbb C \setminus L_1 \] 
where $L_1$ denotes the lens around $\Sigma_1$.
Also $S = RM$ with $R$ satisfying \eqref{R-estimate} gives us
\[ S(z) = M_{n,11}(z) + O(1/n), \qquad z \in \mathbb C \setminus L_1, \,  n \in \mathbb N_{\varepsilon}. \]
where we also used the fact that $M_{n,11}(z)$ remains bounded, see Proposition \ref{prop:modelRHP}.
This proves \eqref{Pnn-asymp} with a $O(1/n)$ term that is uniform for $z \in \mathbb C \setminus L_1$.
Since we have the freedom to open the lens as small as we wish, we find \eqref{Pnn-asymp}
with $M_{n,11}$ given by \eqref{def-Mn11} according to part (c) of Proposition \ref{prop:modelRHP}.

The $O(1/n)$ is also uniform for $t$ in compact subsets of $t_*, t_{**}$, for
values of $n \in \mathbb N_{\varepsilon}$, where $\mathbb N_{\varepsilon}$, see \eqref{Nepsilon}, 
is varying with $t$.

This completes the proof of Theorem \ref{theorem2}, pending the proof of Proposition \ref{prop:modelRHP}
that will follow in the next section.
\end{proof}

\section{Proof of Proposition \ref{prop:modelRHP}} \label{sec:modelRHP}

\subsection{Riemann surface $\mathcal S$}

To solve the model RH problem we use the Riemann surface $\mathcal R$ as before
and with the sheet structure  shown in Figure \ref{fig:three-sheets}. It has genus three.

The Riemann surface has three fold symmetry 
\begin{equation} \label{rho-def} 
	\rho : \mathcal R \to \mathcal R : (z, \xi) \mapsto (\omega z, \omega^2 \xi) 
	\end{equation}
which induces an action of $\mathbb Z_3$ on $\mathcal R$. It will be useful
to consider the orbit space which is a Riemann surface that we call $\mathcal S$,
and the quotient map 
\begin{equation} \label{psi-def} 
	\psi : \mathcal R \to \mathcal S : (z, \xi) \mapsto (z^3, z \xi). 
	\end{equation}
Since $\mathcal R$ is defined by the equation \eqref{def:speccurve}, we find that
$\mathcal S$ has the equation (where $w= z^3$ and $\eta = z \xi$),
\begin{equation} \label{equation-S} 
	\mathcal S : \qquad  \eta^3 - w \eta^2 - (1+t) w\eta + w^2 + Aw = 0. 
	\end{equation}
Then $\mathcal S$ is a genus one Riemann surface, whose sheet structure is shown 
in Figure \ref{fig:surfaceS} with $w_j = z_j^3$ for $j=1,2,3$.  

The branch points are connected by cuts that are the images of the cuts $\gamma_{1,2}$ and $\gamma_{1,3}$
under the mapping $w = z^3$. We denote these by 
\begin{equation} \label{cutsw1w2w3} 
	[w_1, w_2] := \gamma_{1,2}^3, \qquad [w_1,w_3] := \gamma_{1,3}^3 
	\end{equation}
but we emphasize that these are not exact straight line segments.

\begin{figure}[!t]
\centering
\subfigure[$\mathcal S_1$]{
\begin{tikzpicture}[scale=2]
\draw[very thick] (0,0) node[left]{$0$};
\draw[very thick] (0,0)--(0.4,0) node[right]{$w_1$};
\draw[very thick] (0.4,0)--(0.6,0.3) node[right]{$w_2$};
\draw[very thick] (0.4,0)--(0.6,-0.3) node[right]{$w_3$};
\draw  (-1.0,-0.5) rectangle (1.5,0.5);
\end{tikzpicture} 
}  
\subfigure[$\mathcal S_2$]{
\begin{tikzpicture}[scale=2]
\draw[very thick] (0,0) node[above]{$0$};
\draw[very thick] (0,0)--(0.4,0) node[right]{$w_1$};
\draw[very thick] (0.4,0)--(0.6,0.3) node[right]{$w_2$};
\draw[very thick] (0.4,0)--(0.6,-0.3) node[right]{$w_3$};
\draw[dashed](0,0)--(-1.0,0);
\draw  (-1.0,-0.5) rectangle (1.5,0.5);
\end{tikzpicture}
} 

\subfigure[$\mathcal S_3$]{
\begin{tikzpicture}[scale=2]
\draw[very thick] (0,0) node[right]{$0$};
\draw[dashed](0,0)--(-1.0,0);
\draw  (-1.0,-0.5) rectangle (1.5,0.5);
\end{tikzpicture}
} 

\caption{The three sheets $\mathcal S_1$, $\mathcal S_2$ and $\mathcal S_3$ of 
the Riemann surface $\mathcal S$ \label{fig:surfaceS}}
\end{figure}
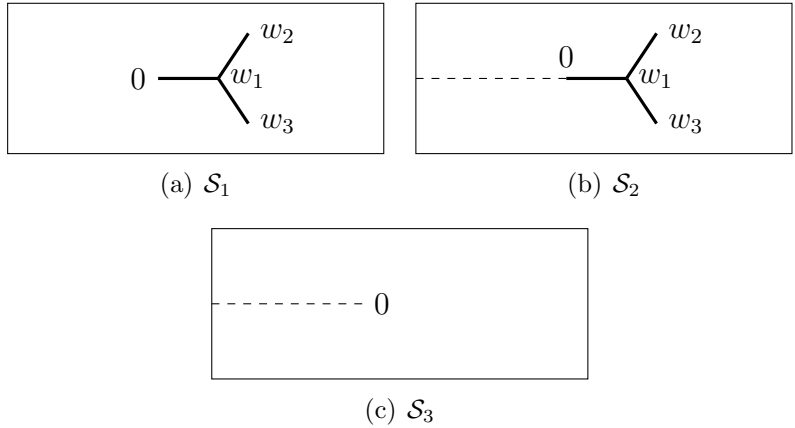

There are three solutions to \eqref{equation-S} with respective asymptotic behaviors
\begin{equation} \label{eta-asymp}
\begin{aligned} 
	\eta_1(w) & = w + t + O(w^{-1}) \\
	\eta_2(w) & = w^{1/2} - \frac{t}{2} + O(w^{-1/2}) \\
	\eta_3(w) & = -w^{1/2} - \frac{t}{2} + O(w^{-1/2}) 
	\end{aligned}
	\end{equation}
as $w \to \infty$. The solution $\eta_j$ is defined and analytic on sheet $\mathcal S_j$ for $j=1,2,3$. Note that
$\eta_j(z^3) = z \xi_j(z)$.

The Riemann surface $\mathcal S$ has an anti-holomorphic involution 
$ \sigma: \mathcal S \to \mathcal S : (w, \eta) \mapsto (\overline{w}, \overline{\eta})$.
The real part of $\mathcal S$ is
\begin{equation}\label{Sreal-def} 
	{\mathcal S}_{real} =  \{ Q \in \mathcal S \mid \sigma(Q) = Q \}. 
	\end{equation}
which is a closed loop that contains the two points at infinity. 
We provide it with an orientation from left to right on the intervals $[w_1, \infty)$ and $(-\infty,0]$
on the first sheet, and the interval $[0, \infty)$ on the third sheet, and from right to left
on the interval $(\infty,w_1]$ on the second sheet. 
This cycle is called $c_S$.  Also
\[ c_S = \psi(c_R) \]
where $c_R$ is the cycle on $\mathcal R$ that was introduced in section \ref{sec:orthopoly}.

There is a cycle $b_S$ that starts and ends at $w_1$ and goes around  $[w_1,w_2]$ on the first
sheet in counterclockwise fashion. 
The cycle $\overline{b}_S$ goes around $[w_1,w_3]$ on the first sheet with clockwise orientation. 
Then ($\overline{b}_S$, $b_S$) is a canonical homology basis for $\mathcal S$ and 
\begin{equation} \label{cS-def} 
	c_S = b_S + \overline{b}_S
	\end{equation}
	with equality in the sense of homotopic cycles.
Also
\begin{equation} \label{aS-def} 
	a_S = b_S - \overline{b}_S 
	\end{equation}
	is a cycle that goes around both $[w_1,w_2]$ and $ [w_1,w_3]$ on the first sheet with counterclockwise orientation.

Since the genus of $\mathcal S$ is one, there is a unique holomorphic differential $\omega_S$ such that
\begin{equation} \label{cS-period} \oint_{c_S} \omega_S = 1. \end{equation}
Using the equation \eqref{equation-S} for $\mathcal S$ we can find
\begin{equation} \label{omegaS-def} 
	\omega_S = \frac{C \, dw}{3\eta^2-2w\eta - (1+t)w}
	\end{equation}
for some constant $C > 0$. The denominator $3 \eta^2 - 2w \eta - (1+t)w$ 
has simple zeros at the branch points $w_1, w_2, w_3$ and a double zero at $0$.
It is also real on ${\mathcal S}_{real}$ and in fact positive on the real parts 
of the first and third sheets, and negative on the real part of the second sheet,
as can be verified from \eqref{eta-asymp} and the fact that there are no other zeros on ${\mathcal S}_{real}$ besides
$w_1$ and $0$.
A local analysis shows that \eqref{omegaS-def}
is indeed holomorphic at the branch points, as well as at the points at infinity.

Note that the holomorphic differential $\omega_R$ from \eqref{omegaR} is
the pullback of $\omega_S$ onto $\mathcal R$:
\[ \psi^*(\omega_S) = \omega_R. \]

Because of symmetry, the $a_S$ period of $\omega_S$ is purely imaginary, 
with positive imaginary part (due to the chosen orientation). We introduce 
\begin{equation} \label{tau-def} 
	\tau := \oint_{a_S} \omega_S  \in i \mathbb R^+ 
\end{equation}
which is the same number as given by \eqref{def-tau}.
Then by \eqref{cS-def}--\eqref{aS-def} and \eqref{cS-period}--\eqref{tau-def}
\begin{align} \label{bS-periods} 
	\oint_{b_S} \omega_S = \frac{1}{2} + \frac{1}{2} \tau, \qquad
 \oint_{\overline{b}_S} \omega_S = \frac{1}{2} - \frac{1}{2} \tau.
\end{align}
The lattice $L$ of periods is thus generated by $\frac{1}{2} - \frac{1}{2} \tau $ and $\frac{1}{2} + \frac{1}{2} \tau$.
\begin{equation} \label{latticeL} 
	L = \{ m(\tfrac{1}{2} - \tfrac{1}{2} \tau) + n (\tfrac{1}{2} + \tfrac{1}{2} \tau) \mid m,n \in \mathbb Z \}. 
	\end{equation}

The Abel map with base point $\infty_1$ (the point at infinity on $\mathcal S_1$) is
\begin{equation} \label{Abelmap} 
	u : \mathcal S \to \mathbb C \slash L : Q \in \mathcal S \mapsto \int_{\infty_1}^{Q} \omega_S  \quad \text{modulo periods} 
	\end{equation}
and it identifies $\mathcal S$ with the complex torus $\mathbb C \slash L$.
A fundamental domain for $\mathbb C \slash L$ is the parallelogram with vertices 
$0, \frac{1}{2} - \frac{1}{2} \tau, 1, \frac{1}{2} + \frac{1}{2} \tau$. The Abel map is real-valued on $c_S$.

To make the Abel map single valued we take away the cuts $[w_1, w_2]$ and $[w_1, w_3]$ (that is, the
$a_S$ cycle), and the restrictions of the Abel map are then denoted by $u_1$, $u_2$, $u_3$.
Thus
\begin{equation} \label{uj-def} 
	u_j(w) = u (w^{(j)}) = \int_{\infty_1}^{w^{(j)}} \omega_S, \qquad j = 1,2,3 
	\end{equation} 
where $w^{(j)}$ is the point on the $j$th sheet that projects onto $w \in \mathbb C$
and the path of integration does not intersect $[w_1,w_2]$ and $[w_1,w_3]$ on the sheets $\mathcal S_1$
and $\mathcal S_2$. Then $u_1$ is defined and analytic  on $\overline{\mathbb C} \setminus 
([0,w_1] \cup [w_1,w_2] \cup [w_1,w_3])$ with $u_1(\infty) = 0$, $u_2$ is analytic on $\mathbb C \setminus 
((-\infty, w_1 \cup [w_1,w_2] \cup [w_1,w_3])$, $u_3$ is 
analytic on $\mathbb C \setminus (-\infty, 0]$ with the following properties
\begin{align} \label{ujump0w1}
	u_{1, \pm}  & = u_{2, \mp} && \text{on } [0, w_1] \\ 
	\label{ujumpw1w2}
	u_{1,\pm}(w) - u_{2,\mp}(w) & = - \oint_{\overline{b}_S} \omega_S 
	 = - \tfrac{1}{2}(1 -\tau) &&  \text{on } [w_1, w_2] \\
	\label{ujumpw1w3}
	u_{1,\pm}(w) - u_{2,\mp}(w) & = - \oint_{b_S} \omega_S 
		= - \tfrac{1}{2}(1 + \tau) && \text{on } [w_1, w_3],
	\end{align}
see \eqref{bS-periods}, and
\begin{align} \label{ujumpinfty0}
	u_{2, \pm}  & = u_{3, \mp} \qquad \text{on } (-\infty,0]. 
	\end{align}
We also note that $u_1 + u_2 + u_3$ is constant on $\mathbb C$, since this sum
is bounded and analytic outside of $(-\infty,w_1]$, $[w_1, w_2]$ and $[w_1,w_3]$,
and has no jumps on any of these arcs as follows from \eqref{ujump0w1}--\eqref{ujumpinfty0}.

\subsection{First step}
The goal of the first step is to remove the prefactors $2$ and $1/2$ in the offdiagonal
entries of the jump matrices on the arcs $\bigcup_j \omega^j \gamma_{1,2}$ and $\bigcup_j \omega^j \gamma_{1,3}$, 
see \eqref{JM}. 
Here we use the components $u_1, u_2, u_3$ of the Abel map of $\mathcal S$ evaluated
in $z^3$. We seek $M$ in the form
\begin{equation} \label{MinN} 
	M(z) = \begin{pmatrix} 1 & 0 & 0 \\ 0 & 2^{2 u_2(\infty)} & 0 
	\\ 0 & 0 & 2^{2 u_3(\infty)} \end{pmatrix} N(z) \begin{pmatrix} 2^{-2 u_1(z^3)} & 0 & 0 \\ 0 & 2^{-2 u_2(z^3)} & 0 \\
	 0 & 0 & 2^{-2 u_3(z^3)} \end{pmatrix}. 
	\end{equation}
Note that $u_2(\infty) = u_3(\infty)$.

In order that $M$ satisfies the RH problem for $M$ we need that
$N$ satisfies the following RH problem. The jumps in \eqref{JN} are obtained from \eqref{JM}, 
\eqref{ujump0w1}--\eqref{ujumpinfty0}, and the definition \eqref{MinN} of $N$ in terms of $M$.

\begin{rhproblem} \label{RHforN}
\begin{itemize}
\item $N$ is analytic in $\mathbb C \setminus (\Sigma_1 \cup \Sigma_2)$.
\item $N_+ = N_- J_N$ on $\Sigma_1 \cup \Sigma_2 $ with
\begin{equation} \label{JN}
	J_N = \begin{cases} \begin{pmatrix} 0 & 1 & 0 \\ -1 & 0 & 0 \\ 0 & 0 & 1 \end{pmatrix} & \quad \text{on } 
		\Sigma_1^o  \\
		\begin{pmatrix} 0 & 2^{-\tau} e^{2\pi in \beta} & 0 \\ -2^{\tau} e^{-2\pi in\beta} & 0 & 0 \\ 0 & 0 & 1 \end{pmatrix} 
	& \quad \text{on } \bigcup_{j=0}^2 \omega^j \gamma_{1,2} \\
	\begin{pmatrix} 0 & 2^{\tau} e^{-2\pi in \beta} & 0 \\ -2^{-\tau} e^{2\pi in\beta} & 0 & 0 \\ 0 & 0 & 1 \end{pmatrix} 
	& \quad \text{on } \bigcup_{j=0}^2 \omega^j \gamma_{1,3} \\
	\begin{pmatrix} 1 & 0 & 0 \\ 0 & 0 & 1 \\ 0 & -1 & 0 \end{pmatrix} 
	& \quad \text{on } \Sigma_2.
 	\end{cases} \end{equation}
\item $N(z)  = (I + O(1/z)) A(z)$ as $z \to \infty$ where $A(z)$ is given by \eqref{defAz}.
\item For $j,k=1, 2, 3$ we have
	\begin{equation} \label{Natbranch} 
		N(z) =  O\left((z- \omega^j z_k)^{-1/4}\right) \qquad \text{ as } z \to \omega^j z_k. 
		\end{equation}
\end{itemize}
\end{rhproblem}
To see that the asymptotic condition in the RH problem for $N$ is indeed the same as the one for $M$
requires some calculations and uses the facts that  
$u_1(z^3) = O(z^{-3})$, $u_2(z^3) = u_2(\infty) + O(z^{-3/2})$, and
$u_3(z^3) = u_3(\infty) + O(z^{-3/2})$ as $z \to \infty$ with $u_2(\infty) = u_3(\infty)$.

The effect of the first step is that the RH problem for $N$ now depends
on the real parameter (recall that $\tau$ is purely imaginary),
\begin{equation} \label{nu-def} 
	\nu = \nu_n = n \beta - \frac{\tau}{2\pi i} \log 2 
	\end{equation} and so
\begin{equation} \label{JN2} 
	J_N = \begin{cases} 
		\begin{pmatrix} 0 & e^{2\pi i \nu} & 0 \\ -e^{-2\pi i \nu} & 0 & 0 \\ 0 & 0 & 1 \end{pmatrix} 
	& \quad \text{on } \bigcup_{j=0}^2 \omega^j \gamma_{1,2} \\
	\begin{pmatrix} 0 & e^{-2\pi i \nu} & 0 \\ - e^{2\pi i \nu} & 0 & 0 \\ 0 & 0 & 1 \end{pmatrix} 
	& \quad \text{on } \bigcup_{j=0}^2 \omega^j \gamma_{1,3}.
	\end{cases} \end{equation}
We consider $\nu \in \mathbb R \slash \mathbb Z$.

\subsection{Second step}
Our next task is to construct functions $v_1, v_2, v_3$ that are defined
and holomorphic on the respective sheets $\mathcal S_1, \mathcal S_2, \mathcal S_3$ of $\mathcal S$
such that
\begin{align} \label{v1v2v3jump1} 
	\begin{pmatrix} v_1 & v_2 & v_3 \end{pmatrix}_+ & = 
	\begin{pmatrix} v_1 & v_2 & v_3 \end{pmatrix}_- \begin{pmatrix} 0 & 1 & 0 \\ -1 & 0 & 0 \\ 0 & 0 & 1 \end{pmatrix}
	\quad \text{on } [0, w_1] \\
	\label{v1v2v3jump2} 
	\begin{pmatrix} v_1 & v_2 & v_3 \end{pmatrix}_+ & = 
	\begin{pmatrix} v_1 & v_2 & v_3 \end{pmatrix}_- \begin{pmatrix} 1 & 0 & 0 \\ 0 & 0 & 1 \\ 0 & -1 & 0 \end{pmatrix}
	\quad \text{on } (-\infty,0]  \\
	\label{v1v2v3jump3} 
	\begin{pmatrix} v_1 & v_2 & v_3 \end{pmatrix}_+ & = 
	\begin{pmatrix} v_1 & v_2 & v_3 \end{pmatrix}_- \begin{pmatrix} 0 & e^{2\pi i \nu} & 0 \\ -e^{-2\pi i \nu} & 0 & 0 
	\\ 0 & 0 & 1 \end{pmatrix} \quad \text{on } [w_1, w_2], \\
	\label{v1v2v3jump4} 
	\begin{pmatrix} v_1 & v_2 & v_3 \end{pmatrix}_+ & = 
	\begin{pmatrix} v_1 & v_2 & v_3 \end{pmatrix}_-
	\begin{pmatrix} 0 & e^{-2\pi i \nu} & 0 \\ - e^{2\pi i \nu} & 0 & 0 \\ 0 & 0 & 1 \end{pmatrix} 
	 \quad \text{on }  [w_1, w_3] \\
	\label{vatbranch} 
	(v_1, v_2, v_3) & = \begin{cases}  O\left((w- w_k)^{-1/4}\right) & \text{ as } w \to w_k \text{ for } k = 1,2,3, \\
		 O(1) & \text{ as } w \to 0, 
		 \end{cases} \\
	\label{vatinfinity}
		v_1 = O(1), & \quad v_2 = O(w^{-1/4}), \quad v_3 = O(w^{-1/4}) \text{ as } w \to \infty.
\end{align}
If we can find such functions $v_j$ then $(v_1(z^3), v_2(z^3), v_3(z^3))$ is a vector
that satisfies $v_+ = v_- J_N$. 

The problem for $v$ clearly depends on $\nu$. 
In the second step we show that it is possible to solve it for the particular value $\nu = 1/2$. 
In this case (and also in the case $\nu = 0$)
all non-zero entries in the jump matrices in \eqref{v1v2v3jump1}--\eqref{v1v2v3jump4} 
are $\pm 1$. This then implies that any solution $(v_1, v_2, v_3)$ of the above vector valued RH problem,
yields a meromorphic function on $\mathcal S$ defined by  
\[ Q = (w, \eta) \in \mathcal S \mapsto v_j^2(w), \qquad \text{if } Q \in \mathcal S_j \text{ for } j= 1,2,3. \] 
Because of \eqref{vatbranch} this function can have simple poles at $w_1, w_2, w_3$, but not at $w=0$,
and because of \eqref{vatinfinity} it has a simple zero at $\infty_2$. Having three poles, the
function must have two more zeros,  
and the only possibility is to have a double
zero somewhere in $\mathcal S_{\real}$. 

We can now solve the problem for $v$ by reversing the arguments. We start by noting that
the function
\begin{equation} \label{eq-F} F : \mathcal S \to \overline{\mathbb C} : \quad (w, \eta) \in \mathcal S \mapsto  
	\frac{\eta^2}{3\eta^2 - 2 w \eta - (1+t)w} \end{equation}
is meromorphic with simple poles at $w=w_1, w_2, w_3$, a simple zero at $\infty_2$ (see also the asymptotics \eqref{eta-asymp}) 
and a double zero at
\[ Q_{1/2} := (-A, 0) \in \mathcal S_1. \]
Indeed, by \eqref{equation-S} we have that $\eta = 0$  implies $w=0$ or $w=-A$. The numerator $\eta^2$
in \eqref{eq-F} thus gives the double zero at $Q_{1/2}$ and an inspection of the behavior of $\eta_1$
reveals that $Q_{1/2} \in \mathcal S_1$. The origin $(w,\eta) = (0,0)$ is not a zero of $F$ since
the denominator of \eqref{eq-F} also vanishes quadatically at the origin and $F(0,0) = \frac{1}{3} > 0$.
Also $F(\infty_1) = 1$.

Let $F_j$ denote the restriction of $F$ to the $j$th sheet.
Since each sheet is simply connected and the zero at $Q_{1/2} \in \mathcal S_1$ is a double zero 
(the other zeros and poles are on the cuts), we can take an analytic square root on each
sheet. We do it in such a way that 
$v_j^2 = F_j$
with
$v_1(\infty) = 1$, $v_1(0) =  - \frac{1}{\sqrt{3}}$,
$v_{2,+}(0) = - v_{2,-}(0) = \frac{1}{\sqrt{3}}$ and $v_3(0) = \frac{1}{\sqrt{3}}$.
The construction is such that \eqref{v1v2v3jump1} and \eqref{v1v2v3jump2} are satisfied.

Then a careful analysis about how the branches of the square root behave under
the analytic continuation (we were assisted by Maple)
shows that the jump matrix on $[w_1,w_2]$ and $[w_1,w_3]$  is 
$\begin{pmatrix} 0 & -1 & 0 \\ 1 & 0 & 0 \\ 0 & 0 & 1 \end{pmatrix}$,
which is the jump matrix in \eqref{v1v2v3jump3} and \eqref{v1v2v3jump4} with $\nu = 1/2$.

Thus we can solve the vector problem for the parameter  $\nu = 1/2$.
The solution is denoted by $(v_{1}^{(1/2)}, v_{2}^{(1/2)}, v_{3}^{(1/2)})$. 
Note that $v_{1}^{(1/2)}$ has a simple zero at the point $-A$ on the negative real line.
We also have $v_1^{(1/2)}(\infty) = 1$
and $v_2^{(1/2)}(w) = O(w^{-1/4})$, $v_3^{(1/2)}(w) = O(w^{-1/4})$, as $w \to \infty$.

The construction gives in particular by \eqref{eq-F}
\begin{equation} \label{v1-def}
	v_1^{(1/2)}(w) = F^{1/2}(w^{(1)}) =  \frac{ \eta_1(w)}{(3\eta_1(w)^2 - 2 w \eta_1(w) - (1+t) w)^{1/2}}.
	\end{equation}

\begin{remark}
There is a constant $c \in (0,1)$ such that
\[ \tilde{F} : \mathcal S \to \overline{\mathbb C} : (w, \eta) \mapsto \frac{\eta^2 - cw}{3\eta^2 - 2 w\eta- (1+t)w} \]
has a double zero at a point $Q_0 \in \mathcal S_2 \cap \mathcal S_{real}$.
It further has simple poles at $w_1, w_2, w_3$ and a simple zero at $\infty_2$.
Then a similar construction, yields a vector
$(v_{1}^{(0)}, v_{2}^{(0)}, v_{3}^{(0)})$ that satisfies the jump conditions
in the vector problem with
$\nu = 0$.  The difference in jumps comes from the fact that $(\eta^2 - cw)^{1/2}$
is not a globally defined analytic function on the Riemann surface (despite having only
double zeros and poles). The change in argument of $\eta^2 - cw$ along the $b_S$ and  $\overline{b}_S$ cycle
is an odd multiple of $2\pi$, and then the change in argument of $(\eta^2-cw)^{1/2}$
is by an odd multiple of $\pi$, which leads to a change in sign in the jump matrices.
\end{remark}

\subsection{Third step}

We start from the functions $v_{1}^{(1/2)}$, $v_{2}^{(1/2)}$ and $v_{3}^{(1/2)}$ 
that solve the vector problem with $\nu  = 1/2$.
We are going to use Jacobi theta functions to modify the functions so that they
solve the vector problem with an arbitrary $\nu \in [0,1)$, see also \cite{BL,DKMVZ}
for similar uses of theta functions in RH problems.
Let $\theta(s)$ be the theta function as in \eqref{theta3-def} which has zeros at the values
$s_0 ~(\mod{L})$ with $s_0$ as in \eqref{def-s0} and no other zeros.
Let 
\begin{equation} \label{delta-def} 
	\delta = u(Q_{1/2}) - s_0 = \int_{-\infty_1}^{-A} \omega_S - \frac{-1 + \tau}{4} 
	\end{equation}
Then $Q  \mapsto \theta(u(Q) - \delta)$ has a simple zero at $Q_{1/2}$ (and no other zeros) and it follows that the
functions
\begin{equation} \label{vjgamma} 
	v_{j}^{(\nu)}(w) = \frac{\theta(u_j(w) - \delta  + \nu- 1/2)}{\theta(u_j(w)- \delta)} 
	\, v_{j}^{(1/2)}(w), \qquad w \in \mathcal S_j, \quad j=1,2,3 
	\end{equation}
are well-defined and analytic. The zero of $\theta(u_1(w) - \delta)$ at $w= - A$ is
cancelled by the zero of $v_1^{(1/2)}(w)$.

\begin{lemma} \label{lem:vgamma}
The vector $(v_1^{(\nu)}, v_2^{(\nu)}, v_3^{(\nu)})$ satisfies the
conditions for the vector problem \eqref{v1v2v3jump1}--\eqref{vatinfinity}.

\end{lemma}
\begin{proof} 
The ratio of theta functions
\begin{equation} \label{def-Theta} 
	\Theta(s) := \frac{\theta(s-\delta + \nu- 1/2)}{\theta(s-\delta)} 
	\end{equation}
has periodicity properties
\[ \Theta(s+1) = \Theta(s), \quad \Theta(s \pm \tfrac{1+\tau}{2}) = -e^{\mp 2\pi i \nu} \Theta(s),
	\quad  \Theta(s \pm \tfrac{1-\tau}{2}) = -e^{\pm 2\pi i\nu} \Theta(s), \]
which easily follows from \eqref{theta-periodicity} and \eqref{def-Theta}.	

Let $w \in [w_1, w_2]$. Then $v_{1,+}^{(1/2)}(w) =  v_{2,-}^{(1/2)}(w)$
and $u_{1,+}(w) = u_{2,-}(w) + \frac{-1 + \tau}{2}$ by \eqref{ujumpw1w2} so that
by the periodicity property
\[ \Theta(u_{1,+}(w)) = \Theta(u_{2,-}(w) - \tfrac{1  - \tau}{2}) =
	      -e^{-2\pi i \nu} \Theta(u_{2,-}(w)). \]
Thus
\begin{align*} 
 v_{1,+}^{(\nu)}(w) & =  \Theta(u_{1,+}(w)) v_{1,+}^{(1/2)}(w) \\
	& =  - e^{-2\pi i \nu} \Theta(u_{2,-}(w)) v_{2,-}^{(1/2)}(w) \\
	& = - e^{-2\pi i \nu} v_{2,-}^{(\nu)}(w). 
	\end{align*}
Similarly, $v_{2,+}^{(\nu)}(w) = e^{2\pi i \nu} v_{1,-}^{(\nu)}(w)$, which gives
the jump \eqref{v1v2v3jump3} on $[w_1,w_2]$.
A similar calculation shows that $v_1^{(\nu)}, v_2^{(\nu)}, v_3^{(\nu)}$ satisfies the jump 
\eqref{v1v2v3jump4} on $[w_1, w_3]$.
The jumps on $[0, w_1]$ and $(-\infty,0]$, as well as the asymptotic  conditions \eqref{vatbranch} 
and \eqref{vatinfinity}  are straightforward to verify.
\end{proof}

By Lemma \ref{lem:vgamma} we can solve the vector problem for any real $\nu$.
The solution is not unique, since we can multiply by a common constant.
If $v_1^{(\nu)}(\infty) \neq 0$, then we can normalize the solution  and we find that  
\begin{equation} \label{N1j-def} 
 \begin{aligned}
	N_{1,j}(z) & = \frac{v_j^{(\nu)}(z^3)}{v_1^{(\nu)}(\infty)}   \\
	& = \frac{\theta(- \delta)}{\theta(- \delta  + \nu-1/2)}
	\frac{\theta(u_j(z^3) - \delta  + \nu- 1/2)}{\theta(u_j(z^3)- \delta)}  
		\frac{v_{j}^{(1/2)}(z^3)}{v_1^{(1/2)}(\infty)}
		\end{aligned}
		\end{equation}
	gives a vector $(N_{1,j}, N_{2,j}, N_{3,j})$ that satisfies the conditions
	for the first row in the RH problem \eqref{RHforN}.

This fails if $v_1^{(\nu)}(\infty) = 0$,  which by \eqref{vjgamma} and the fact that $u_1(\infty) = 0$
comes down to $\theta(-\delta + \nu - 1/2)  = 0$. Thus $-\delta + \nu - 1/2 \equiv  s_0 \mod{L}$.
From \eqref{delta-def}  it then follow that there is unique
$\nu = \nu^*$ for which this holds, namely
\[ \nu^* = \frac{1}{2} + \int_{-\infty_1}^{-A} \omega_S \qquad  \mod{\mathbb Z}. \]
Using $\psi^*(\omega_S) = \omega_R$, we can also write
\begin{equation} \label{def-nustar} 
	\nu^* = \frac{1}{2} + \int_{-\infty_1}^{-A^{1/3}} \omega_R \qquad \mod{\mathbb Z}, 
	\end{equation}
where $-A^{1/3}$ denotes the point $(-A^{1/3},0)$ that is on the first sheet $\mathcal R_1$
of the Riemann surface.

In view of \eqref{nu-def} and \eqref{def-nustar} we  conclude that if
\begin{equation} \label{n-nonspecial} 
	n \beta \not\equiv \beta^* \quad \mod{\mathbb Z}, 
\end{equation}
where $\beta^*$ is given by \eqref{def-betastar},
then \eqref{N1j-def} solves the first row in the RH problem for $N$.
It also follows that if, $n \in \mathbb N_{\varepsilon}$, see \eqref{Nepsilon}, then
the distance from $n \beta - \nu^* - \frac{\tau}{2\pi i} \log 2$ to the set
of integers is at least $\varepsilon$, and then the entries \eqref{N1j-def} are
uniformly bounded for $z$ in compact subsets of $\overline{\mathbb C}$ away from
the branch points with a bound that only depends on $\varepsilon$.

\subsection{Fourth step}

Before we can continue with filling in the other rows of $N$, we need a lemma.
Recall that $\rho$ is the symmetry \eqref{rho-def} of the Riemann surface $\mathcal R$.
The real part of $\mathcal R$ is
\[ \mathcal R_{real} = \{ (\xi, z) \in \mathbb R \times \mathbb R \mid P(\xi,z) = 0 \} \]
where $P$ is the algebraic equation \eqref{def:speccurve} for $\mathcal R$.

\begin{lemma} \label{lem:nonspecial}
Let $P \in \mathcal R_{real}$. Then 
the divisor $D = P + \rho(P) + \rho^2(P)$ is non-special.  
\end{lemma}

\begin{proof}
The space $L(D)$ contains all meromorphic functions on $\mathcal R$ with poles only at $P$, $\omega P$, $\omega^2 P$.
These are all simple poles if $P \not\in \{ \infty_1, \infty_2\}$, and at most poles of order three otherwise.
We have to show that $\dim L(D) = 1$, that is, the only functions in $L(D)$ are the constant functions.

The three points $P, \rho(P), \rho^2(P)$ are
all mapped by \eqref{psi-def} to the same point $Q \in \mathcal S$. Note that $\dim L(Q) = 1$, 
since there are no special points
on a genus one Riemann surface. A meromorphic function $\tilde{f}$ on $S$ gives rise to a meromorphic
function $f= \tilde{f} \circ \psi$ on $\mathcal R$ which is invariant under the $\mathbb Z_3$ action, and any $\mathbb Z_3$
invariant meromorphic function can be obtained that way. It follows that $L(D)$ does not contain
any $\mathbb Z_3$ invariant meromorphic functions, except for constants.

Let $f \in L(D)$. Then $f + f \circ \rho + f \circ \rho^2$ is $\mathbb Z_3$ invariant, and therefore
a constant, say
\begin{equation} \label{symmsum} 
	f + f \circ \rho + f \circ \rho^2  = 3c 
	\end{equation}
	for some constant $c \in \mathbb C$. We now distinguish three cases $P \not\in \{\infty_1, \infty_2\}$, $P = \infty_1$ and $P=\infty_2$.
\paragraph{Case $P \not\in \{\infty_1, \infty_2\}$}
In this case we have Laurent expansions of $f$ about $\infty_1$ and $\infty_2$ of the form
\begin{align*} 
	f(z) = \begin{cases} c + c_1 z^{-1} + \cdots &  \text{ as } z \to \infty_1, \\
			 c + d_1 z^{-1} + \cdots & \text{ as } z \to  \infty_2 
			\end{cases} \end{align*}
	where we recall that $z^{-1}$ is  the local coordinate around $\infty_1$ and $z^{-1/2}$ around $\infty_2$.
We do not have fractional exponents in the expansion about $\infty_2$ since that
would be incompatible with \eqref{symmsum}.  Then $z \mapsto z(f(z)-c)$ is holomorphic at both $\infty_1$ and $\infty_2$ and since
we do not introduce any other poles,
\[ f_1 =  \pi_1 (f - c) \in L(D), \qquad \pi_1 : \mathcal R \to \widehat{\mathbb C} : (z,\xi) \mapsto z. \]
We iterate this argument, and inductively find a sequence $(f_n)$ of functions and a sequence $(c_n)$ of constants
such that 
\[ f_{n+1} = \pi_1 (f_n - c_n) \in L(D). \]
Since $L(D)$ is finite dimensional, there is a linear relations between $f, f_1, \ldots, f_n$ for some $n$.
Then $f$ is rational in $z$, which means that together with pole at $P$, it also poles at other
points on the Riemann surface with the same $z$-coordinate. This forces $f$ to be a constant.

\paragraph{Case $P = \infty_1$}
In this case there is a possible pole at $\infty_1$ of order $\leq 3$ and the Laurent expansions have the form 
\begin{align*} 
	f(z) = \begin{cases} 
		c_{-3} z^{3} + c_{-2} z^2 + c_{-1} z + c + O(z{-1}) + \cdots &  \text{ as } z \to \infty_1, \\
	  c + O(z^{-1}) + \cdots &  \text{ as } z \to  \infty_2. 
		\end{cases} \end{align*}
The identity \eqref{symmsum} implies $c_{-3} = 0$, and as in the first case we find
\[ f_1  = \pi_1 (f - c) \in L(D). \]
Then we can argue as above and conclude that $f$ is a constant.

\paragraph{Case $P= \infty_2$}
In this case there are expansions
\begin{align*} 
		f(z) = \begin{cases} 
		c + c_1 z^{-1} +  O(z^{-2}) &  \text{ as } z \to \infty_1, \\
	  c_{-1} z + c + O(z^{-1}) & \text{ as } z \to  \infty_2,
		\end{cases}
		\end{align*}
since again there can be no terms with $z^{3/2}$, $z^{1/2}$ and $z^{-1/2}$ because of \eqref{symmsum}.
If $c_{-1} \neq 0$, then $\pi_1 (f - c)$ has a pole of order $4$ at $\infty_2$, and so it does not belong to $L(D)$,
as in the other cases. However we now use that $f \circ \rho$ also belongs to $L(D)$ and
has expansions in local coordinates
\begin{align*} 
	(f \circ \rho)(z)  =  		\begin{cases} 
	c + c_1 \omega^2 z^{-1} +  O(z^{-2})  & \text{ as } z \to \infty_1,  \\
	c_{-1} \omega z + c + O(z^{-1}) &  \text{ as } z \to  \infty_2.
	\end{cases}
\end{align*}
where we assume $c_{-1} \neq 0$ (otherwise $f$ has no poles at all, and clearly is a constant).
Then
\[ g = \omega^2 f - f \circ \rho - \omega^2 c + c \]
is in $L(D)$ with
\begin{align*} 
	g(z) = \begin{cases}
		O(z^{-2})  &  \text{ as } z \to \infty_1, \\
	 c_{-1} (\omega^2 - \omega) z + O(z^{-1})  & \text{ as } z \to  \infty_2.
	\end{cases}
\end{align*}
Thus $g$ has a double pole at $\infty_2$ and a double zero at $\infty_1$. There are no other zeros
or poles, and so $1/g$ has a double pole at $\infty_1$, which means that
\[ 1/g \in L(3 \infty_1). \]
From the second case we already know that $L(3 \infty_1)$ consists of constant functions only. 
Thus $g$ is a constant, which is a contradiction with $c_{-1} \neq 0$.
\end{proof}

We can now complete solution of $N$ in the same way as in \cite[section 6.6]{BK}.

We assume that $\nu \neq \nu^*$. Then \eqref{N1j-def} gives us the first
row of $N$. 
The ratio \eqref{def-Theta} of shifted theta-functions has a zero at
$s= s_0 + \delta - \nu + 1/2 = u(Q_{1/2}) - \nu + 1/2$.
There is a value $Q_{\nu}$ on the real part of $\mathcal S$ with
\[ u(Q_{\nu}) = u(Q_{1/2}) - \nu + 1/2 \qquad \mod{\mathbb Z}. \]
If $Q_{\nu} = (w, \eta_j(w))$ is on the sheet $\mathcal S_j$,
then it follows from \eqref{vjgamma} that $v_j^{\nu}(w) = 0$,
and hence $N_{1,j}(z) = 0$ whenever $z^3 = w$.
Also $Q_{\nu} \neq \infty_1$, since $\nu \neq \nu^*$.

Then $\psi^{-1}(Q_{\nu}) = \{ P_{\nu}, \rho(P_{\nu}), \rho^2 (P_{\nu})\}$
for some $P_{\nu} \in \mathcal R_{real} \setminus \{\infty_1\}$.
The divisor
\[ D = P_{\nu} + \rho(P_{\nu}) + \rho^2 (P_{\nu}) \]
is non-special by Lemma \ref{lem:nonspecial}. Thus $\dim L(D) = 1$.
It then follows that $L(D + 2 \infty_2)$ is three dimensional (by the Riemann-Roch theorem). 
It has
a basis $\{1, f, g\}$.
Let $f_1, f_2, f_3$, and $g_1, g_2, g_3$ denote the restrictions of $f$
and $g$ to the respective sheets of $\mathcal S$, and put
\begin{equation} \label{defB} 
	B := \begin{pmatrix} N_{1,1} & N_{1,2} & N_{1,3} \\
	f_1 N_{1,1} & f_2 N_{1,2} & f_{3} N_{1,3} \\
   g_1 N_{1,1} & g_2 N_{1,2} & g_3 N_{1,3} \end{pmatrix}. 
	\end{equation}
If $P_{\nu} = (z, \xi_j(z)) \in \mathcal R_j$ then $f_j$ and $g_j$
have a possible pole at $z$, $\omega z$, and $\omega z^2$. However, the poles 
are compensated by the zero of $N_{1,j}$ and it follows that $B$ is analytic in $\mathbb C \setminus (\Sigma_1 \cup \Sigma_2)$. 

It is then easy to verify that $B$ satisfies the jumps $B_+ = B_- J_N$
as in the RH problem \ref{RHforN} for $N$.
Since $A_+ = A_- J_N$ on $\Sigma_2$, we find that $B A^{-1}$ is analytic across $\Sigma_2$
and therefore it is single valued at infinity.
It can be verified that $B(z) = O(z^{1/4})$ and $A(z) = O(z^{1/4})$ as $z \to \infty$
which means that the Laurent expansion of $B A^{-1}$ has the form
\begin{equation} \label{BAinv} 
	(BA^{-1})(z) = C + O(z^{-1}) \qquad \text{ as } z \to \infty 
	\end{equation}
with a constant matrix $C$. From \eqref{defAz} we see that $\det A \equiv 1$. From $B_+ = B_-J_N$
where $\det J_N \equiv 1$ it also follows that $\det B \equiv b$ for some constant $b$.
If the constant were zero, then we see from \eqref{defB} that there is a nontrivial
linearly combination $h= c_0 + c_1 f + c_2 g$ such that $h_j N_{1,j} \equiv 0 $ for each $j$.
The functions $N_{1,j}$ are analytic and they do not vanish identically, which implies
that $h=0$ and this is impossible since $\{1, f, g \}$ are linearly independent.  

Then by \eqref{BAinv} we have
\[ \det C = \lim_{z \to \infty} \det (BA^{-1}(z)) = b \neq 0 \]
and so $C$ is an invertible matrix.

Then $N = C^{-1} B$ satisfies the jump conditions \eqref{JN} in the RH problem for $N$.
The asymptotic condition is satisfied because by \eqref{BAinv}
\[ N = C^{-1} B = C^{-1} \left( C  + O(z^{-1})\right) A(z) = \left( I + O(z^{-1}) \right) A(z) \]
as $z \to \infty$. Also \eqref{Natbranch} is satisfied, since it holds for the first row
and the functions $f$ and $g$ are analytic at the branch points, except in the case where
$Q_{\nu}$ coincides with the branch point $w_1$. In that case, $f$ and $g$ may have a
pole at $z_1$, $\omega z_1$, $\omega^2 z_1$. However, in that case $N_{1,1}$ and $N_{1,2}$ 
behave like $(z- \omega^j z_1)^{1/4}$ as $z \to \omega^j z_1$ for $j=0,1,2$,
and then \eqref{Natbranch} still holds.

\subsection{Proof of Proposition \ref{prop:modelRHP}}

\begin{proof}

The RH problem for $N$ is solvable if and only if
$n \beta - \nu^* - \frac{\tau}{2\pi i} \log 2$ is not an integer.
Then  \eqref{MinN} gives the solution of the RH problem \ref{RHforM} for $M$
and this proves part (a) of Proposition \ref{prop:modelRHP}.

\medskip

Part (b) follows from the similar statement about $N$. 
Alternatively, it can also be deduced by a compactness argument.

\medskip

For part (c) it remains to show that $M_{n,11}$ is given by \eqref{def-Mn11}.
First from \eqref{MinN} and \eqref{uj-def} we get
\[ M_{n,11}(z) = 2^{-2u_1(z^3)} N_{n,11}(z) = 2^{2 \int_{z^3}^{\infty_1} \omega_S} N_{n,11}(z) \]
Since $\psi^*(\omega_S) = \omega_R$ this is also
\[ M_{n,11}(z) = 2^{2 \int_{z}^{\infty_1} \omega_R} N_{n,11}(z). \]
For $N_{n,11}(z)$ we have the expression \eqref{N1j-def} with $j=1$,
$\delta$ as in \eqref{delta-def} and $\nu$ as in \eqref{nu-def}.
Since $\psi^*(\omega_S)$ we also have
\[ \delta = -s_0 - \int_{-A^{1/3}}^{\infty_1} \omega_R, \]
and (see \eqref{uj-def} for $u_1$),
\[ u_1(z^3) = \int_{\infty_1}^{z} \omega_R, \]
with integration on the first sheet of $\mathcal R$.

Finally, by \eqref{eq-F} and \eqref{v1-def}, we have $v_1^{(1/2)}(\infty) = 1$
and 
\begin{align*} v_1^{(1/2)}(z^3) & = 
	\frac{ \eta_1(z^3)}{(3\eta_1(z^3)^2 - 2 z^3 \eta_1(z^3) - (1+t) z^3)^{1/2}} \\
	& =  \frac{\xi_1(z)}{(3\xi_1(z)^2 - 2 z^2 \xi_1(z) - (1+t) z)^{1/2}}
 \end{align*}
since $\eta_1(z^3) = z \xi_1(z)$. Combining all this we find \eqref{def-Mn11}.
\end{proof}

\section*{Acknowledgements}
We thank Guilherme Silva for useful comments and for assistance in producing the Figures~5 and 9.

The first author is supported by KU Leuven Research Grant OT/12/073,
the Belgian Interuniversity Attraction Pole P07/18, and FWO Flanders
projects G.0641.11 and G.0934.13.

\end{document}